\documentclass[nosumlimits,intlimits]{lmcs}
\pdfoutput=1

\usepackage{lastpage}
\renewcommand{\dOi}{14(3:10)2018}
\lmcsheading{}{1--\pageref{LastPage}}{}{}%
{Jan.~12, 2016}{Aug.~27, 2018}{}

\pdfoutput=1
\usepackage{hyperref}

\newcommand\noqed{\def\qed{}}

\usepackage{etex}

\usepackage[normalem]{ulem}			%
\usepackage{proof}					%
\usepackage{amssymb}				%
\usepackage{stmaryrd,latexsym}		%
\usepackage{wasysym}				%
\usepackage{ushort}					%
\usepackage{slashed}				%
\usepackage{yhmath}					%
\usepackage{graphicx}
\usepackage{bm}						%
\usepackage{optparams}				%
\usepackage{xspace}
\usepackage{upgreek}
\usepackage{lscape}
\usepackage{bussproofs}

\usepackage{mathrsfs}
\usepackage{thmtools}

\theoremstyle{definition}
\newtheorem{rexa}[thm]{Example}
\newtheorem{assn}[thm]{Assumption}

  {\gdef\scalefactor{#1}\begin{center}\proofSkipAmount \leavevmode}%
  {\scalebox{\scalefactor}{\DisplayProof}\proofSkipAmount \end{center} }

\usepackage{enumitem}

\setdescription{font=\it}

\usepackage[all,cmtip]{xy}			%
\xyoption{rotate}

\usepackage{ifdraft}  				%
\usepackage{todos}

\usepackage[layout=footnote,final]{fixme}

\usepackage{enumitem}

\FXRegisterAuthor{sg}{asg}{SG}	%
\FXRegisterAuthor{cr}{acr}{CR}	%
\FXRegisterAuthor{ls}{als}{LS}	%

\usepackage[utf8]{inputenc}

\usepackage{hypcap}
\DeclareFontFamily{OT1}{pzc}{}
\DeclareFontShape{OT1}{pzc}{m}{it}{<-> s * [1.100] pzcmi7t}{}
\DeclareMathAlphabet{\mathscr}{OT1}{pzc}{m}{it}

\newcommand{\copl}{\hm{[\kern-2.6pt[}}
\newcommand{\copr}{\hm{]\kern-2.6pt]}}

\newcommand{\pheq}{\phantom{{}={}}} %

\newcommand{\ev}{\operatorname{\sf ev}}
\newcommand{\ext}{\operatorname{\sf ext}}
\newcommand{\out}{\operatorname{\sf out}}
\newcommand{\tuo}{\operatorname{\sf out}^{\mone}}
\newcommand{\dist}{\operatorname{\sf dist}}
\newcommand{\assoc}{\operatorname{\sf assoc}}
\newcommand{\coit}{\operatorname{\sf coit}}
\newcommand{\curry}{\operatorname{\sf curry}}
\newcommand{\uncurry}{\operatorname{\sf uncurry}}

\newcommand{\istar}{\dagger}  					%
\newcommand{\iistar}{\ddagger}  				%

\newcommand{\IF}{\operatorname{\sf if}}
\newcommand{\ifTerm}[3]{\IF #1\kern2.2pt {\sf then}\kern1.2pt #2\kern2.2pt {\sf else}\kern2.2pt #3}
\newcommand{\ifTermO}[3]{\IF_{\nu} #1\kern2.2pt {\sf then}\kern2.2pt #2\kern2.2pt {\sf else}\kern2.2pt #3}

\newcommand{\WHILE}{\operatorname{\sf while}}
\newcommand{\whileTerm}[3]{\LET #1\kern1.2pt\WHILE\kern1.2pt #2 \kern2.2pt{\sf do}\kern2.2pt #3}
\newcommand{\whileTermI}[4]{\LET #1\kern1.2pt\WHILE^{#2}\kern1.2pt #3 \kern2.2pt{\sf do}\kern2.2pt #4}
\newcommand{\whileTermS}[2]{\WHILE\kern1.2pt #1 \kern2.2pt{\sf do}\kern2.2pt #2}

\newcommand{\SEQ}{\operatorname{\sf seq}}
\newcommand{\seqTerm}[2]{\SEQ_{\nu}\kern1.2pt#1; #2}

\newcommand{\letin}[2]{\operatorname{\sf let} #1\kern2.2pt \operatorname{\sf in}\kern1.2pt #2}
\newcommand{\match}[2]{\operatorname{\sf let}\kern2.2pt #1\kern2.2pt \operatorname{\sf in}\kern1.2pt #2}
\newcommand{\letmon}[3]{\operatorname{\sf let}\kern1.2pt  #1\kern2.2pt \operatorname{\sf by} #2\kern2.2pt \operatorname{\sf in}\kern1.2pt #3}

\newcommand{\klstar}{\star}  					%
\newcommand{\kklstar}{\text{\kreuz}} %

\newcommand{\wrt}{\mathit{write}}
\newcommand{\rd}{\mathit{read}}

\newcommand{\skp}{\mathit{skip}}
\newcommand{\Pow}{\mathcal{P}}

\newcommand{\Nat}{\mathbb{N}}
\newcommand{\Nom}{\catname{Nom}}
\newcommand{\Names}{\mathbb{A}}

\usepackage{stmaryrd}

\DeclareMathOperator{\lsmOp}{\lfloor\kern0.6pt}

\DeclareMathOperator{\rsmOp}{\rfloor\kern0.6pt}

\newcommand{\SigF}{\Sigma}
\newcommand{\TF}{T_\SigF}
\newcommand{\BBTF}{\BBT_\SigF}
\newcommand{\bang}{\mathop{!}\xspace}

\DeclareMathAlphabet{\mathscr}{OT1}{pzc}{m}{it}

\newcommand{\guard}{\diamondsuit}

\newcommand{\takeout}[1]{}

\newcommand\numberthis{\addtocounter{equation}{1}\tag{\theequation}}

\renewcommand{\by}[1]{\text{/\!/~#1}}			%
\newcommand{\comp}{\operatorname{\kern-2pt}} %

\providecommand{\catname}{\mathbf} 
\providecommand{\clsname}{\mathcal}
\providecommand{\oname}[1]{\operatorname{\mathsf{#1}}}

\def\defcatname#1{\expandafter\def\csname B#1\endcsname{\catname{#1}}}
\def\defcatnames#1{\ifx#1\defcatnames\else\defcatname#1\expandafter\defcatnames\fi}
\defcatnames ABCDEFGHIJKLMNOPQRSTUVWXYZ\defcatnames

\def\defclsname#1{\expandafter\def\csname C#1\endcsname{\clsname{#1}}}
\def\defclsnames#1{\ifx#1\defclsnames\else\defclsname#1\expandafter\defclsnames\fi}
\defclsnames ABCDEFGHIJKLMNOPQRSTUVWXYZ\defclsnames

\def\defbbname#1{\expandafter\def\csname BB#1\endcsname{\mathbb{#1}}}
\def\defbbnames#1{\ifx#1\defbbnames\else\defbbname#1\expandafter\defbbnames\fi}
\defbbnames ABCDEFGHIJKLMNOPQRSTUVWXYZ\defbbnames

\def\Set{\catname{Set}}
\def\Cpo{\catname{Cpo}}
\def\Cppo{\catname{Cppo}}

\providecommand{\argument}{\operatorname{-\!-}}

\DeclareOldFontCommand{\bf}{\normalfont\bfseries}{\mathbf}

\providecommand{\PSet}{{\mathcal P}}				
\providecommand{\Id}{\operatorname{Id}}

\providecommand{\Hom}{\mathsf{Hom}}
\providecommand{\id}{\mathsf{id}}

\providecommand{\comp}{\mathbin{\circ}}

\providecommand{\bang}{\operatorname!}				

\providecommand{\dar}{\kern-1.2pt\operatorname{\downarrow}}	
\providecommand{\uar}{\kern-1.2pt\operatorname{\uparrow}}	

\providecommand{\xto}[1]{\xrightarrow{#1}}

\providecommand{\bigjoin}{\bigsqcup}

\providecommand{\fst}{\oname{fst}}
\providecommand{\snd}{\oname{snd}}

\providecommand{\brks}[1]{\langle #1\rangle}

\providecommand{\inl}{\oname{inl}}
\providecommand{\inr}{\oname{inr}}
\providecommand{\inj}{\oname{in}}

\DeclareSymbolFont{Symbols}{OMS}{cmsy}{m}{n}
\DeclareMathSymbol{\iobj}{\mathord}{Symbols}{"3B}

\providecommand{\curry}{\oname{curry}}
\providecommand{\uncurry}{\oname{uncurry}}
\providecommand{\ev}{\oname{ev}}

\providecommand{\lsem}{\llbracket}
\providecommand{\rsem}{\rrbracket}
\providecommand{\sem}[1]{\lsem #1 \rsem}

\providecommand{\by}[1]{\text{/\!/~#1}}			
\providecommand{\pacman}[1]{}					

\providecommand{\noqed}{\def\qed{}}				

\providecommand{\mone}{{\text{\kern.5pt\rmfamily-}\sf\kern-.5pt1}}

\makeatletter
\@ifpackageloaded{enumitem}{}{\usepackage[loadonly]{enumitem}}		
\makeatother

\newlist{citemize}{itemize}{1}
\setlist[citemize]{label=\labelitemi,wide} 

\newlist{cenumerate}{enumerate}{1}
\setlist[cenumerate,1]{label=\arabic*.~,ref={\arabic*},wide} 

\makeatletter
\def\mfix#1{\oname{#1}\@ifnextchar\bgroup\@mfix{}}	
\def\@mfix#1{#1\@ifnextchar\bgroup\mfix{}}			
\makeatother

\providecommand{\case}[3]{\mfix{case}{\mathbin{}#1}{of}{#2}{\kern-1pt;}{\mathbin{}#3}}

\begin{document}

\title[Unguarded Recursion on Coinductive Resumptions]{Unguarded Recursion on Coinductive Resumptions\rsuper*}
\titlecomment{{\lsuper*}This work forms part of the DFG project HighMoon2 (GO 2161/1-2 / SCHR 1118/8-2)}

\author[Goncharov]{Sergey Goncharov}	%

\author[Rauch]{Christoph Rauch%
}	%

\author[Schr\"oder]{Lutz Schr\"oder}	%

\author[Jakob]{Julian Jakob}	%
\address{Department of Computer Science, Friedrich-Alexander-Universit\"at Erlangen-N\"urnberg}	%
\email{\{sergey.goncharov,christoph.rauch,lutz.schroeder,julian.jakob\}@fau.de}  %


\begin{abstract}
We study a model of side-effecting processes obtained by starting
from a monad modelling base effects and adjoining free operations
using a cofree coalgebra construction; one thus arrives at what one
may think of as types of non-wellfounded side-effecting trees,
generalizing the infinite resumption monad. Correspondingly, the
arising monad transformer has been termed the coinductive
generalized resumption transformer. Monads of this kind have
received some attention in the recent literature; in particular, it
has been shown that they admit guarded iteration. Here, we show that
they also admit unguarded iteration, i.e.\ form complete Elgot
monads, provided that the underlying base effect supports unguarded
iteration. Moreover, we provide a universal characterization of the
coinductive resumption monad transformer in terms of coproducts of
complete Elgot monads.

\end{abstract}
\maketitle

\allowdisplaybreaks[4]

\section{Introduction}

\noindent Subsequent to seminal work by Moggi~\cite{Moggi91a}, monads are
widely used to represent computational effects in program semantics,
and in fact in actual programming languages~\cite{Wadler97}. Their
main attraction lies in the fact that they provide an interface to a
generic notion of side-effect at the right level of abstraction: they
subsume a wide variety of side-effects such as state, nondeterminism,
random, and I/O, and at the same time retain enough internal structure
to support a substantial amount of generic meta-theory and
programming, the latter witnessed, for example, by the monad class
implemented in the Haskell basic libraries~\cite{Peyton-Jones03}.

In the current work, we study a particular construction on monads
motivated partly by the goal of modelling generic side-effects in the
semantics of reactive processes. Specifically, given a base monad~$T$
and a strong functor~$\SigF$, we have final coalgebras
\begin{equation*}
  \TF X = \nu\gamma.\,T(X+\SigF\gamma)
\end{equation*}
for each object~$X$, assuming enough structure on~$T$,~$\Sigma$, and
the base category.  Inhabitants of $\TF X$ are understood as
(possibly) nonterminating processes that proceed in steps, where each
step produces side-effects specified by~$T$ (e.g.~writing to shared
global memory, nondeterminism) and performs communication actions
specified by~$\Sigma$. E.g.~in the simplest case,~$\Sigma$ is of the
form $a\times (-)^b$, which may be understood as reading inputs of
type~$b$ and writing outputs of type~$a$.

The construction of~$\TF X$ from~$T$ is an infinite version
of the generalized resumption transformer introduced by Cienciarelli
and Moggi~\cite{CenciarelliMoggi93}. It has been termed the
\emph{coinductive generalized resumption} transformer by Pir\'og and
Gibbons~\cite{PirogGibbons13,PirogGibbons14}, who show that on the
Kleisli category of~$T$,~$\TF$ is the free completely iterative monad
generated by~${T\Sigma}$.

The result that~$\TF$ is a completely iterative monad brings us to
the contribution of the current paper. Recall that complete
iterativity of~$\TF$ means that for every morphism
\begin{equation*}
  e:X\to \TF(Y+X),
\end{equation*}
read as an equation defining the inhabitants of~$X$, thought of as
variables, as terms over the defined variables (from~$X$) and
parameters from~$Y$, has a unique \emph{solution}
\begin{equation*}
  e^\istar:X\to \TF Y
\end{equation*}
in the evident sense, \emph{provided} that~$e$ is \emph{guarded}. The
latter concept is defined in terms of additional structure of~$\TF$
as an \emph{idealized monad}, which essentially allows distinguishing
terms beginning with an operation from mere variables. Guardedness of
$e$ then means that recursive calls can happen only under a free
operation. Similar results on guarded recursion abound in the
literature; for example, the fact that~$\TF$ admits guarded recursive
definitions can also be deduced from more general results by Uustalu
on parametrized monads~\cite{Uustalu03}.

The central result of the current paper is to remove the guardedness
restriction in the above setup. That is, we show that a solution
$e^\istar:X\to\TF Y$ exists for \emph{every} morphism
$e:X\to\TF(X+Y)$. Of course, the solution is then no longer unique
(for example, we admit definitions of the form $x=x$); moreover, we
clearly need to make additional assumptions about~$T$. Our result
states, more precisely, that~$\TF$ allows for a principled
\emph{choice} of solutions~$e^\istar$ satisfying standard equational
laws for recursion~\cite{SimpsonPlotkin00}, thus making~$\TF$ into a
\emph{complete Elgot monad}~\cite{AdamekMiliusEtAl10}\footnote{We
  modify the original definition of Elgot monad, which requires the
  object~$X$ of variables to be a finitely presentable object in an
  lfp category, by admitting unrestricted objects of variables. This
  change is owed mostly to the fact that we do not assume the base
  category to be lfp, and in our own estimate appears to be
  technically inessential, although we have not checked details for
  the obvious variants of our results that arise by replacing complete
  Elgot monads with Elgot monads.}. The assumption on~$T$ that we need
to enable this result is that~$T$ itself is a complete Elgot monad (e.g.\
partiality, nondeterminism, or combinations of these with state),
i.e.\ we show that \emph{the class of complete Elgot monads is stable under the
  coinductive generalized resumption transformer}. We show moreover
that the structure of~$\TF$ as a complete Elgot monad is uniquely determined
as extending that of~$T$.

The motivation for these results is, well, to free non-wellfounded
recursive definitions from the standard guardedness constraint. Note
for example that in~\cite{PirogGibbons13}, it was necessary to assume
guards in all loop iterations when interpreting a while-language with
actions originally proposed by Rutten~\cite{Rutten99} over a
completely iterative monad. Contrastingly, given that~$\TF$ is a
(complete) Elgot monad, one can now just write unrestricted while
loops. We elaborate this example in Section~\ref{sec:examples}, and
recall a standard example of unguarded recursion in process algebra in
Section~\ref{sec:bsp}. 

An earlier version of this work has appeared as~\cite{GoncharovRauchEtAl15};
the present version not only has full proofs, but also works in a
generalized setup with an arbitrary strong functor~$\Sigma$ (admitting
the requisite final coalgebras) instead of just functors of the form
$a\times(-)^b$.

The material is organized as follows. We present the mentioned
examples involving unguarded iteration in Sections~\ref{sec:examples}
and~\ref{sec:bsp}.  In Section~\ref{sec:prelim}, we collect
preliminaries on (strong) monads and their Kleisli categories. We
discuss the concept of complete Elgot monad in
Section~\ref{sec:elgot}, and recall the coinductive generalized
resumption transformer in Section~\ref{sec:resumptions}, showing in
particular that it preserves strength.  Sections~\ref{sec:cr_elgot}
and~\ref{sec:sums} contain our main results, showing that the
coinductive generalized resumption transformer preserves complete Elgotness and
can be seen as freely extending complete Elgot monads with communication
actions. We discuss related work in Section~\ref{sec:related}, and
conclude in Section~\ref{sec:concl}.

\section{Example: Unrestricted While Loops}\label{sec:examples}

We proceed to discuss examples, aimed mainly at illustrating the
benefits of not being restricted to guarded equations in recursive
definitions thanks to complete Elgotness of coinductive resumption
monads (Theorem~\ref{thm:cr_elg}). We work with the intuitive
understanding of monads,~$T_\Sigma$, guardedness, and complete
iterativity provided in the introduction, and briefly explain the
requisite categorical notation regarding strong monads and
distributive categories along the way, deferring a more formal
treatment to Sections~\ref{sec:prelim} and~\ref{sec:resumptions}. 

Our first example is a simple while-language with actions proposed by
Rutten, given by the grammar
\begin{equation*}
  P,Q::= A\mid P;Q\mid\ifTerm{b}{P}{Q}\mid\whileTermS{b}{P} %
\end{equation*}
and, following Pir\'og and Gibbons~\cite{PirogGibbons13}, interpreted
in the Kleisli category of a monad~$\BBM$. Here,~$A$ ranges over atomic
actions interpreted as Kleisli morphisms $\sem{A}:n\to M n$ for some
fixed object~$n$,
and~$b$ over atomic predicates, interpreted as Kleisli morphisms
$\sem{b}:n\to M(1+1)$ (where we read the left-hand summand as `false'
and the right-hand one as `true', and~$1$ denotes the terminal
object). We say that~$A$ is of \emph{output type} if $\sem{A}:n\to Mn$
has the form $\sem{A}=(M\fst)\comp\tau\comp\brks{\id_n,p}$ for some
$p:n\to M1$, where~$\fst$ denotes first projection and
$\tau:n\times M1\to M(n\times 1)$ is the strength
of~$M$. Moreover,~$A$ is of \emph{input type} if $\sem{A}:n\to Mn$
factors through the unique morphism $!:n\to 1$.  Sequential
composition $P;Q$ is interpreted as Kleisli composition
$\sem{Q}^\klstar \comp \sem{P}$, and
\begin{align*}
\sem{\ifTerm{b}{P}{Q}}=[\sem{Q}\comp\fst,\sem{P}\comp\fst]^\klstar\comp (M\dist)\comp\tau\comp\brks{\id_n,\sem{b}}
\end{align*}
where $\dist:n\times(1+1)\to (n\times 1)+(n\times 1)$ is a
distributivity isomorphism that we postulate in our general setup
(Section~\ref{sec:prelim}). The key point, of course, is the
interpretation of the while loop, given in the presence of iteration
$(\argument)^\istar$ by
\begin{equation}\label{eq:while-sem}
  \sem{\whileTermS{b}{P}}=\bigl([(M\inl)\comp\eta\comp\fst, (M\inr)\comp\sem{P}\comp\fst]^\klstar\comp (M\dist)\comp\tau\comp\brks{\id,\sem{b}}\bigr)^\istar
\end{equation}
where the typing of the expression under the iteration operator
$(\argument)^\istar$ is visualized as
\begin{align*}
  n \xto{~~\brks{\id,\sem{b}}~~}~~& n\times M(1+1)\\
  \xto{~~\tau~~}~~&M(n\times(1+1))\\
  \xto{~~M\dist~~}~~&M(n\times 1 +n\times 1)\qquad \\
  \xto{~~[(M\inl)\comp\eta\comp\fst, (M\inr)\comp\sem{P}\comp\fst]^\klstar~~}~~&M(n+n).
\end{align*}
It has been observed by Pir\'og and Gibbons that if one instantiates
$\BBM$ with a completely iterative monad, one needs to guard every
iteration of the while loop, i.e.\ change the semantics of while to be
\begin{equation*}
  \sem{\whileTermS{b}{P}}=\bigl([(M\inl)\comp\eta\comp\fst, (M\inr)\comp\sem{P}\comp\fst]^\klstar\comp (M\dist)\comp \tau\comp\brks{\id,\sem{b}}\comp\gamma\bigr)^\istar
\end{equation*}
where $\gamma:n\to Mn $ is guarded, as otherwise the iteration may
fail to be defined (recall from the introduction that over completely
iterative monads, definedness of iteration depends on guardedness). If
we instantiate~$\BBM$ with a complete Elgot monad, such as~$\BBTF$ for
a complete Elgot monad~$\BBT$ (by Theorem~\ref{thm:cr_elg}), then the
guard is unnecessary, i.e.\ we can stick to the original
semantics~\eqref{eq:while-sem}. As an example, consider a
simple-minded form of processes that input and output symbols from~$n$
and have side effects specified by~$\BBT$; i.e.\ we work in
$\BBM=\BBT_\Sigma$ for $\Sigma X=n\times X+X^n$ where we think of
$\Sigma$ as being generated by an output operation $1\to n$ and an
input operation $n\to 1$. We correspondingly assume an atomic action
$\wrt$ that outputs a symbol from~$n$, and an atomic action~$\rd$ that
inputs a symbol. We interpret $\wrt$ as being of output type, i.e.~by
$\sem{\wrt}=(M\fst)\comp\tau\comp\brks{\id_n,w}$ where $w:n\to M1$ is
obtained from a canonical transformation $\iota^\BBT:\Sigma\to\TF=M$
that will be introduced in Section~\ref{sec:sums}; intuitively,
$\iota^\BBT$ converts actions into single-step processes without side
effects. Explicitly,~$w$ is the composite
\begin{equation*}
  n\xto{~~\brks{\id_n,!_n}~~} n\times 1\xto{~~\inl~~}\Sigma 1
  \xto{~~\iota^\BBT~~}M1.
\end{equation*}
Moreover, we interpret $\rd$ as being of input type,
i.e.~$\sem{\rd}=r\,\comp\,!_n$ where $r:1\to M n$ is obtained
analogously, i.e.~$r$ is the composite
\begin{equation*}
  1 \xto{~~r_0~~} n^n \xto{~~\inr~~} \Sigma n\xto{~~\iota^\BBT~~}  Mn
\end{equation*}
where $r_0:1\to n^n$ arises by currying $\snd:1\times n\to n$.
Moreover,
assume a basic predicate~$b$ whose interpretation is largely
irrelevant to the example as long as it may take both truth values;
for example,~$b$ might just pick a truth value nondeterministically
or at random, depending on the nature of the base monad~$\BBT$. Consider
the program
\begin{equation*}
  \rd;\whileTermS{\mathit{true}}{\ifTerm{b}{\skp}{\wrt}}
\end{equation*}
where $\skp$ is an atomic action interpreted as the unit of~$M$, a
process that does nothing and terminates immediately. It is possible
for the loop to not perform any write operations, as~$b$ might happen
to always pick the left-hand branch; that is, the loop body fails to
be guarded. Since~$M$ is a complete Elgot monad and not just
completely iterative, the semantics of the loop is defined
(by~\eqref{eq:while-sem}) nonetheless.

\section{Example: Simple Process Algebra}\label{sec:bsp}

Baeten et al.~\cite{BaetenBastenEtAl10} introduce a simple process
algebra BSP (\emph{Basic Sequential Processes}) featuring finite
choice and action prefixing, and show that it can express all
countable transition systems if unguarded recursion is
allowed~\cite[Theorem 5.7.3]{BaetenBastenEtAl10}. The idea of the
proof is to introduce variables~$X_{ik}$ for $i,k\in\mathbb{N}$
representing the $k$-th transition of the $i$-th state, with~$X_{i0}$
representing the $i$-th state itself, and (unguarded) recursive
equations
\begin{equation}\label{eq:bsp}
  X_{ik}=b_{ik}.X_{j(i,k),0}+X_{i,k+1}
\end{equation}
where the $k$-th transition of the $i$-th state performs action
$b_{ik}$ and reaches the $j(i,k)$-th state. (The use of unguarded
recursion is essential here, as guarded recursive definitions in BSP
will clearly produce only finitely branching systems.)  To model this
phenomenon using the coinductive generalized resumption transformer,
we take $\BBT=\Pow_{\omega_1}$, the countable powerset monad on $\Set$
(details are in Example~\ref{exp:guard_pa}), and the functor~$\Sigma$
generated by $a$-many unary operations where~$a$ is the set of
actions; that is, $\Sigma X=a\times X$.  We thus regard countable
nondeterminism as the base effect, and add action prefixing via
coinductive generalized resumptions.  Representing variables~$X_{i,k}$
by their indices $(i,k)$, we then cast the definition~\eqref{eq:bsp}
as an equation morphism
\begin{equation*}
  e:\Nat\times\Nat\to\TF(\Nat\times\Nat)\cong
  \TF(0+\Nat\times\Nat)
\end{equation*}
as follows. Eliding isomorphic conversions, we write elements of
$\TF(\Nat\times\Nat)$ as subsets of
$(\mathbb{N}\times\mathbb{N})+a\times
\TF(\mathbb{N}\times\mathbb{N})$; in this notation,
\begin{equation*}
  e(i,k)=\{\inr(b_{ik},\{\inl(j(i,k),0)\}),\inl(i,k+1)\}.
\end{equation*}
Again, our result that~$\BBTF$ is a complete Elgot monad
(Theorem~\ref{thm:cr_elg}) guarantees that this equation has a
solution~$e^\istar$, and moreover that the choice~$(-)^\istar$ of
solutions in~$\BBTF$ is uniquely determined as forming a complete Elgot monad
and extending the usual structure of $\BBT=\Pow_{\omega_1}$ as a
complete Elgot monad, which takes least fixed points. We emphasize that
solutions in~$\BBTF$ do not arise as least fixed points; in
particular, recall that simulation is only a preorder on $\TF X$.

\section{Preliminaries}\label{sec:prelim}
\noindent According to Moggi~\cite{Moggi91}, a notion of computation
can be formalized as a strong monad~$\BBT$ over a Cartesian category (i.e.\ a category with finite products).
In order to support the constructions occurring in the main object of study, we work in a \emph{distributive category}~$\BC$, i.e.\ a category with finite products and coproducts (including a final and an initial object) such that the natural transformation
\begin{align*}
X\times Y + X\times Z \xto{~~[\id\times\inl,\id\times\inr]~~} X\times (Y+Z)
\end{align*}
is an isomorphism~\cite{Cockett93}, whose inverse we denote by $\dist_{X,Y,Z}$.
Here we denote injections into binary coproducts by $\inl:X\to X+Y$, $\inr:Y\to
X+Y$, while $\fst:X\times Y\to X$, $\snd:X\times Y\to Y$ denote projections from
binary products; pairing is denoted by $\langle\argument,\argument\rangle$, and
copairing of $f:X\to Z$, $g:Y\to Z$ by $[f,g]:X+Y\to Z$. Unique morphisms $A\to
1$ into the terminal object are written~$!_X$, or just~$!$. We write~$|\BC|$ for
the class of objects of~$\BC$.
Distributivity essentially allows using context variables in case
expressions, i.e.\ in copairing.
We omit indices on natural
transformations where this is unlikely to cause confusion.%

A \emph{strong} functor on~$\BC$ is a functor $F:\BC\to\BC$ equipped
with a natural transformation 
\begin{equation*}
  \rho_{X,Y}: X\times FY\to F(X\times Y)
\end{equation*}
called \emph{strength}, subject to the equations
\begin{align*}
	\snd &= (F \snd)\comp\rho \tag{\sc str$_1$} \\
	(F\assoc)\comp\rho &= \tau\comp(\id \times \rho)\comp\assoc \tag{\sc str$_2$} 
\end{align*}
where $\assoc: (X \times Y) \times Z \to X \times (Y \times Z)$ is the
associativity isomorphism of products, explicitly,
$\assoc=\brks{\fst\comp\fst,\brks{\snd\fst,\snd}}$.  A natural
transformation $\alpha:F\to G$ between strong functors~$F$,~$G$ (with
the strength denoted~$\rho$ in both cases) is \emph{strong} if it
commutes with strength:
\begin{equation*}
  \xymatrix@R25pt@C58pt@M6pt{X\times FY \ar[d]_{\rho_{X,Y}} \ar[r]^{\id_X\times\alpha_Y}
  & X\times GY\ar[d]^{\rho_{X,Y}}\\
  F(X\times Y)\ar[r]_{\alpha_{X\times Y}} & G(X\times Y)
  }
\end{equation*}

\noindent Recall that a monad~$\BBT$ over~$\BC$ can be given by a
\emph{Kleisli triple} $(T,\eta,\argument^{\klstar})$ where~$T$ is an
endomap of~$|\BC|$ (in the following, we always denote monads
and their functor parts by the same letter, with the former in
blackboard bold), the \emph{unit}~$\eta$ is a family of morphisms
$\eta_X:X\to TX$, and the \emph{Kleisli lifting} $(\argument)^{\klstar}$
maps $f:X\to TY$ to $f^{\klstar}:TX\to TY$, subject to the equations
\begin{align*}
\eta^{\klstar}=\id&& f^{\klstar}\comp\eta=f&& (f^{\klstar}\comp g)^{\klstar}=f^{\klstar}\comp g^{\klstar}.
\end{align*}
This is equivalent to the presentation in terms of an endofunctor $T$
with natural transformations unit and multiplication. 

A \emph{strong monad} is a monad whose underlying endofunctor is strong and the corresponding 
strength $\tau$ additionally satisfies the following additional coherence
conditions~\cite{Moggi91} (with modifications reflecting the switch
from monad multiplication to Kleisli lifting):
\begin{align*}\label{page:str}
	\tau\comp(\id \times \eta) &= \eta \tag{\sc str$_3$} \\
	(\tau\comp(\id \times f))^{\klstar}\comp\tau &= \tau\comp(\id \times f^{\klstar}) \tag{\sc str$_4$}
\end{align*}
The typing of the law
({\sc str$_4$}) capturing compatibility of the strength with Kleisli
lifting is shown in the diagram
\begin{equation*}
  \xymatrix@R25pt@C28pt@M6pt{X\times TY \ar[rr]^{\id\times f^\klstar}
  \ar[d]_\tau & & X\times TZ \ar[d]^{\tau}\\
  T(X\times Y) \ar[rr]_{(\tau(\id\times f))^\klstar} && T(X\times Z)\\
  X\times Y  \ar[r]^{\id\times f} & X\times TZ \ar[r]^{\tau}& T(X\times Z)
  }
\end{equation*}
(For distinction, we denote strengths of monads by $\tau$ and
strengths of functors by $\rho$ throughout.)  Strength enables
interpreting programs over more than one variable, and allows for
internalization of the Kleisli lifting, thus legitimating expressions
like $\lambda x.\,(f(x))^\klstar:X\to (TY\to TZ)$ for
$f:X\to (Y\to TZ)$, which encodes
$\curry(\uncurry(f)^\klstar\comp\tau)$.  Strength is equivalent to the
monad being enriched over~$\BC$~\cite{Kock72}; in particular, every
monad on $\Set$ is strong. Henceforth we shall use the term `monad' to
mean `strong monad' unless explicitly stated otherwise. We emphasize
however that all our results remain valid under the removal of all
strength assumptions and claims (that is, replacing the terms strong
monad, strong functor, and strong natural transformation with monad,
functor, and natural transformation, respectively, throughout).

The standard intuition for a monad $\BBT$ is to think of $TX$ as the
set of terms in some algebraic theory, with variables taken from
$X$. In this view, the unit converts variables into terms, and a
Kleisli lifting $f^\klstar$ applies a substitution $f:X\to TY$ to
terms over $X$. In our setting, the `terms' featuring here are often
infinite; nevertheless, we sometimes call them \emph{algebraic terms}
for emphasis.

The \emph{Kleisli category} $\BC_{\BBT}$ of a monad $\BBT$ has the
same objects as $\BC$, and $\BC$-morphisms $X\to TY$ as morphisms
$X\to Y$. The identity on $X$ in $\BC_{\BBT}$ is $\eta_X$; and the
\emph{Kleisli composite} of $f:X\to TY$ and $g:Y\to TZ$ is
$g^{\klstar}\comp f$. %
A monad $\BBT$ \emph{has rank $\kappa$} for a regular
cardinal~$\kappa$ if~$T$ preserves $\kappa$-filtered colimits. On
$\Set$, this condition means that $T$ is determined by its values on
sets of cardinality less than $\kappa$, in the sense that every
element of $TX$ comes from an element of $TY$ for some
subset~$Y\subseteq X$ with $|Y|<\kappa$; intuitively, all operations
of~$T$ have arity less than~$\kappa$. A monad is \emph{ranked} if it
has some rank $\kappa$.

\begin{rexa}\label{expl:sig}
  As indicated in the introduction, in the main motivating examples
  the strong functor $\Sigma$ plays the role of a signature of
  communication actions. Technical details are as follows. Assume that
  $\BC$ has \emph{exponentials} of the form $X^b$ (for $b$ ranging
  over a subset of $|\BC|$), i.e.\ objects adjoint to Cartesian
  products $X\times b$, which means that for any $X$ and~$Y$, there is
  an isomorphism
\begin{align*}
\curry_{X,Y}:\Hom_{\BC}(X\times b,Y)\cong\Hom_{\BC}(X,Y^b),
\end{align*}
natural in $X$ and $Y$. We write $\uncurry_{X,Y}$ for the inverse map
$\curry^{\mone}_{X,Y}$. The \emph{evaluation morphism}
$\ev_X: X^b\times b\to X$ (natural in $X$) is obtained as
$\uncurry_{X^b,X}(\id_{X^b})$. 

It is easy to see that the functors $X\mapsto a\times X$ and
$X\mapsto X^b$ are strong and that composites and coproducts of strong
functors (as plain functors) are again strong functors. Hence, the
functor
\begin{align*}
\Sigma X = \sum_{i} a_i\times X^{b_i}
\end{align*}
is strong. Intuitively (and formally correctly on $\Set$), $\Sigma X$
can be seen as the set of flat terms over variables from $X$ in the
signature~$\Sigma$, i.e.\ the elements of $\Sigma X$ are of the form
$f_i(c;x_1,\ldots,x_{n_i})$ where $f$ is a parametrized operation from
the signature, $c$ is a \emph{parameter} from $a_i$, and
$x_1,\ldots,x_{n_i}$ are elements of $X$.
The computational meaning of exponents in~$X^b$ is thus to capture a
notion of arity of algebraic operations generating effects, e.g.\
$b=2$ would correspond to binary operations such as nondeterministic
choice. An example of an operation taking a parameter would be the
operation of writing a value $\mathit{val}$ to position $\mathit{ind}$
of an array, $\mathit{update}(\brks{\mathit{ind}, \mathit{val}}; x)$
(see~\cite{PlotkinPower02} for details).

A more general setup involves categories enriched over a symmetric
monoidal closed category $\BV$ whose objects are then treated as
arities (and \emph{coarities}, i.e.\ objects used for indexing
families of operations)~\cite{HylandPower06,HylandPlotkinEtAl06}.  One
then replaces products with \emph{tensors} and exponentials with
\emph{cotensors}. 

Another example are functors on the topos of nominal sets
and equivariant maps $\Nom$ built using constant functors, identity,
coproducts, finite products, and the so-called \emph{abstraction
  functor} $[\Names](-)$, where $\Names$ is a set of \emph{names} and
$[\Names]X$ consists of pairs $(a,x)\in\Names\times X$ modulo a
natural notion of $\alpha$-equivalence~\cite{Pitts13}. Such functors
represent so-called \emph{binding signatures}, whose operations may
bind names, such as $\lambda$-abstraction or $\pi$-calculus-style
fresh name binders~$\nu$; terms are then taken modulo
$\alpha$-equivalence. E.g.~the $\lambda$-calculus syntax is
rendered as the initial algebra of the functor
$LX=\Names+X\times X +[\Names]X$ (see~\cite{GabbayPitts99}).
\end{rexa}

\section{Complete Elgot Monads}\label{sec:elgot}

As indicated in the introduction, we will be interested in recursive
definitions over a monad~$\BBT$; abstractly, these are morphisms
\begin{equation*}
  f:X\to T(Y+X)
\end{equation*}
thought of as associating to each variable $x:X$ a definition $f(x)$ in
the shape of an algebraic term from $T(Y+X)$, which thus employs
parameters from $Y$ as well as the defined variables from $X$. The
latter amount to recursive calls of the definition. This notion is
agnostic to what happens in the case of non-terminating recursion. For
example, $T$ might identify all non-terminating sequences of recursive
calls into a single value $\bot$ signifying non-termination; at the
other extreme, $T$ might be a type of infinite trees that just records
the tree of recursive calls explicitly.

To a recursive definition $f$ as above, we wish to associate a solution
\begin{equation*}
  f^\istar:X\to TY,
\end{equation*}
which amounts to a non-recursive definition of the elements of $X$ as
terms over $Y$ only. As we do not assume any form of guardedness, this
solution will in general fail to be unique. We thus require a coherent
selection of solutions $f^\istar$ for all equations $f$, where by
coherent we mean that the selection satisfies a collection of well-established
\mbox{(quasi-)}equational properties. Formally:

\begin{defi}\textbf{(Complete Elgot monads)}\label{defn:elgot}
  A \emph{complete Elgot monad} is a monad $\BBT$ equipped with an
  operator $\argument^\istar$, called \emph{iteration}, that assigns
  to each morphism $f:X\to T(Y+X)$ a morphism $f^{\istar}:X\to TY$
  such that the following laws hold:
\begin{itemize}
  \item\emph{fixpoint:} $[\eta, f^{\istar}]^{\klstar} \comp f = f^{\istar}$;
  \item\emph{naturality:} $g^{\klstar} \comp f^{\istar} = ([T \inl \comp\,g, \eta
	\comp\,\inr]^{\klstar} \comp f)^{\istar}$ for $g : Y \to TZ$;
  \item\emph{codiagonal:} $(T[\id,\inr] \comp\, g)^{\istar} = g^{\istar\istar}$ for  $g : X \to T((Y + X) + X)$;
  \item\emph{uniformity:} $f \comp h = T(\id+ h) \comp g$ implies
	$f^{\istar} \comp h = g^{\istar}$ for  $g: Z \to T(Y + Z)$ and
	$h: Z \to X$.
\end{itemize}
Additionally, iteration must be \emph{compatible with strength} in the
sense that
\begin{equation*}
\tau \comp (\id \times f^{\istar}) = (T \dist \comp\, \tau \comp (\id
\times f))^{\istar}
\end{equation*}
for $f:X\to T(Y+X)$.
\end{defi}
\noindent It has recently been shown~\cite{EsikGoncharov16,GoncharovMiliusEtAl16} that
\emph{dinaturality}, previously standardly included in axiomatizations
of iteration~\cite{BloomEsik93}, is in fact derivable from the other
axioms in Definition~\ref{defn:elgot}. We record this for future
reference:
\begin{lem}[Dinaturality]\label{lem:dinat}
  Every complete Elgot monad satisfies \emph{dinaturality}:
  \begin{equation*}
  ([\eta \comp \inl, h]^{\klstar} \comp g)^{\istar} = [\eta,
  \left([\eta \comp \inl, g]^{\klstar} \comp
  h\right)^{\istar}]^{\klstar} \comp g
  \text{ for }g : X \to T(Y + Z)\text{ and }h:Z\to T(Y+X).
  \end{equation*}
\end{lem}
\begin{rem}
  The above definition is inspired by the axioms of parametrized
  uniform iterativity~\cite{SimpsonPlotkin00}, which go back to
  Bloom and \'{E}sik~\cite{BloomEsik93}. Ad\'amek et
  al.~\cite{AdamekMiliusEtAl10} define \emph{Elgot monads} by means of
  a slightly different system of axioms: the codiagonal (and
  dinaturality) laws are replaced with the \emph{Beki\'{c}
  identity}. Both axiomatizations are however equivalent, which is
  essentially a result about iteration
  theories~\cite[Section~6.8]{BloomEsik93}; we record a self-contained
  proof of this equivalence in Proposition~\ref{prop:bekic-from-elgot}
  below. Moreover, the iteration operator in~\cite{AdamekMiliusEtAl10}
  is defined only for $f:X\to T(Y+X)$ with finitely presentable $X$,
  under the assumption that~$\BC$ is \emph{locally finitely
  presentable}; hence our use of the term `complete Elgot monad'
  instead of `Elgot monad'. We have the impression that this
  difference is not technically essential but have not checked details
  for the finitary variant of our results.
\end{rem}
\begin{prop}[Beki\'{c} identity]\label{prop:bekic-from-elgot}
  A complete Elgot monad $\BBT$ is equivalently a monad satisfying
  \emph{fixpoint}, \emph{naturality}, \emph{uniformity} (as in
  Definition~\ref{defn:elgot}), and the \emph{Beki\'{c} identity}
  \begin{align}
  \tag{Beki\'{c}} ((T\alpha) \comp [f, g])^{\istar} = [\eta,
  h^{\istar}]^{\klstar} \comp [\eta \comp \inr, g^{\istar}]
  \end{align}
  where $g : X \to T((Z + Y) + X)$, $f : Y \to T((Z + Y) + X)$,
  $h = [\eta, g^{\istar}]^{\klstar} \comp f : Y \to T(Z + Y)$, with
  $\alpha : (A + B) + C \to A + (B + C)$ being the obvious coproduct
  associativity morphism.
\end{prop}
\begin{proof}
  Let us show that complete Elgot monads validate the Beki\'{c}
  identity. Let
  \begin{align*}
  u = T((\id + \inl) + \inr) \comp [f, g] : Y + X \to T((Z+(Y+X))
  + (Y+X)).
  \end{align*}
  By \textit{codiagonal},
  \begin{equation}
  (T[\id, \inr] \comp u)^{\istar} =
  (u^{\istar})^{\istar}.\label{eq:u-codiag}
  \end{equation}
  Now the left-hand side of \eqref{eq:u-codiag} simplifies to
  \begin{align*}
  &~(T[\id, \inr] \comp T((\id + \inl) + \inr) \comp [f,g])^{\istar} \\
  =&~(T[\id + \inl, \inr \comp \inr] \comp [f, g])^{\istar} \\
  =&~((T\alpha) \comp [f, g])^{\istar},
  \end{align*}  
  i.e.~to the left-hand side of the Beki\'{c} identity.  Now observe
  that, by \textit{uniformity} and \textit{naturality},
  \begin{align}\label{eq:bek_proof}
  u^{\istar} \comp \inr = (T(\id + \inl) + \id) \comp g)^{\istar} = T(\id
  + \inl) \comp g^{\istar}.
  \end{align}
  Therefore, the right-hand side of \eqref{eq:u-codiag} can be
  rewritten in the form
  \begin{flalign*}
  &&(u^{\istar})^{\istar} =&~ ([\eta, u^{\istar}]^{\klstar} \comp u)^{\istar}&\by{fixpoint} \\
  &&=&~ ([\eta \comp (\id + \inl), u^{\istar} \comp \inr] \comp [f, g])^{\istar} \\
  &&=&~ ([T(\id + \inl) \comp \eta, T(\id + \inl) \comp g^{\istar}]^{\klstar} \comp [f, g])^{\istar} &\by{\ref{eq:bek_proof}}\\
  &&=&~ (T(\id + \inl) \comp [\eta, g^{\istar}]^{\klstar} \comp [f, g])^{\istar} \\
  &&=&~ (T(\id + \inl) \comp [[\eta, g^{\istar}]^{\klstar} \comp f, g^{\istar}])^{\istar} &\by{fixpoint} \\
  &&=&~ ([\eta \comp \inl, \eta \comp \inr \comp \inl]^{\klstar} \comp [[\eta, g^{\istar}]^{\klstar} \comp f, g^{\istar}])^{\istar} \\
  &&=&~ [\eta, ([\eta \comp \inl, [[\eta, g^{\istar}]^{\klstar} \comp f, g^{\istar}]]^{\klstar} \comp \eta \comp \inr \comp \inl)^{\istar}]^{\klstar}\\&&& ~~\comp [[\eta, g^{\istar}]^{\klstar} \comp f, g^{\istar}] &\by{dinaturality, Lemma~\ref{lem:dinat}} \\
  &&=&~ [\eta, ([\eta, g^{\istar}]^{\klstar} \comp f)^{\istar}]^{\klstar} \comp [[\eta, g^{\istar}]^{\klstar} \comp f, g^{\istar}] \\
  &&=&~ [([\eta, g^{\istar}]^{\klstar} \comp f)^{\istar}, [\eta, ([\eta, g^{\istar}]^{\klstar} \comp f)^{\istar}]^{\klstar} \comp g^{\istar}] &\by{fixpoint} \\
  &&=&~ [h^{\istar}, [\eta, h^{\istar}]^{\klstar} \comp g^{\istar}] \\
  &&=&~ [\eta, h^{\istar}]^{\klstar} \comp [\eta \comp \inr, g^{\istar}],
  \end{flalign*}
  i.e.~equals the right-hand side of the Beki\'{c} identity.

  For the opposite direction, we need to show that the Beki\'{c}
  identity implies \textit{codiagonal}. So let
  $k : X \to T((Y+X) + X)$. By the Beki\'{c} identity,
  \begin{align*}
  ((T\alpha) \comp [k,k])^{\istar} =&~ [\eta, ([\eta, k^{\istar}]^{\klstar} \comp
  k)^{\istar}]^{\klstar} \comp [\eta \comp \inr, k^{\istar}] \\
  =&~ [\eta, k^{\istar\istar}]^{\klstar} \comp [\eta \comp \inr, k^{\istar}].
\intertext{Thus, $((T\alpha) \comp [k, k])^{\istar} \comp \inl = ((T\alpha) \comp [k,
  k])^{\istar} \comp \inr = k^{\istar\istar}$. On the other hand, by
  \textit{uniformity},}
  ((T\alpha) \comp [k, k])^{\istar} =&~ (T[\id, \inr] \comp k)^{\istar} \comp [\id, \id]
\intertext{and therefore}
  k^{\istar\istar} =&~ ((T\alpha) \comp [k, k])^{\istar} \comp \inr =
  (T[\id, \inr] \comp k)^{\istar}
\end{align*}
  as required.
\end{proof}
\noindent Given a complete Elgot monad $\BBT$, we can parametrize the
iteration operator $\argument^\istar$ with an additional argument to
be carried over the recursion loop, i.e.\ we derive an operator
$\argument^\iistar$ sending $f:Z\times X\to T(Y+X)$ to
$f^\iistar:Z\times X\to TY$ by
\begin{equation}\label{eq:iter_par}%
f^\iistar=\bigl(T(\snd+\id)\comp (T\dist)\comp\tau_{Z,Y+X}\comp\brks{\fst,f}\bigr)^\istar.
\end{equation}
We call the derived operator $\argument^\iistar$ \emph{strong iteration}.

The key examples of complete Elgot monads are, one the one hand,
so-called $\omega$-continuous monads
(Definition~\ref{def:omega-cont}), and, on the other hand, extensions
of complete Elgot monads, e.g.~of $\omega$-continuous monads, with
free operations. The latter arise by application of the coinductive
generalized resumption transformer as introduced in
Section~\ref{sec:resumptions}. We proceed to discuss
$\omega$-continuous monads, which are defined as having a suitable
order-enrichment of their Kleisli category. Recall here that a
category $\BD$ is \emph{enriched} over a category $\BV$~\cite{Kelly82}
(in our application, $\BV$ is Cartesian; in general, $\BV$ only needs
to be monoidal) if $\BD$ has hom-objects from $\BV$ in place of
hom-sets, and both composition and selection of identities are
morphisms in $\BV$, with the usual equational laws of categories
expressed as commuting diagrams in $\BV$.
\begin{defi}\textbf{($\omega$-continuous monad)}\label{def:omega-cont}
  An \emph{$\omega$-continuous} monad consists of a monad $\BBT$ and
  an enrichment of the Kleisli category $\BC_{\BBT}$ of $\BBT$ over
  the category $\Cppo$ of $\omega$-complete partial orders with bottom and
  (nonstrict) continuous maps, satisfying the following conditions:
\begin{itemize}
 \item strength is $\omega$-continuous: $\tau \comp (\id\times\bigjoin_i f_i)=\bigjoin_i(\tau \comp (\id\times f_i))$;
 \item copairing in $\BC_{\BBT}$ is $\omega$-continuous in both arguments: $[\bigjoin_i f_i,\bigjoin_i g_i]= \bigjoin_i [f_i,g_i]$;
 \item bottom elements are preserved by strength and by postcomposition in $\BC_{\BBT}$: $\tau \comp {(\id\times\bot)}=\bot$, $f^\klstar\comp\bot=\bot$.
\end{itemize}
\end{defi}
\begin{exa}
  Many of the standard computational monads on $\Set$~\cite{Moggi91}
  are $\omega$-continuous, including nontermination ($TX=X+1$),
  nondeterminism ($TX=\PSet(X)$), and the nondeterministic state monad
  ($TX=\PSet(X\times S)^S$ for a set $S$ of states). On $\Cppo$, lifting
  ($TX=X_\bot$) and the various power domain monads are
  $\omega$-continuous.
\end{exa}

\begin{rem}\label{rem:cpo_cat}
  As observed by Kock~\cite{Kock72}, monad strength is equivalent to
  enrichment over the base category. One consequence of this
  fundamental fact is that if $\BC$ is enriched over the category
  $\Cpo$ of bottomless $\omega$-complete partial orders and
  $\omega$-continuous maps (i.e.\ $\BC$ is an $\BO$-category in the
  sense of Wand~\cite{Wand79} and of Smyth and
  Plotkin~\cite{SmythPlotkin77}), with the bi-Cartesian closed
  structure enriched in the obvious sense, then $\BC_{\BBT}$ is also
  enriched over $\Cpo$, since~$T$, being a strong functor, is an
  $\Cpo$-functor (aka \emph{locally continuous}
  functor~\cite{SmythPlotkin77}). Then $\BBT$ is $\omega$-continuous
  in the sense of Definition~\ref{def:omega-cont} iff each
  $\Hom(X,TY)$ has a bottom element preserved by strength and
  postcomposition in $\BC_{\BBT}$. This allows for incorporating
  numerous domain-theoretic examples by taking $\BC$ to be a suitable
  category of predomains, and $\BBT$, in the simplest case, the
  \emph{lifting monad} $TX=X_\bot$.
\end{rem}
\noindent
If $\BBT$ is an $\omega$-continuous monad, then the endomap
\begin{equation}\label{eq:lfp}
  h\mapsto [\eta,h]^\klstar \comp f
\end{equation}
on the hom-set $\Hom_{\BC}(A, TB)$ is continuous because copairing and
Kleisli composition in $T$ are continuous, and hence has a least
fixpoint by Kleene's fixpoint theorem. We can define an iteration
operator by taking $f^{\istar}$ to be this fixpoint; in other words,
$f^\istar$ is defined to be the least solution of the fixpoint
law as per Definition~\ref{defn:elgot}. This yields
\begin{thm}\label{thm:co_elg}
  On every $\omega$-continuous monad, defining iteration by taking
  least fixpoints determines a complete Elgot monad structure.
\end{thm}
\noindent This result is to be expected in the light of analogous
facts known for Bloom and \'{E}sik's \emph{$\omega$-continuous
  theories}~\cite[Theorem~8.2.15, Exercise~8.2.17]{BloomEsik93}.
\begin{proof}
Let $\BBT$ be an $\omega$-continuous monad, and let $f^\istar$ be the least fixpoint of~\eqref{eq:lfp}.

Let us verify the axioms of complete Elgot monads one by one. To that end we employ the following \emph{uniformity rule} for least fixpoints of continuous functionals~\cite{SimpsonPlotkin00}:
\begin{align}\label{eq:unif}
\infer{U(\mu F) = \mu G}{UF = GU~~&~~U(\bot) = \bot}
\end{align}
Moreover, in several places below we use \emph{fixpoint induction} to show
that $f^\istar\sqsubseteq g$ for given $f:A\to T(B+A)$ and
$g:A\to TB$: Since $f^\istar$ is a supremum of the chain
$(F^i(\bot))_{i\in\Nat}$ where $F: (A \to TB) \to (A\to TB)$ is the
functional defined by
\begin{equation*}
  F(h) = [\eta,h]^\klstar f^\istar,
\end{equation*}
$f^\istar\sqsubseteq g$ follows as soon as we prove
$F^i(\bot)\sqsubseteq g$ for all $i\in\Nat$, a claim that we typically
prove by induction on $i$. The induction base $i=0$ is always trivial,
so we consistently do only the inductive step. More generally, we can
apply the same principle to conclude $\alpha(f^\istar)\sqsubseteq r$,
for given $r:C\to TD$ and a function $\alpha:(A\to TB)\to (C\to TD)$,
from $\alpha(F^i(\bot))\sqsubseteq r$ for all $i$, provided that
$\alpha$ is $\omega$-continuous, a condition that will always be
immediate from our assumptions. In the more general case, we need to
pay attention to the base case, typically be ensuring that $\alpha$
preserves $\bot$.
\begin{citemize}
  \item \emph{Fixpoint.} This holds by definition.
  \item \emph{Naturality.} In~\eqref{eq:unif} take $F(u)=[\eta,u]^\klstar f$, $G(u)=[\eta,u]^\klstar [(T\inl) g,\eta\inr]^\klstar f$ and $U(u)=g^\klstar u$. By definition, $U(\bot)=\bot$, $\mu F=f^\istar$, $\mu G=([(T\inl) g,\eta\inr]^\klstar\comp f)^\istar$. Then we have
\begin{align*}
 U(F(u))=g^\klstar [\eta,u]^\klstar f = [\eta,g^\klstar u]^\klstar [(T\inl) g,\eta\inr]^\klstar f = G(U(u)).
\end{align*}
Therefore, by~\eqref{eq:unif}, $g^{\klstar} f^{\istar}=U(\mu F)=\mu G = ([(T \inl) g, \eta\inr]^{\klstar} f)^{\istar}$.
 \item \emph{Codiagonal.} Recall that we are claiming that
  $$(T[\id, \inr]
  g)^{\istar} = (g^{\istar})^{\istar}$$ with $g:A\to T((B+A)+A)$. We
  first show that $g^{\istar\istar}$ is a fixpoint of the functional
  defining the left-hand side as a least fixpoint, thus proving
  $\sqsubseteq$. That is, we have to show that
  \begin{equation}
  \label{eq:codiag-base-fp}
  g^{\istar\istar} = [\eta,g^{\istar\istar}]^\klstar T[\id, \inr]g.
  \end{equation}
  We proceed as follows:
  \begin{flalign*}
  &&g^{\istar\istar} & = [\eta,g^{\istar\istar}]^\klstar g^\istar &&\by{fixpoint}\\
  &&& = [\eta,g^{\istar\istar}]^\klstar[\eta,g^\istar]^\klstar g &&\by{fixpoint}\\
  &&& = [[\eta,g^{\istar\istar}],[\eta,g^{\istar\istar}]^\klstar g^\istar]^\klstar g\\
  &&& = [[\eta,g^{\istar\istar}],g^{\istar\istar}]^\klstar g &&\by{fixpoint}\\
  &&& = [[\eta,g^{\istar\istar}], [\eta,g^{\istar\istar}]\inr]^\klstar g \\
  &&& = [\eta, g^{\istar\istar}]^\klstar\,T[\id, \inr]g.
  \end{flalign*}
  For the converse inequality, we use fixpoint induction.  So let
  $f\sqsubseteq (T[\id, \inr] \comp\, g)^{\istar}$. We have to show
  \begin{equation*}
  [\eta,f]^\klstar g^\istar\sqsubseteq(T[\id, \inr]
  g)^{\istar}.
  \end{equation*}
  We establish this by a second fixpoint induction on the occurrence
  of $g^\istar$ on the left hand side. For the base case, just
  recall that Kleisli composition from the left preserves $\bot$. For
  the inductive step, assume that
  $[\eta,f]^\klstar h \sqsubseteq(T[\id, \inr] g)^{\istar}$, with
  $h:A\to T(B+A)$; we have to show that
  \begin{equation*}
  [\eta,f]^\klstar [\eta,h]^\klstar g\sqsubseteq(T[\id, \inr]
  g)^{\istar}.
  \end{equation*}
  We calculate as follows:
  \begin{flalign*}
  &&[\eta,f]^\klstar [\eta,h]^\klstar g =&\; [[\eta,f],[\eta,f]^\klstar h]^\klstar g\\
  && \sqsubseteq&\;[[\eta,f],(T[\id, \inr] g)^{\istar}]^\klstar g & \by{inner IH}\\
  && \sqsubseteq&\; [[\eta,(T[\id, \inr] g)^{\istar}],(T[\id, \inr] g)^{\istar}]^\klstar g& \by{outer IH}\\
  && =&\; [\eta,(T[\id, \inr] g)^{\istar}]T[\id, \inr] g\\
  && =&\; (T[\id, \inr] g)^{\istar}. & \by{fixpoint}
  \end{flalign*}

  \item \emph{Uniformity.} Let $f : A \to T(X+A)$, $g : B \to T(X+B)$, $h:B\to A$ and assume that $f\comp h = T(\id+h)\comp g$.
Let us define $G(u) = [\eta, u]^{\klstar} g$, $F(u) = [\eta,
	u]^{\klstar} f$, and $U(u) = u\comp h$. Then $U(\bot) = \bot$ and
	\begin{flalign*}
		&&UF(u) =&\, [\eta, u]^{\klstar}  f\comp  h &&\\
		&&=&\, [\eta, u]^{\klstar}  T(\id+h)\comp g &&\\
		&&=&\, [\eta, u\comp h]^{\klstar} \comp g &&\\
 		&&=&\,GU(u).
	\end{flalign*}
	Therefore by~\eqref{eq:unif}, $f^{\istar}\comp h=U(\mu F) = \mu G = g^{\istar}$.
\end{citemize}
To prove compatibility of strength and iteration, we proceed by first showing
\[
  ((T\dist) \tau (\id \times f))^\istar \sqsubseteq \tau (\id \times f^{\istar}).
\]
First observe that, for any $g: A \to TB$,
\begin{align}\label{eq:str_dist}
\vcenter{\vbox{
\xymatrix@R25pt@C58pt@M6pt{
C\times(B+A)
	\ar[d]_-{\id \times [\eta, g]}
	\ar@<.8ex>[r]^-{\dist}
&
C \times B + C \times A
	\ar@<.5ex>[l]^-{~\dist^{\mone}}
	\ar[d]^-{[\eta, \tau (\id \times g)]}
\\
C \times TB \ar[r]^-{\tau}
&
T(C \times B).
}
}}%
\end{align}
This is easily checked componentwise starting from $C \times B + C \times A$ and using the fact that by definition
$\dist^{\mone} = [\id \times \inl, \id \times \inr]$.
Then we have
\begin{flalign*}
	&&&~\tau (\id \times f^{\istar}) \\
	&&=&~\tau (\id \times [\eta, f^{\istar}]^{\klstar} f) \\
	&&=&~\tau (\id \times [\eta, f^{\istar}]^{\klstar}) (\id \times f) \\
	&&=&~(\tau (\id \times [\eta, f^{\istar}]))^{\klstar} \tau (\id \times f)         &\by{\sc str$_4$} \\
	&&=&~([\eta, \tau (\id \times f^{\istar})] \dist)^{\klstar} \tau (\id \times f)   &\by{\ref{eq:str_dist}} \\
	&&=&~[\eta, \tau (\id \times f^{\istar})]^{\klstar} (T\dist) \tau (\id \times f).
\end{flalign*}
Therefore,  $\tau (\id \times f^{\istar})$ is a fixed point of the functional defining $((T\dist) \tau (\id \times f))^{\istar}$ as a least fixpoint and the inequality above holds.
The converse inequality,
\[
  \tau (\id \times f^{\istar}) \sqsubseteq ((T\dist \tau) (\id \times f))^\istar,
\]
is shown by fixpoint induction. For the base case, we calculate the
left hand side:
\[
\tau(\id\times\bot)=\tau \brks{\fst, \bot \snd} = \tau \brks{\fst, \bot \fst} = \tau \brks{\id, \bot}
  \fst = \bot \fst = \bot.
\]
\noindent
For the inductive step, assume that
$\tau (\id \times g) \sqsubseteq (T\dist \tau (\id \times f))^{\istar}.$ We can then calculate
\begin{flalign*}
	&&&~\tau (\id \times [\eta, g]^{\klstar} f) \\
	&&=&~\tau (\id \times [\eta, g]^{\klstar}) (\id \times f) \\
	&&=&~(\tau (\id \times [\eta, g]))^{\klstar} \tau (\id \times f) &\by{\sc str$_4$} \\
	&&=&~([\eta, \tau (\id \times g)] \dist)^{\klstar} \tau (\id \times f) &\by{\ref{eq:str_dist}} \\
	&&\sqsubseteq&~([\eta, (T\dist \tau (\id \times f))^{\istar}] \dist)^{\klstar} \tau (\id \times f) \\
	&&=&~[\eta, (T\dist \tau (\id \times f))^{\istar}]^{\klstar} T\dist \tau (\id \times f) \\
	&&=&~(T\dist \tau (\id \times f))^{\istar}
\end{flalign*}
which completes the proof.
\end{proof}
\noindent Every complete Elgot monad $\BBT$ can express
\emph{unproductive divergence} as the generic effect
\begin{displaymath}
  \bot_{X,Y}=\bigl(X\xto{\;\eta\inr\;} T(Y+X) \bigr)^\istar:X\to TY.
\end{displaymath}
This computation never produces any effects, i.e.\ behaves like a
deadlock. If $\BBT$ is $\omega$-continuous, then unproductive
divergence coincides with the least element of $\Hom(X,TY)$, for which
reason we use the same symbol $\bot_X$, but in general, there is no
ordering in which unproductive divergence could be a least element. 
\begin{lem}\label{lem:divergence}
  Unproductive divergence is \emph{constant}, i.e.~for $f:Z\to X$, we
  have $\bot_{X,Y}f=\bot_{Z,Y}$, and \emph{coconstant}, i.e.~for
  $h:Y\to TW$ we have $h^\klstar\bot_{X,Y}=\bot_{X,W}$.
\end{lem}
\begin{proof}
\emph{Constancy:} We have to show
$(\eta_{Y+X}\inr)^\istar f=(\eta_{Y+Z}\inr)^\istar$. By uniformity,
it suffices to show that
$\eta_{Y+X}\inr f = T(\id + f)\comp \eta_{Y+Z}\inr$. We
calculate the right-hand side:
\begin{flalign*}
  &&T(\id + f)\eta_{Y+Z}\inr &\,= \eta_{Y+X}(\id+f)\inr & \by{naturality of $\eta$}\\
  &&&\, = \eta_{Y+X}\inr f.       
\end{flalign*} 
\emph{Coconstancy:} We have 
\begin{flalign*} 
&&h^\klstar\bot_{X,Y} &\, = h^\klstar(\eta \inr_{Y+X})^\istar \\
&&        &\, = ([T\inl h,\eta_{W+X}\inr]^\klstar \eta_{Y+X}\inr)^\istar &\by{naturality of $\istar$}\\
&&        &\, = (\eta_{W+X}\inr)^\istar = \bot_{X,W}.&\qedhere
\end{flalign*}
\noqed\end{proof}

\noindent The following lemma shows that there can be only one
unproductive divergence:
\begin{lem}\label{lem:one-divergence}
  Let $e:X\to T(Y+X)$ have the form $e=\eta\inr u$ for
  $u:X\to X$. Then $e^\istar=\bot_{X,Y}$.
\end{lem}
\begin{proof}
  By constancy, $\bot_{X,Y}=\bot_{1,Y}\bang_X$, so we are to show
  $(\eta_{Y+1}\inr)^\istar\bang_X=e^\istar$. By uniformity, it suffices to
  show $\eta_{Y+1}\inr\bang_X=T(\id+\bang_X)\eta\inr u$, which is immediate
  by naturality of $\eta$.
\end{proof}
\section{The Coinductive Generalized Resumption Transformer}
\label{sec:resumptions}
\noindent
We proceed to recall the definition of the coinductive generalized
resumption transformer~\cite{PirogGibbons13}. One of our main results
will be stability of the class of complete Elgot monads under this
construction (Theorem~\ref{thm:cr_elg}). In the remainder of the
paper, we work with the following set of standing assumptions.

\begin{assn}\label{asm:fp}
  We fix 
  \begin{itemize}
  \item a distributive category $\BC$;
  \item a strong functor $\SigF:\BC\to\BC$ with strength $\rho$; 
  \item a strong monad $\BBT$ on $\BC$ with strength $\tau$;
  \end{itemize}
  and assume that \emph{the final coalgebra
    $\nu\gamma.\,T(X+\SigF\gamma)$ of $T(X+\SigF)$ exists for all
    $X\in|\BC|$}.
\end{assn}
\noindent As indicated in the introduction, we think of $\Sigma$ as
specifying a signature of communication actions, and of $T$ as
encapsulating a notion of side-effect. 

We can then define a functor $\TF$ whose action on objects is given by
\begin{equation*}
\TF X=\nu\gamma.\,T(X+\SigF\gamma). %
\end{equation*}
Intuitively, $\TF X$ is a type of possibly non-terminating computation
trees, in which each step triggers a computational effect specified by
$T$, and then either terminates with a result in $X$ or branches
according to an operation from the signature represented by $\Sigma$,
with arguments being again computation trees.

\begin{rem}
  There are two broad classes of models satisfying
  Assumption~\ref{asm:fp}:
  \begin{itemize}
  \item $\BC$ is a locally presentable category and $\BBT$ is ranked; or
  \item $\BC$ is $\Cpo$-enriched and has colimits of $\omega$-chains,
  and $T$ is $\omega$-continuous (Remark~\ref{rem:cpo_cat}).
  \end{itemize}
  Satisfaction of Assumption~\ref{asm:fp} in the first case follows
  from the fact that categories of coalgebras for accessible functors
  over locally presentable categories are again locally presentable,
  in particular complete~\cite[Exercise~2.j, Chapter
  2]{AdamekRosicky94}.  This covers most of the interesting choices of
  base categories, such as $\Set$, $\Cpo$, various categories of
  predomains, and presheaf categories, as well as almost all
  computationally relevant monads~\cite{Moggi91,PlotkinPower02}.
  The fact that Assumption~\ref{asm:fp} is satisfied in the second
  case follows from Barr's work on algebraically compact
  functors~\cite[Theorem 5.4]{Barr92}, which also implies that the
  greatest fixed points of interest coincide with least fixed
  points. One example covered by the second clause but not by the
  first one is the continuation monad $TX=(X\to R)\to R$ on $\Cpo$,
  provided that~$R$ has a least element.
\end{rem}
Let
\begin{equation*}
\out_X:\TF X\to T(X+\SigF\TF X)
\end{equation*}
be the final coalgebra structure, and let $\coit(g):Y\to\TF X$ denote
the final morphism induced by a coalgebra $g:Y\to T(X+\SigF Y)$:
\begin{equation*}
  \xymatrix@R25pt@M6pt@C=8em{Y\ar[d]_g \ar[r]^{\coit(g)} & \TF X \ar[d]^{\out_X}\\
  T(X+\SigF Y)\ar[r]^{T(X+\SigF\coit(g))} & T(X+\SigF\TF X).}
\end{equation*}
Intuitively, $\coit(g)$ encapsulates (in $\TF X$) a computation tree
that begins by executing~$g$, terminates in a leaf of type $X$ if $g$
does, and otherwise (co-)recursively continues to execute~$g$, forming
a new tree node for each recursive call. By Lambek's lemma, $\out_X$
is an isomorphism. As we see below, it is also natural in $X$. Thus,
$T$ maps into $\TF$ via
\begin{equation}\label{eq:ext}
\ext=\bigl(T~\xto{~~T\inl~~} T(\Id + \SigF\TF)\xto{~~\tuo~~} \TF\bigr).
\end{equation}
We record explicitly that $\TF$ is a strong monad:
\begin{thm}\label{lem:kl_dec} Given a monad $\BBT$, $\TF$ is the functorial part of a monad $\BBTF$, with the strong monad structure denoted $\tau^\nu$, $\eta^\nu$, and $(-)^\kklstar$ (for Kleisli star) and characterized by the following properties.
\begin{enumerate}
\item The unit $\eta^{\nu}:X\to \TF X$ is defined by
  $\out\comp\,\eta^\nu=\eta\comp\inl$ (i.e.\ $\eta^\nu=\tuo\comp\,\eta\comp\inl$).
 \item Given $f:X\to\TF Y$, the Kleisli lifting \mbox{$f^\kklstar:\TF X\to\TF Y$} is the unique solution of the
	 equation
\begin{align}\label{eq:kl_def}
\out\comp f^\kklstar=[\out \comp f,\eta\comp\inr\comp \SigF f^\kklstar]^\klstar\comp\out.
\end{align}
 \item Given $f:X\to\TF Y$, let $g=[f,\eta^\nu]:X+Y\to\TF Y$; then $g^\kklstar$ 
is a final morphism from $(\TF (X+Y), [T(\id+\SigF\TF\inr)\out g,\eta\inr]^\klstar\out:\TF (X+Y)\to T(Y+{\Sigma \TF (X+Y)}))$ to
$(\TF Y,\out_Y)$,	i.e.\
\begin{align}\label{eq:cor_kl}
g^\kklstar=\coit\bigl([T(\id+\SigF\TF\inr)\out g,\eta\inr]^\klstar\out\bigr).
\end{align}
 \item The strength $\tau^\nu:X\times\TF Y\to\TF(X\times Y)$ is the unique solution of
\begin{align}\label{eq:str_def}
\out\comp\, \tau^\nu=T(\id+\SigF\tau^\nu)\comp (T\delta)\comp\tau\comp(\id\times\out)
\end{align}
with $\delta:X\times (Y+\SigF Z)\to X\times Y+\SigF(X\times Z)$ being the transformation $\delta=(\id+\rho)\comp\dist$ where $\rho_{X,Y}:X\times\Sigma Y\to\Sigma(X\times Y)$ is the strength of $\Sigma$.
\end{enumerate}
\end{thm}
\noindent
This justifies calling $\BBTF$ the \emph{coinductive generalized
  resumption monad} (over $\BBT$). The proof of Theorem~\ref{lem:kl_dec}
is facilitated by the fact that $T(X+\SigF)$ can be shown to
be a \emph{parametrized monad}, which implies that $\BBTF$ is a
monad~\cite[Theorems~3.7 and~3.9]{Uustalu03}. Alternatively, the fact
that $\BBTF$ is a monad can be read off directly from the results
of~\cite{PirogGibbons13}. What is new here is that we show that
$\BBTF$ is, in fact, strong, and hence supports an interpretation of
Moggi's computational metalanguage~\cite{Moggi91}. This amounts
to showing that the strength defined in the last item satisfies the
requisite laws in p.~\pageref{page:str}. One preliminary fact of potentially
independent interest used in the proof of these laws is
\begin{lem}\label{lem:final-functor}
  The object assignment $X\mapsto\TF X$ extends to a functor $\TF$,
  and $\out:\TF\to T(\Id+\SigF\TF)$ then becomes a natural
  transformation. For any functor $G:\BB\to\BC$,
  $\out_G:\TF G\to T(G+\Sigma\TF G)$ is a final
  $T(G+\Sigma(-))$-coalgebra in $[\BB,\BC]$.
\end{lem}
\begin{proof}
  Functoriality follows from the fact that, as stated in
  Theorem~\ref{lem:kl_dec} and proved independently from this lemma in
  the proof of the theorem, $\TF$ carries a monad structure. That is,
  $\TF f=(\eta^\nu f)^\kklstar$, so by the description of $\kklstar$
  we have
\begin{align*}
  \out\TF f & = [\out\eta^\nu f,\eta\inr\SigF\TF f]^\klstar\out\\
  & = [\eta \inl f,\eta\inr \SigF\TF f]^\klstar\out \\
  & = T[\inl f,\inr \SigF\TF f]\out \\
  & = T(f+ \SigF\TF f)\out,
\end{align*}
i.e.\ $\out$ is natural.

To show finality, let $\beta:F\to T(G+\SigF F)$ be a natural
transformation. We define the universal arrow $f:F\to \TF G$ componentwise by the equation
\begin{equation*}
  \out f_X = T(\id+\SigF f_X)\beta_X
\end{equation*}
using finality of the components $\out_{GX}:\TF GX\to T(GX+\SigF\TF GX)$. We have to show
that~$f$ is natural (uniqueness is clear). So let $g:X\to Y$; we have
to show $f_Y Fg=(\TF Gg) f_X$. Note that we have a
$T(GY+F)$-coalgebra
\begin{equation*}
  FX\xto{~~\beta_X~~} T(GX+\SigF FX)\xto{~~T(Gg+\id)~~} T(GY+\SigF FX);
\end{equation*}
we show that both $f_Y Fg$ and $(\TF Gg)f_X$ are coalgebra
morphisms into $\TF GY$ for $T(Gg+\id)\beta_X$. On the one hand, we have
\begin{flalign*}
&&  \out f_YFg =&\; T(\id+\SigF f_Y)\beta_Y Fg 	& \by{definition of $f_Y$} \\
&& =&\; T(\id+\SigF f_Y)T(Gg+\SigF Fg)\beta_X 	& \by{naturality of $\beta$}\\
&& =&\; T(Gg+\SigF (f_YFg))\beta_X.
\end{flalign*}
On the other hand,
\begin{flalign*}
&& \out (\TF Gg)f_X =&\; T(Gg+\SigF\TF Gg)\out f_X & \by{naturality of $\out$}\\
&& =&\; T(Gg+\SigF\TF Gg)T(\id+\SigF f_X)\beta_X & \by{definition of $f_X$}\\
&& =&\; T(Gg+\SigF((\TF Gg) f_X))\beta_X.
\end{flalign*}
Using the fact that there is unique morphism from a given coalgebra to the final one,
 we conclude that indeed $f_Y Fg=(\TF Gg) f_X$.
\end{proof}
\begin{proof}[Proof of Theorem~\ref{lem:kl_dec}.]
Since $T(X+\SigF )$ extends to a \emph{parametrized monad}, as shown in~\cite[Theorems~3.7 and~3.9]{Uustalu03},
$\BBTF$ is a monad whose Kleisli lifting is uniquely characterized by~\eqref{eq:kl_def}. What is missing is
to show that $\BBTF$ is a strong monad, as we need here. 
Let $g$ be defined as in clause (3) of the theorem, and let us first show~\eqref{eq:cor_kl}.
By definition, $g^\kklstar$ is the unique morphism making the following diagram commute:
\begin{align*}
\xymatrix@R25pt@C90pt@M6pt{
\TF(X+Y) \ar[d]_-{\out}\ar[r]^-{g^\kklstar}  &  \TF Y\ar[d]^-{\out} \\
T(X+Y+\SigF\TF(X+Y)) \ \ar[r]^-{[\out g,\,\eta\inr \SigF{g^\kklstar}]^\klstar} & T(Y+\SigF\TF Y)
}
\hspace{10ex}
\end{align*}
We then have on the one hand,
\begin{flalign*}
&&\out g^\kklstar=&\;[\out [f,\eta^\nu],\,\eta\inr\SigF{g^\kklstar}]^\klstar\out& \by{definition of $\argument^\kklstar$}\\
&&=&\;[ [\out f,\out\eta^\nu],\eta\inr\SigF{g^\kklstar}]^\klstar\out\\
&&=&\;[ [\out f,\eta\inl],\eta\inr\SigF{g^\kklstar}]^\klstar\out& \by{definition of $\eta^\nu$}
\end{flalign*}
and also  on the other hand,
\begin{flalign*}
&&T(\id+\SigF{g^\kklstar})&\;[T(\id+\SigF\TF\inr)\out g,\eta\inr]^\klstar\out&\\
&&=&\;[ T(\id+\SigF(g^\kklstar\TF\inr)) [\out f,\out \eta^{\nu}],\eta\inr \SigF g^{\kklstar}]^\klstar\out\\
&&=&\;[ [\out f,\eta\inl],\eta\inr \SigF{g^\kklstar}]^\klstar\out,&
\end{flalign*}
i.e.\, indeed $g^\kklstar$ satisfies the characteristic property of the final morphism~\eqref{eq:cor_kl}.

We proceed to prove that $\BBTF$ is strong. We define the strength $\tau^\nu$ as the unique final coalgebra morphism shown in the following diagram:
\begin{align*}
\xymatrix@R25pt@C68pt@M6pt{
X\times\TF Y \ar[d]_-{\tau^\nu }\ar[r]^-{(T\delta)\tau(\id\times\out)}  &  T(X\times Y+\SigF(X\times\TF Y))\ar[d]^-{T(\id+\SigF\tau^\nu)} \\
\TF(X\times Y) \ \ar[r]^-{\out} & T(X\times Y+\SigF\TF(X\times Y))
}
\end{align*}
That is, $\tau^\nu$ is the unique solution of equation $\out \tau^\nu=T(\id+\SigF\tau^\nu) (T\delta)\tau(\id\times\out)$. By Lemma~\ref{lem:final-functor}, $\tau^\nu$ is a composite of natural transformations and hence itself natural. Let us check the axioms of strength from p.~\pageref{page:str}.
\begin{citemize}
  \item {\sc(str$_1$)} The identity $\snd=(\TF\snd)\tau^\nu$ follows from $\TF\brks{\bang,\id}\snd=\tau^\nu$ where $!$ is a suitable terminal morphism $X\to 1$, since obviously $\snd=(\TF\snd)\TF\brks{\bang,\id}\snd$. Since $\tau^\nu$ is uniquely defined by the corresponding characteristic identity~\eqref{eq:str_def}, it suffices to show that $\TF\brks{\bang,\id}\snd$ satisfies the same identity. Indeed,
\begin{flalign*}
&&T(\id +\,& \SigF (\TF\brks{\bang,\id}\snd)) (T\delta)\tau(\id\times\out)\\
&&=~&T(\brks{\bang,\id}\snd+\,\SigF(\TF\brks{\bang,\id}\snd))(T\delta)\tau(\id\times\out)\\
&&=~&T(\brks{\bang,\id}+\,\SigF\TF\brks{\bang,\id})T(\snd+\SigF\snd)(T\delta)\tau(\id\times\out)\\
&&=~&T(\brks{\bang,\id}+\,\SigF\TF\brks{\bang,\id})T(\snd+\snd)(T\dist)\tau(\id\times\out)&\by{{\sc str$_1$} for $\rho$} \\
&&=~&T(\brks{\bang,\id}+\SigF\TF\brks{\bang,\id}) (T\snd)\tau(\id\times\out)&\by{defintion of $\dist$}\\
&&=~&T(\brks{\bang,\id}+\SigF\TF\brks{\bang,\id})\out\snd&\by{{\sc str$_1$} for $\tau$} \\
&&=~&\out(\TF\brks{\bang,\id})\snd.&\by{naturality of $\tau$}
\end{flalign*}
\item {\sc(str$_2$)} In order to prove that $(\TF\assoc)\tau^\nu =
  \tau^\nu (\id\times\tau^\nu)\assoc:(X\times Y)\times\TF Z\to
  \TF((X\times Y)\times Z)$,
  it suffices to show that $(\TF\assoc^{\mone})\tau^\nu (\id\times\tau^\nu)\assoc$
  satisfies the characteristic identity~\eqref{eq:str_def} for $\tau^\nu$, i.e.
  \begin{multline*}
  \out(\TF\assoc^{\mone})\tau^\nu(\id\times\tau^\nu)\assoc
  =T(\id+\SigF((\TF\assoc^{\mone})\tau^\nu(\id\times\tau^\nu)\assoc))(T\delta)\tau(\id\times\out).
  \end{multline*}
  We calculate, transforming the left hand side,
  \begin{flalign*}
  && \out(\TF&\assoc^{\mone})\tau^\nu(\id\times\tau^\nu)\assoc \\
  &&=\;& T(\assoc^{\mone}+\SigF\TF\assoc^{\mone})\out\tau^\nu(\id\times\tau^\nu)\assoc&\by{naturality of $\out$}\\
  &&=\;& T(\assoc^{\mone}+\SigF\TF\assoc^{\mone})\\
  &&&~~T(\id+\SigF\tau^\nu)(T\delta)\tau(\id\times\out)(\id\times\tau^\nu)\assoc &\qquad\by{definition of $\tau^\nu$}\\
  &&=\;& T(\assoc^{\mone}+\SigF((\TF\assoc^{\mone})\tau^\nu))\\
  &&&~~(T\delta)\tau(\id\times T(\id+\SigF\tau^\nu)(T\delta)\tau(\id\times\out))\assoc&\by{definition of $\tau^\nu$}\\
  &&=\;& T(\assoc^{\mone}+\SigF((\TF\assoc^{\mone})\tau^\nu))\\
  &&&~~(T\delta) T(\id\times(\id+\SigF\tau^\nu)\delta)\tau(\id\times\tau(\id\times\out))\assoc&\by{naturality of $\tau$}.
\intertext{
and then continue to transform the last part of the term:}
  && \tau(\id\times\,\tau&(\id\times\out))\assoc\\
  &&=\;& \tau(\id\times\tau)(\id\times(\id\times\out))\assoc\\
  &&=\;& \tau(\id\times\tau)\assoc((\id\times\id)\times\out) &\by{naturality of $\assoc$}\\
  &&=\;& (T\assoc)\tau(\id\times\out) & \by{{\sc str$_2$} for $\tau$} 
  \end{flalign*}
  (contracting a product of identities into an identity in the last
  step).  Summing up, it remains to show that
  \begin{align*}
  &T(\assoc^{\mone}+\SigF((\TF\assoc^{\mone})\tau^\nu))(T\delta) T(\id\times(\id+\SigF\tau^\nu)\delta)(T\assoc)\tau(\id\times\out)\\
  =\; &T(\id+\SigF((\TF\assoc^{\mone})\tau^\nu(\id\times\tau^\nu)\assoc))(T\delta)\tau(\id\times\out),
  \end{align*}
  which we reduce, removing $\tau(\id\times\out)$ and $T$, multiplying
  from the left with $\assoc+\Sigma\TF\assoc$, and removing
  $\tau^\nu$ on the left, to
  \begin{equation*}
  \delta(\id\times(\id+\SigF\tau^\nu)\delta)\assoc =
  (\assoc+(\SigF(\id\times\tau^\nu)\assoc))\delta.
  \end{equation*}
  For the latter we calculate
  \begin{flalign*}
 && \delta(\id&\times(\id+\SigF\tau^\nu)\delta)\assoc  \\
 && & = \delta(\id\times(\id+\SigF\tau^\nu))(\id\times\delta)\assoc \\
 && & = (\id\times\id+\SigF(\id\times\tau^\nu))\delta(\id\times\delta)\assoc  & \by{naturality of $\delta$}\\
 && & = (\id+\SigF(\id\times\tau^\nu))(\id+\rho)\dist(\id\times(\id+\rho)\dist)\assoc  & \by{definition of $\delta$}\\
 && & = (\id+\SigF(\id\times\tau^\nu))(\id+\rho(\id\times\rho))\dist(\id\times\dist)\assoc  & \by{naturality of $\dist$}\\
 && & = (\id+\SigF(\id\times\tau^\nu))(\id+\rho(\id\times\rho))(\assoc+\assoc)\dist  & \by{\ref{eq:assoc_dist}}\\
 && & = (\id+\SigF(\id\times\tau^\nu))(\assoc+\SigF\assoc\rho)\dist&\by{{\sc str$_2$} for $\rho$} \\
 && & = (\id+\SigF(\id\times\tau^\nu))(\assoc+\SigF\assoc)\delta& \by{definition of $\delta$}\\
 && & = (\assoc+\SigF((\id\times\tau^\nu)\assoc))\delta.
  \end{flalign*}
  Here, we use the obvious coherence property
  \begin{equation}\label{eq:assoc_dist}
  \dist(\id\times\dist)\assoc = (\assoc+\assoc)\dist.
  \end{equation}
\item {\sc(str$_3$)} In order to obtain the identity $\tau^\nu(\id\times\eta^\nu)=\eta^\nu$, we show that the left hand side satisfies the characteristic equation for $\eta^\nu$, i.e.\ $\out\tau^\nu (\id\times\eta^\nu)=\eta\inl$. Indeed,
\begin{flalign*}
&&\out\tau^\nu (\id\times\eta^\nu) =\;& T(\id+\SigF\tau^\nu) (T\delta)\tau(\id\times\out) (\id\times\eta^\nu) & \by{definition of $\tau^\nu$} \\
&&=\;& T(\id+\SigF\tau^\nu) (T\delta)\tau(\id\times\eta)(\id\times\inl) & \by{definition of $\eta^\nu$} \\
&&=\;& T(\id+\SigF\tau^\nu) (T\delta)\eta(\id\times\inl) &\by{{\sc str$_3$} for $\tau$} \\
&&=\;& T(\id+\SigF\tau^\nu) (T\delta)T(\id\times\inl)\eta & \by{naturality of $\eta$} \\
&&=\;& T(\id+\SigF\tau^\nu) (T\inl)\eta &\\
&&=\;& \eta\inl. &
\end{flalign*}
\item {\sc(str$_4$)} Given $f:X\to\TF Z$, we show that $(\tau^\nu(\id\times f))^\kklstar\tau^\nu=\tau^\nu (\id\times f^\kklstar)$. Let $g=[f,\eta^\nu]$ and let us show first that $(\tau^\nu(\id\times g))^\kklstar\tau^\nu=\tau^\nu (\id\times g^\kklstar)$. This implies the identity for $f$ as follows:
\begin{flalign*}
&&(\tau^\nu(\id\times f))^\kklstar\tau^\nu=\;&(\tau^\nu(\id\times g)(\id\times \inl))^\kklstar\tau^\nu &\\
&&=\;&(\tau^\nu(\id\times g))^\kklstar \TF(\id\times \inl)\tau^\nu & \\
&&=\;&(\tau^\nu(\id\times g))^\kklstar \tau^\nu(\id\times \TF\inl) & \by{naturality of $\tau^\nu$} \\
&&=\;&\tau^\nu (\id\times g^\kklstar)(\id\times \TF\inl) &\by{{\sc str$_4$} for $g$ and $\tau^\nu$} \\
&&=\;&\tau^\nu (\id\times g^\kklstar\TF\inl) &\\
&&=\;&\tau^\nu (\id\times f^\kklstar). &
\end{flalign*}
As we have shown above both $g^\kklstar$ and $\tau^\nu$ are final morphisms from suitable coalgebras. By composing the corresponding commutative squares we obtain the following diagram:
\begin{align*}
\xymatrix@L=8pt@R20pt@C58pt@M6pt{
X\times\TF Y\ar[r]^->>>>>>>>>>>>>>>{\id\times[T(\id+\SigF\TF\inr)\out g,\eta\inr]^\klstar\out}\ar[dd]_-{\id\times g^\kklstar} & X\times T(Z+\SigF\TF Y)\ar[r]+<-60pt,0pt>^-{(T\delta)\tau}\ar[d]+<0pt,0pt>^-{\id\times T(\id+\SigF{g^\kklstar})} &
\save[]+<+3mm,0mm>*{T(X\times Z+\SigF(X\times\TF Y))}\ar@/^2.2pc/[]+<+10pt,-12pt>;[dddl]+<+60pt,-20pt>^->>>>>>>>>>>>>>>{{T(\id+\SigF(\tau^\nu(\id\times g^\kklstar)))}} \restore\\
 &  \save[]+<0mm,-5mm>*{X\times T(Z+\TF Z)} \ar[dd]+<0mm,+8mm>^-{(T\delta)\tau} \restore \\
X\times\TF Z\ar[dd]_-{\tau^\nu}\ar[ur]+<-45pt,-5mm>^-{\id\times\out}\ar[dr]+<-65pt,+5mm>_-{(T\delta)\tau(\id\times\out)} & \\
 & \save[]+<0mm,+5mm>*{T(X\times Z+\SigF(X\times\TF Z))} \ar[d]+<0mm,5mm>^-{T(\id+\SigF\tau^\nu)} \restore \\
\TF(X\times Z)\ar[r]_-{\out} & \TF(X\times Z +\SigF\TF(X\times Z))
}
\end{align*}
from which we conclude that
\begin{align*}
\tau^\nu (\id\times g^\kklstar) = \coit\bigl((T\delta)\tau(\id\times[h,\eta\inr]^\klstar\out)\bigr)
\end{align*}
where $h$ denotes $T(\id+\SigF\TF\inr)\out g$.

We will be done once we show that also $(\tau^\nu(\id\times g))^\kklstar\tau^\nu$ is a morphism from the same coalgebra to the final one, i.e.\ the identity
\begin{align}\label{eq:str4}
T(\id+\SigF((\tau^\nu(\id\times g))^\kklstar\tau^\nu))(T\delta)\tau(\id\times[h ,\eta\inr]^\klstar\out)=\out(\tau^\nu(\id\times g))^\kklstar\tau^\nu.
\end{align}
Let us show that the following diagram commutes:
\begin{align}\label{eq:str4_dia}
\vcenter{\vbox{
\xymatrix@R25pt@C58pt@M6pt{
T(X\times (Y+\SigF(\TF Y)))\ar[d]_-{(\tau(\id\times [h,\eta\inr]))^\klstar} \ar[r]^-{T\delta} & T(X\times Y+ \SigF(X\times\TF Y))\ar[d]^-{[(T\delta)\tau(\id\times h),\eta\inr]^\klstar}\\
T(X\times (Z+\SigF(\TF Y)))\ar[r]^-{T\delta} & T(X\times Z+ \SigF(X\times\TF Y))
}}}
\end{align}
Indeed,
\begin{flalign*}
&&[(T\delta)&\tau(\id\times h),\eta\inr]^\klstar (T\delta)\\
&&=\;&([(T\delta)\tau(\id\times h),\eta\inr]\delta)^\klstar\\
&&=\;&([(T\delta)\tau(\id\times h),\eta\inr](\id+\rho) \dist)^\klstar&\by{definiton of~$\delta$}\\
&&=\;&([(T\delta)\tau(\id\times h),\eta\inr\rho]\dist)^\klstar\\
&&=\;&([(T\delta)\tau(\id\times h),\eta\,\delta(\id\times\inr)]\dist)^\klstar&\by{definiton of~$\delta$}\\
&&=\;&([(T\delta)\tau(\id\times h),(T\delta)\tau(\id\times\eta)(\id\times\inr)]\dist)^\klstar&\by{{\sc str$_3$} for~$\tau$}\\
&&=\;&((T\delta)\tau[\id\times h,\id\times\eta\inr]\dist)^\klstar\\
&&=\;&(T\delta)(\tau(\id\times [h,\eta\inr]))^\klstar
\end{flalign*}
where the last step is due to the obvious identity $[u\times v,u\times w]=(u\times [v,w])\dist^{\mone}$.
Finally, we obtain~\eqref{eq:str4} as follows:
\begin{flalign*}
&&T(\id&+\SigF((\tau^\nu(\id\times g))^\kklstar\tau^\nu))\\
&&&~~(T\delta)\tau(\id\times[h ,\eta\inr]^\klstar\out)\\
&&=\;&T(\id+\SigF((\tau^\nu(\id\times g))^\kklstar\tau^\nu))\\
&&&~~(T\delta)(\tau(\id\times[h ,\eta\inr]))^\klstar\tau(\id \times \out)&\by{{\sc str$_4$} for $\tau$}\\
&&=\;&T(\id+\SigF((\tau^\nu(\id\times g))^\kklstar\tau^\nu))\\
&&&~~[(T\delta)\tau(\id\times h ),\eta\inr]^\klstar(T\delta)\tau(\id \times \out)&\by{\ref{eq:str4_dia}}\\
&&=\;&T(\id+\SigF((\tau^\nu(\id\times g))^\kklstar\tau^\nu))\\
&&&~~[(T\delta) \tau(\id\times T(\id+\SigF\TF\inr)) (\id\times\out g),\eta\inr]^\klstar\\
&&&~~(T\delta)\tau(\id\times\out)&\by{definition of~$h$}\\
&&=\;&T(\id+\SigF((\tau^\nu(\id\times g))^\kklstar\tau^\nu))\\
&&&~~[(T\delta) T(\id\times (\id+\SigF\TF\inr)) \tau(\id\times\out g),\eta\inr]^\klstar\\
&&&~~(T\delta)\tau(\id\times\out)&\by{naturality of $\tau$}\\
&&=\;&T(\id+\SigF((\tau^\nu(\id\times g))^\kklstar\tau^\nu))\\
&&&~~[T(\id+\SigF(\id\times\TF\inr))(T\delta)\tau(\id\times\out g),\eta\inr]^\klstar\\
&&&~~(T\delta)\tau(\id\times\out)&\by{naturality of $\delta$}\\
&&=\;&[T(\id+\SigF((\tau^\nu(\id\times g))^\kklstar\tau^\nu(\id\times\TF\inr)))(T\delta)\tau(\id\times\out g),\\
&&&~~\eta\inr \SigF((\tau^\nu(\id\times g))^\kklstar\tau^\nu) ]^\klstar(T\delta)\tau(\id\times\out)&\by{naturality of $\eta$}\\
&&=\;&[T(\id+\SigF((\tau^\nu(\id\times g\inr))^\kklstar\tau^\nu))(T\delta)\tau(\id\times\out g),\\
&&&~~\eta\inr \SigF((\tau^\nu(\id\times g))^\kklstar\tau^\nu))]^\klstar(T\delta)\tau(\id\times\out)&\by{naturality of~$\tau^\nu$}\\
&&=\;&[T(\id+\SigF((\tau^\nu(\id\times\eta^\nu))^\kklstar\tau^\nu))(T\delta)\tau(\id\times\out g),\\
&&&~~\eta\inr \SigF((\tau^\nu(\id\times g))^\kklstar\tau^\nu))]^\klstar(T\delta)\tau(\id\times\out)&\by{since $g=[f,\eta^\nu]$}\\
&&=\;&[T(\id+\SigF\tau^\nu)(T\delta)\tau(\id\times\out g),\\
&&&~~\eta\inr \SigF((\tau^\nu(\id\times g))^\kklstar\tau^\nu))]^\klstar(T\delta)\tau(\id\times\out)&\by{{\sc str$_3$} for $\tau^{\nu}$}\\
&&=\;&[T(\id+\SigF\tau^\nu)(T\delta)\tau(\id\times\out)(\id\times g),\\
&&&~~\eta\inr\SigF((\tau^\nu(\id\times g))^\kklstar\tau^\nu)]^\klstar(T\delta)\tau(\id\times\out)\\
&&=\;&[\out\tau^\nu(\id\times g),\eta\inr\SigF((\tau^\nu(\id\times g))^\kklstar\tau^\nu)]^\klstar(T\delta)\tau(\id\times\out)&\by{definition of~$\tau^\nu$}\\
&&=\;&[\out\tau^\nu(\id\times g),\eta\inr\SigF(\tau^\nu(\id\times g))^\kklstar)]^\klstar\\
&&&~~T(\id+\SigF\tau^\nu)(T\delta)\tau(\id\times\out)\\
&&=\;&[\out\tau^\nu(\id\times g),\eta\inr\SigF(\tau^\nu(\id\times g))^\kklstar]^\klstar\out\tau^\nu&\by{definition of~$\tau^\nu$}\\
&&=\;&\out(\tau^\nu(\id\times g))^\kklstar\tau^\nu.&\by{definition of~$\argument^\kklstar$}
\end{flalign*}
\end{citemize}
We have thus shown all properties ({\sc str$_1$})--({\sc str$_4$}) and the proof is completed.
\end{proof}
\noindent Following Uustalu~\cite{Uustalu03} (and other
work~\cite{PirogGibbons13,AczelAdamekEtAl03}), we next introduce a
notion of guardedness.
\begin{defi}\textbf{(Guardedness)}\label{defn:guard}
A morphism $f:X\to \TF(Y+Z)$ is \emph{guarded} if there is $u:X\to T(Y+\SigF\TF(Y+Z))$ such that $\out \comp f=T(\inl+\id)\comp u$:
\begin{equation*}
\xymatrix@L=8pt@R20pt@C58pt@M6pt{
X\ar[r]^{f}\ar[d]_{u}& \TF(Y+Z)\ar[d]^{\out}\\
  T(Y+\SigF\TF(Y+Z))\ar[r]^-{T(\inl+\id)} & T((Y+Z)+\SigF\TF(Y+Z)).}
\end{equation*}
\end{defi}

\noindent Guardedness of $f:X\to \TF(Y+Z)$ intuitively means that any
call to a computation of type $Z$ in $f$ occurs only under a free
operation, i.e.~via the right hand summand in
$T((Y+Z)+\SigF\TF(Y+Z))$. A familiar instance of this notion occurs in
process algebra~\cite{BergstraPonseEtAl01}, illustrated in simplified
form as follows.
\begin{exa}\label{exp:guard_pa}
  Let $\BBT$ be the countable powerset monad over a suitable category,
  i.e.\ $TX=\PSet_{\omega_1} X=\{Y\subseteq X\mid |Y|\leq\omega\}$.
  Take $\Sigma= A\times (-)$; then the object $T_\Sigma X=\nu\gamma.\,\PSet_{\omega_1}(X+A\times\gamma)$
  can be considered as the domain of possibly infinite countably
  nondeterministic processes over actions from $A$ with final results
  in $X$. A morphism $n\to T_\Sigma(X+n)$ can be seen as a system of $n$
  mutually recursive process definitions; the latter is guarded in the
  sense of Definition~\ref{defn:guard} iff every recursive call of a
  process is preceded by an action, which coincides with the standard
  notion of guardedness from process algebra. We recall an example of
  an \emph{unguarded} definition in this setting in
  Section~\ref{sec:bsp}.
\end{exa}
\noindent

\section{Iteration on Coinductive Resumptions}
\label{sec:cr_elgot}
We next establish one of the main technical contributions of the paper
by proving that iteration operators, i.e.~Elgot monad structures,
propagate uniquely along extensions \mbox{$\BBT\to\BBTF$}, implying
that Elgot monads are closed under the coinductive generalized
resumption transformer.
\begin{thm}\label{thm:cr_elg}
  Let\/ $\BBT$ be a complete Elgot monad and let 
  $\BBTF$ be the monad identified in Theorem~\ref{lem:kl_dec}, i.e.~the
  coinductive generalized resumption monad over $\BBT$.
\begin{enumerate}
\item\label{item:it-tab} There is a unique iteration operator making
  $\BBTF$ a complete Elgot monad that extends iteration of $\BBT$
  in the sense that for $f:X\to\TF(Y+X)$ and $g:X\to T(Y+X)$, if
  \begin{align}\label{eq:cung}
  \out\comp f=(T\inl)\comp g
  \end{align}
  (i.e.\ $f=\tuo\comp(T\inl)\comp g$) then
  \begin{align}\label{eq:cungit}
  \out\comp f^\istar=(T\inl)\comp g^\istar.
\end{align}
 \item \label{item:it-tab-guarded} For any guarded morphism $f:X\to \TF(Y+X)$, $f^\istar$ is the \emph{unique} morphism satisfying the fixpoint law \mbox{$[\eta^\nu, f^{\istar}]^{\kklstar}\comp f = f^{\istar}$.}
\end{enumerate}
\end{thm}
\noindent The proof of Theorem~\ref{thm:cr_elg} relies on a fairly
complicated chain of calculations and will, to aid readability, be
partitioned into separate lemmas. Before we dive into these details,
let us outline the general idea.

Uustalu already proves that guarded morphisms $f$ have unique iterates
$f^{\istar}$ satisfying the fixpoint law~\cite[Theorem~3.11]{Uustalu03}, which readily implies the
second clause.  In showing the first clause of
Theorem~\ref{thm:cr_elg}, the key step is then to define $f^{\istar}$
for unrestricted $f$ in a consistent manner. For $f: X \to \TF(Y+X)$,
let $\guard f : X \to \TF(Y+X)$ be the composite
    \begin{align*}\samepage
    X\xto{~~w^\istar~}~&T(Y+\SigF\TF(Y+X))\\
     \xto{~T(\inl + \id)~}~&T((Y+X)+\SigF\TF(Y+X))\\
    \xto{~~\tuo~~}~&\TF(Y+X)
    \intertext{(guarded by definition), where $w$ is the composite}
	X\xto{~~f~~}~&\TF(Y+X)\\
	\xto{~\out~}~&T((Y+X)+\SigF\TF(Y+X))\\
	\xto{~~T\pi~~}~&T((Y+\SigF\TF(Y+X))+X)
	\end{align*}
  with $\pi = \left[\inl+\id, \inl \inr \right]$. That is,
  $\guard f$ makes $f$ guarded by iterating
  \begin{equation*}
  \out \comp f : X \to T( (Y + X) + \SigF\TF(Y+X))
  \end{equation*}
  (in the complete Elgot monad $\BBT$) over the middle summand
  $X$. It is easy to check that $\guard f = f$ when $f$ is
  guarded. We hence can consistently define
  \begin{equation}\label{eq:guard_def}
  f^{\istar} = (\guard f)^{\istar}
  \end{equation}
  (in $\BBTF$). The remaining technical challenge is now to prove that this
  definition indeed satisfies the axioms of complete Elgot monads and that it is the
unique such iteration operator on $\BBT_{\Sigma}$ extending the given iteration operator on $\BBT$.
\begin{lem}\label{lem:ext_ax}
Given $f:X\to \TF(Y+X)$, $f^\istar:f:X\to\TF Y$ defined by~\eqref{eq:guard_def}
satisfies \emph{fixpoint}, \emph{naturality}, and \emph{uniformity}.
\end{lem}
\begin{proof}
To make sure that definition~\eqref{eq:guard_def}
introduces the iteration consistently with the iteration for guarded morphisms 
we check that $\guard f=f$ whenever $f$ is guarded. Suppose that
$\out f=T(\inl+\id) u$. Then $f=\tuo T(\inl+\id) u$ and therefore
\begin{align*}
\guard f=&~\tuo T(\inl+\id) ((T\pi)\out  f)^\istar\\
=&~ \tuo T(\inl+\id) ((T\pi)\out  \tuo T(\inl+\id) u)^\istar\\
=&~ \tuo T(\inl+\id) ((T\pi) T(\inl+\id) u)^\istar\\
=&~ \tuo T(\inl+\id) (T\inl u)^\istar\\
=&~ \tuo T(\inl+\id) u\\
=&~ f.
\end{align*}
Let us check \emph{fixpoint}, \emph{naturality},  and \emph{uniformity} (Definition~\ref{defn:elgot}) in order.
\begin{citemize}
  \item \emph{Fixpoint.} For any $f:X\to \TF(Y+X)$ we have
\begin{flalign*}
&&f^{\istar} =&~[\eta^\nu,f^{\istar}]^{\kklstar}  \guard f &\by{definition of $\argument^\istar$}\\
&&=&~[\eta^\nu,f^{\istar}]^{\kklstar} \tuo T(\inl+\id) ((T\pi)\out  f)^\dagger&\by{definition of $\guard$}\\
&&=&~\tuo  \bigl[[\eta\inl,\out  f^{\istar}], \eta\inr \SigF{[\eta^\nu,f^{\istar}]^\kklstar}\bigr]^\klstar\\
&&&~~T(\inl+\id) ((T\pi)\out  f)^\istar&\by{definition of $\argument^\kklstar$}\\
&&=&~\tuo  \bigl[\eta\inl, \eta\inr \SigF{[\eta^\nu,f^{\istar}]^\kklstar}\bigr]^\klstar ((T\pi)\out  f)^\istar\\
&&=&~\tuo  T(\id+\SigF{[\eta^\nu,f^{\istar}]^\kklstar}) ((T\pi)\out  f)^\istar
\intertext{and thus we obtain the following intermediate equation:}
&&\out f^{\istar} =&~  T\bigl(\id+\SigF{[\eta^\nu,f^{\istar}]^\kklstar}\bigr) ((T\pi)\out  f)^\istar.\numberthis\label{eq:unf_int}\\
\intertext{Now, continuing the above calculation we obtain}
&&f^{\istar} =&~\tuo  T(\id+\SigF{[\eta^\nu,f^{\istar}]^\kklstar}) ((T\pi)\out  f)^\istar\\
&&=&~\tuo  T(\id+\SigF{[\eta^\nu,f^{\istar}]^\kklstar}) [\eta,((T\pi)\out  f)^\istar]^\klstar (T\pi)\out  f &\by{fixpoint}\\
&&=&~\tuo  [T(\id+\SigF{[\eta^\nu,f^{\istar}]^\kklstar})\eta,\out f^\istar]^\klstar (T\pi)\out  f &\by{\ref{eq:unf_int}}\\
&&=&~\tuo  [\eta(\id+\SigF{[\eta^\nu,f^{\istar}]^\kklstar}),\out f^\istar]^\klstar (T\pi)\out  f &\by{naturality of $\eta$}\\
&&=&~\tuo  \bigl[[\eta\inl,\out f^\istar], \eta\inr \SigF{[\eta^\nu,f^{\istar}]^\kklstar}\bigr]^\klstar \out  f \\
&&=&~\tuo  \bigl[\out [\eta^\nu,f^\istar], \eta\inr \SigF{[\eta^\nu,f^{\istar}]^\kklstar}\bigr]^\klstar \out  f &\by{definition of $\eta^\nu$}\\
&&=&~[\eta^\nu,f^{\istar}]^\kklstar \tuo \out  f&\by{naturality of $\argument^\kklstar$}\\
&&=&~[\eta^\nu,f^{\istar}]^\kklstar  f.
\end{flalign*}
\item\emph{Naturality.} Assume that $h:X\to\TF(Y+X)$ is guarded and show that so is $[(\TF \inl) g,\eta^{\nu}
	\inr]^\kklstar h$ for any $g:Y\to Z$. Let $u$ be such that $\out h=T(\inl+\id) u$ and let $w=[(\TF \inl)
	g,\eta^{\nu} \inr]$. Then
\begin{align*}
\out [&\,(\TF \inl) g,\eta^{\nu} \inr]^\kklstar h\\
 =&\; [\out w,\eta\inr \SigF w^\kklstar]^\klstar\out h\\
=&\; [\out w,\eta\inr \SigF w^\kklstar]^\klstar T(\inl+\id) u\\
=&\; [\out w\inl,\eta\inr \SigF w^\kklstar]^\klstar u\\
=&\; [\out (\TF \inl) g,\eta\inr \SigF w^\kklstar ]^\klstar u\\
=&\; [T(\inl+\SigF\TF \inl)\out g,\eta\inr \SigF w^\kklstar]^\klstar u\\
=&\; T(\inl+\id) \bigl[T(\id+\SigF\TF \inl)\out g,\eta\inr \SigF w^\kklstar\bigr]^\klstar u.
\end{align*}
Now, since $t=[(\TF \inl) g,\eta^{\nu} \inr]^\kklstar \guard f$ is guarded, it is the unique fixpoint of the map
\begin{displaymath}
	t\mapsto [\eta^\nu,t]^\kklstar [(\TF \inl) g,\eta^{\nu} \inr]^\kklstar \guard f.
\end{displaymath}
However, on the other hand,
\begin{flalign*}
	[\eta^\nu,g^\kklstar f^\istar]^\kklstar &[(\TF \inl) g,\eta^{\nu} \inr]^\kklstar \guard f\\
 =&\;  [ g,g^\kklstar f^\istar]^\kklstar \guard f\\
 =&\; [ g,g^\kklstar (\guard f)^\istar]^\kklstar \guard f\\
 =&\; g^\kklstar[\eta^\nu, (\guard f)^\istar]^\kklstar \guard f\\
 =&\; g^\kklstar f^\istar
\end{flalign*}
and therefore $t^\istar=g^\kklstar f^\istar$. It remains to show that 
\begin{align*}
[(\TF \inl) g,\eta^{\nu} \inr]^\kklstar \guard f=
\guard[(\TF \inl) g,\eta^{\nu} \inr]^\kklstar f.
\end{align*} 
Since $\guard f$ is guarded by definition, we know by the calculation above
that $[(\TF \inl) g,\eta^{\nu} \inr]^\kklstar \guard f$ is guarded and therefore
\[ [(\TF \inl) g,\eta \inr]^\kklstar \guard f = \guard [(\TF \inl) g, \eta \inr]^{\kklstar} \guard f. \]
To finish the proof, we calculate
\[
\guard [(\TF \inl) g, \eta^{\nu} \inr]^{\kklstar} \guard f = \tuo T(\inl + \id) ((T\pi) \out [(\TF \inl) g,
\eta^{\nu}\inr]^{\kklstar} \guard f)^{\istar}.
\]
Further transforming the dagger expression in the previous term yields
\begin{flalign*}
&& ((T\pi)&\, \out [(\TF \inl) g, \eta \inr]^{\kklstar} \guard f)^{\istar} \\
&&=&~((T\pi) [\out w, \eta \inr \SigF w^{\kklstar}]^{\klstar} T(\inl + \id) ((T\pi) \out f)^{\istar})^{\istar}\\
&&=&~((T\pi) [\out (\TF\inl) g, \eta \inr \SigF w^{\kklstar}]^{\klstar} ((T\pi) \out f)^{\istar})^{\istar}\\
&&=&~([(T\inl) (T\pi) [\out (\TF\inl) g, \eta \inr \SigF w^{\kklstar}], \eta \inr]^{\klstar} (T\pi) \out f)^{\istar\istar} &\by{nat. for $\BBT$}\\
&&=&~(T[\id,\inr] [(T\inl) (T\pi) [\out (\TF\inl) g, \eta \inr \SigF w^{\kklstar}], \eta \inr]^{\klstar}\comp(T\pi) \out f)^{\istar} &\by{codiag. for $\BBT$}\\
&&=&~([[(T\pi) \out (\TF\inl) g, \eta \inr],(T\pi) \eta \inr \SigF w^{\kklstar}]^{\klstar} \out f)^{\istar} \\
&&=&~((T\pi) [[\out(\TF\inl) g, \eta \inl \inr],\eta \inr \SigF w^{\kklstar}]^{\klstar} \out f)^{\istar} \\
&&=&~((T\pi) [\out [(\TF\inl) g, \eta^{\nu} \inr],\eta \inr \SigF w^{\kklstar}]^{\klstar} \out f)^{\istar} \\
&&=&~((T\pi) \out [(\TF\inl) g, \eta^{\nu} \inr]^{\kklstar} f)^{\istar}
\end{flalign*}
and therefore
\begin{align*}
	\guard [(\TF \inl) g, \eta \inr]^{\kklstar} \guard f = \guard [(\TF\inl) g, \eta \inr]^{\kklstar} f.
\end{align*}
 
\item\emph{Uniformity.} First, we show uniformity under the assumption that $g$ is guarded. Suppose that $f h = \TF(\id+ h) g$. It is then sufficient to verify that $f^{\istar} h$ satisfies the fixpoint law for $g$. Indeed,
	  \begin{align*}
		  f^{\istar}  h = ~& [\eta^\nu, f^{\istar}]^{\kklstar}  f h \\
		  = ~& [\eta^\nu, f^{\istar}]^{\kklstar} \TF(\id +h) g \\
		  = ~& [\eta^\nu, f^{\istar} h]^{\kklstar}  g.
	  \end{align*}
Now consider the general case. Suppose that again we have $f h = \TF(\id+ h) g$. We prove the following auxiliary identity:
\begin{align}\label{eq:unif_int}
((T\pi)\out f)^{\istar} h = T\bigl(\id+\SigF\TF(\id+h)\bigr)((T\pi)\out g)^{\istar}.
\end{align}
Observe that
\begin{align*}
((T\pi)\out f) h =&\; (T\pi)\out \TF(\id+ h) g\\
=&\; (T\pi) T(\id+h+\SigF\TF(\id+h))\out g\\
=&\; T(\id +h) T( (\id+\SigF\TF(\id +h)) +\id) (T\pi)\out g,
\end{align*}
from which by uniformity of the iteration operator of $\BBT$, we obtain
\begin{align*}
((T\pi)\out f)^{\istar} h =\bigl(T( (\id+\SigF\TF(\id +h)) +\id) (T\pi)\out g\bigr)^{\istar}.
\end{align*}
After transforming the right hand side by naturality of the iteration operator of $\BBT$ we arrive at~\eqref{eq:unif_int}.

Next we prove that $(\guard{f}) h = \TF(\id+ h)\guard{g}$:
\begin{flalign*}
&&(\guard{f}) h =\;& \tuo  T(\inl +\id)((T\pi)\out f)^{\istar} h &\by{definition of~$\guard$}\\
&&=\;& \tuo T(\inl +\id) T(\id+\SigF\TF(\id+h))((T\pi)\out g)^{\istar}&\by{\ref{eq:unif_int}} \\
&&=\;& \tuo T((\id + h) +\SigF\TF(\id + h))  T(\inl +\id)((T\pi)\out g)^{\istar} \\
&&=\;& \TF(\id + h)\tuo  T(\inl +\id)((T\pi)\out g)^{\istar} &\by{Lemma~\ref{lem:final-functor}} \\
&&=\;& \TF(\id + h)\guard{g}.&\by{definition of~$\guard$}
\end{flalign*}
\noindent
We have shown before that for guarded $g$ uniformity holds, and therefore $f^\istar h=(\guard{f})^\istar h=(\guard g)^\istar=g^\istar$.\qed
\end{citemize}
\noqed\end{proof}
\noindent We now deal with the last axiom, \emph{codiagonal}, whose
proof is more involved that that of the other properties and therefore
handled in a separate lemma:
\begin{lem}\label{lem:ext_cod}
  The assignment of $f^\istar:X\to\TF Y$ to $f:X\to \TF(Y+X)$ defined
  by~\eqref{eq:guard_def} satisfies the \emph{codiagonal} law.
\end{lem}
\begin{proof}
Let $g:X\to\TF((Y+X)+X)$. We shall
  show below that
\begin{align}
\guard(\TF[\id,\inr] \guard g) =&\; \guard(\TF[\id,\inr] g).
\label{eq:guard_nab} 	%
\end{align}

\noindent
Since $\TF[\id,\inr]^\istar g$ is the unique fixpoint of the map
\begin{displaymath}
\gamma\mapsto [\eta^\nu,\gamma]^\kklstar\guard(\TF[\id,\inr] g)
\end{displaymath}
we will be done once we show that $g^{\istar\istar}$ is also a fixpoint of the same map, i.e.\
\begin{align}\label{eq:fix_dd}
g^{\istar\istar} = [\eta^\nu,g^{\istar\istar}]^\kklstar\guard(\TF[\id,\inr] g).
\end{align}
Let us again denote by $\pi: (Y+X)+X\to (Y+X)+X$ the morphism swapping the last
 two components of the coproduct. We consider the following three cases.
\begin{cenumerate}
 \item\label{item:bi-guard} \emph{$\TF[\id,\inr] g$ is guarded.} Then we obtain~\eqref{eq:fix_dd} directly as follows
\begin{flalign*}
&&g^{\istar\istar} =\;& [\eta^\nu,g^{\istar\istar}]^\kklstar g^\istar&\by{fixpoint}\\
&&=\;& [\eta^\nu,g^{\istar\istar}]^\kklstar [\eta^\nu,g^\istar]^\kklstar g&\by{fixpoint}\\
&&=\;& \bigl[[\eta^\nu,g^{\istar\istar}], [\eta^\nu,g^{\istar\istar}]^\kklstar g^\istar\bigr]^\kklstar g\\
&&=\;& \bigl[[\eta^\nu,g^{\istar\istar}], g^{\istar\istar}\bigr]^\kklstar g&\by{fixpoint}\\
&&=\;& [\eta^\nu,g^{\istar\istar}]^\kklstar\TF[\id,\inr] g\\
&&=\;&[\eta^\nu,g^{\istar\istar}]^\kklstar\guard(\TF[\id,\inr] g).
\end{flalign*}
\item\label{item:co-guard} \emph{~$(\TF\pi) g$ is guarded.} E.g.\ let $(\TF\pi) g=\tuo T(\inl+\id) u$. Then
	$\TF[\id,\inr]\guard g$ is also guarded, which is certified by the following calculation, involving the definitions of $g$, $\guard$ and the naturality law for $\argument^\istar$:
\begin{flalign*}
&&\TF[\id,&\inr]\guard g\\
&&=\;& \TF[\id,\inr]\guard((\TF\pi)\tuo T(\inl+\id) u)&\\
&&=\;& \TF[\id,\inr]\tuo T(\inl+\id)\bigl((T\pi)\out(\TF\pi)\tuo T(\inl+\id) u\bigr)^\istar\\
&&=\;& \TF[\id,\inr]\tuo T(\inl+\id)\bigl((T\pi) T(\pi+\SigF\TF\pi) T(\inl+\id) u\bigr)^\istar\\
&&=\;& \TF[\id,\inr]\tuo T(\inl+\id)\bigl(T((\inl+\id)+\id)(T\pi)T(\id+\SigF\TF\pi) u\bigr)^\istar\\
&&=\;& \TF[\id,\inr]\tuo T(\inl+\id)T(\inl+\id)\bigl((T\pi)T(\id+\SigF\TF\pi)  u\bigr)^\istar\\
&&=\;& \TF[\id,\inr]\tuo T(\inl\inl+\id)\bigl((T\pi)T(\id+\SigF\TF\pi)  u\bigr)^\istar\\
&&=\;& \tuo T([\id,\inr]+\SigF\TF[\id,\inr]) T(\inl\inl+\id)\bigl((T\pi)T(\id+\SigF\TF\pi)  u\bigr)^\istar\\
&&=\;& \tuo T(\inl+\SigF\TF[\id,\inr])\bigl((T\pi)T(\id+\SigF\TF\pi)  u\bigr)^\istar\\
&&=\;& \tuo T(\inl+\id)T(\id+\SigF\TF[\id,\inr])\bigl((T\pi)T(\id+\SigF\TF\pi) u\bigr)^\istar.
\end{flalign*}
The proof of~\eqref{eq:fix_dd} now can be completed as follows:
\begin{flalign*}
&&g^{\istar\istar} =\;&(\guard g)^{\istar\istar}&\by{definition of $\argument^\istar$}\\
&&=\;& [\eta^\nu,(\guard g)^{\istar\istar}]^\kklstar\guard( \TF[\id,\inr]\guard g)&\by{Clause~\eqref{item:bi-guard}}\\
&&=\;& [\eta^\nu,g^{\istar\istar}]^\kklstar\guard( \TF[\id,\inr] g).&\by{\ref{eq:guard_nab}}
\end{flalign*}
\item\label{item:guard} \emph{~$g$ is guarded.} Let $h=(\TF\pi)\guard(\TF\pi)g$. It is easy to see that $h$ is guarded. We use the following identity
\begin{align}\label{eq:guard_unfold}
\guard g^\istar = [\eta^\nu, g^\istar]^\kklstar h
\end{align}
whose proof runs as follows. Let $g = \tuo T(\inl + \id) u$ for some $u$ and observe that $\pi \inl = (\inl + \id)$. We
apply $\out$ to the right-hand side of the equation,
\begin{flalign*}
	&&\out [\eta^{\nu},\,&g^{\istar}]^{\kklstar} (\TF\pi) \guard(\TF\pi)g \\
	&&&=[\out [\eta^{\nu}, g^{\istar}], \eta \inr \SigF[\eta^{\nu}, g^{\istar}]^{\kklstar}]^{\klstar} \\
    &&&\qquad \out (\TF\pi) \tuo T(\inl + \id) ((T\pi) \out (\TF \pi) g)^{\istar}&\by{defn.~of~$\argument^\istar$, $\guard$} \\
	&&&=[\out [\eta^{\nu}, g^{\istar}]\pi \inl, \eta \inr \SigF([\eta^{\nu},g^{\istar}]^{\kklstar}
	\TF\pi)]^{\klstar}((T\pi) \out (\TF\pi) g)^{\istar} \\
	&&&=[\out [\eta^{\nu} \inl, g^{\istar}], \eta \inr \SigF([\eta^{\nu},g^{\istar}]^{\kklstar}
	\TF\pi)]^{\klstar}\\
    &&&\qquad ((T\pi) T(\pi + \SigF\TF\pi) \out g)^{\istar} \\
	&&&=([(T\inl) [\out [\eta^{\nu} \inl, g^{\istar}], \eta \inr ([\eta^{\nu},g^{\istar}]^{\kklstar} \SigF\TF\pi)],
	\eta \inr]^{\klstar} \\&&&\qquad (T\pi) T(\pi + \SigF\TF\pi) \out g)^{\istar} &\by{naturality} \\
	&&&=([[(T\inl)\out[\eta^{\nu} \inl, g^{\istar}], \eta \inl \inr \SigF([\eta^{\nu},g^{\istar}]^{\kklstar}
	\TF\pi)], \eta \inr]^{\klstar} \\&&&\quad (T\pi) T(\pi +\SigF\TF\pi) \out g)^{\istar} \\
	&&&=([[(T\inl)\out[\eta^{\nu} \inl, g^{\istar}], \eta \inr], \eta \inl \inr \SigF([\eta^{\nu},g^{\istar}]^{\kklstar}
	\TF\pi)]^{\klstar} \\&&&\qquad T(\pi + \SigF\TF\pi) \out g)^{\istar} &\by{defn.\ $T\pi$} \\
	&&&=([[(T\inl)\out[\eta^{\nu} \inl, g^{\istar}], \eta \inr]\pi, \\
    &&&\qquad\eta \inl \inr \SigF([\eta^{\nu},g^{\istar}]^{\kklstar}
	 (\TF\pi)(\TF\pi))]^{\klstar} T(\inl + \id) u)^{\istar} &\by{$g$ guarded} \\
	&&&=([[(T\inl)\out[\eta^{\nu} \inl, g^{\istar}], \eta \inr] \pi \inl, \eta \inl \inr
	\SigF[\eta^{\nu},g^{\istar}]^{\kklstar}]^{\klstar} u)^{\istar} \\
	&&&=([[\eta \inl \inl \inl, \eta \inr], \eta \inl \inr \SigF[\eta^{\nu},g^{\istar}]^{\kklstar}]^{\klstar}
	u)^{\istar} \\
	&&&=([\eta (\inl \inl + \id), \eta \inl \inr \SigF[\eta^{\nu},g^{\istar}]^{\kklstar}]^{\klstar}
	u)^{\istar}.
\intertext{On the other hand, applying $\out$ to the left-hand side yields the same result:}
	&&&~\out \guard (g^{\istar})& \\
	&&&=T(\inl + \id) ((T\pi) \out (g^{\istar}))^{\istar} \\
	&&&=([(T\inl) \eta (\inl + \id), \eta \inr]^{\klstar} (T\pi) \out (g^{\istar}))^{\istar} &\by{naturality} \\
	&&&=([ [\eta \inl \inl \inl, \inl \inr], \eta \inr]^{\klstar} (T\pi) \out [\eta^{\nu}, g^{\istar}]^{\kklstar} g)^{\istar} \\
	&&&=([ \eta (\inl \inl + \id), \eta \inl \inr]^{\klstar} [\out [\eta^{\nu}, g^{\istar}], \eta \inr
	\SigF[\eta^{\nu},g^{\istar}]^{\kklstar}]^{\klstar} \out g)^{\istar} &\by{defn.\ $T\pi$} \\
	&&&=([ \eta (\inl \inl + \id), \eta \inl \inr]^{\klstar} [ [\eta \inl, \out g^{\istar}], \eta \inr
	\SigF[\eta^{\nu},g^{\istar}]^{\kklstar}]^{\klstar} \\&&&\qquad T(\inl + \id) u)^{\istar} &\by{$g$ guarded} \\
	&&&=([ \eta (\inl \inl + \id), \eta \inl \inr]^{\klstar} [ \eta \inl, \eta \inr
	\SigF[\eta^{\nu},g^{\istar}]^{\kklstar}]^{\klstar} u)^{\istar} \\
	&&&=([ \eta (\inl \inl + \id), \eta \inl \inr \SigF[\eta^{\nu},g^{\istar}]^{\kklstar}]^{\klstar} u)^{\istar}.
\intertext{Then the goal can be obtained as follows. First, observe the following:}
&&g^{\istar\istar}\, &= ([\eta^\nu, g^\istar]^\kklstar h)^\istar&\by{\ref{eq:guard_unfold}}\\
&&&= ([[\eta^\nu\inl,\eta^\nu\inr], g^\istar]^\kklstar h)^\istar\\
&&&= ([[\eta^\nu\inl, g^\istar], \eta\inr]^\kklstar \guard (\TF\pi)\, g)^\istar&\by{defn.~of~$\pi$}\\
&&&= ( (\TF[\id,\inr]) [[\eta^\nu\inl\inl, \TF\inl g^\istar], \eta\inr]^\kklstar \guard (\TF\pi)\, g)^\istar\\
&&&= ([[\eta^\nu\inl\inl, \TF\inl g^\istar], \eta\inr]^\kklstar \guard (\TF\pi)\, g)^{\istar\istar}&\by{Clause~\eqref{item:co-guard}}\\
&&&= ([\TF\inl[\eta^\nu\inl, g^\istar], \eta\inr]^\kklstar \guard (\TF\pi)\, g)^{\istar\istar}\\
&&&= ([\eta^\nu\inl, g^\istar]^\kklstar(\guard (\TF\pi)\, g)^\istar)^\istar&\by{naturality}\\
&&&= ([\eta^\nu\inl, g^\istar]^\kklstar((\TF\pi)\, g)^\istar)^\istar.&\by{defn.~of~$\argument^\istar$}\\
&&&= [\eta^\nu,([\eta^\nu\inl, g^\istar]^\kklstar((\TF\pi)\, g)^\istar)^\istar]^\kklstar [\eta^\nu\inl, g^\istar]^\kklstar((\TF\pi)\, g)^\istar&\by{fixpoint}\\
&&&= [\eta^\nu,g^{\istar\istar}]^\kklstar [\eta^\nu\inl, g^\istar]^\kklstar((\TF\pi)\, g)^\istar\\
&&&= [\eta^\nu, [\eta^\nu,g^{\istar\istar}]^\kklstar g^\istar]^\kklstar((\TF\pi)\, g)^\istar\\
&&&= [\eta^\nu, g^{\istar\istar}]^\kklstar ((\TF\pi)\, g)^\istar.&\by{fixpoint}
\end{flalign*}
It is easy to see that $((\TF\pi)\, g)^\istar$ is guarded, and hence, by the previous calculation,
$g^{\istar\istar}=((\TF\pi)\, g)^{\istar\istar}$. Finally, by Clause~\eqref{item:co-guard}, $((\TF\pi)\,
g)^{\istar\istar}=( (\TF[\id,\inr]) (\TF\pi)\, g)^\istar=(\TF[\id,\inr] g)^\istar$.
\item \emph{~$g$ is unrestricted.} Then,
\begin{flalign*}
&&g^{\istar\istar}=&\,(\guard g)^{\istar\istar}\\
&&=&\,((\TF[\id,\inr])\guard g)^\istar&\by{Clause~\eqref{item:guard}}\\
&&=&\,(\guard(\TF[\id,\inr])\guard(\TF\pi) g)^\istar\\
&&=&\,(\guard(\TF[\id,\inr])(\TF\pi) g)^\istar&\by{\ref{eq:guard_nab}}\\
&&=&\,(\TF[\id,\inr] g)^\istar
\end{flalign*}
and we are done.
It remains to prove \eqref{eq:guard_nab}. Observe that by definiton,
\begin{flalign*}
	&&&\guard(\TF[\id,\inr])\guard(\TF\pi) g\\
	&&=&~\guard\TF[\id,\inr] \tuo T(\inl+\id) ( (T\pi) \out g)^{\istar} \\
	&&=&~\tuo T(\inl+\id) ( (T\pi) \out \TF[\id,\inr] \tuo T(\inl + \id) ( (T\pi) \out g)^{\istar})^{\istar}
\intertext{Let us further transform the expression after $\tuo T(\inl+\id)$:}
	&&&~( (T\pi) T([\id,\inr] + \SigF\TF[\id,\inr]) T(\inl +\id) ( (T\pi) \out g)^{\istar})^{\istar}\\
	&&=&~( (T\pi) T(\id + \SigF\TF[\id,\inr]) ( (T\pi) \out g)^{\istar})^{\istar} \\
	&&=&~(  ([T\inl \pi \eta (\id + \SigF\TF[\id,\inr]),\eta \inr]^{\klstar} (T\pi) \out
	g)^{\istar})^{\istar}&\by{naturality} \\
	&&=&~([(T\pi) \eta (\id + \SigF\TF[\id,\inr]),(T\pi) \eta \inl \inr]^{\klstar} (T\pi) \out g)^{\istar} &\by{codiagonal} \\
	&&=&~( (T\pi) [ [\eta \inl, \eta \inl \inr], \eta \inr \SigF\TF[\id,\inr]]^{\klstar} \out g)^{\istar} \\
	&&=&~( (T\pi) [ \out (\eta^{\nu} [\id, \inr]), \eta \inr \SigF\TF[\id, \inr]]^{\klstar} \out g)^{\istar} \\
	&&=&~( (T\pi) \out (\TF[\id,\inr]) g)^{\istar}.
\intertext{Therefore,}
	&&&~\guard(\TF[\id,\inr])\guard(\TF\pi) g \\
	&&=&~\tuo T(\inl+\id) ( (T\pi) \out (\TF[\id,\inr]) g)^{\istar} \\
	&&=&~\guard \TF[\id,\inr] g
\end{flalign*}
and we are done.\qed
\pacman{
Let $h=(\TF\pi)\guard(\TF\pi)g$. Then $(\TF\pi) h$ is guarded. Using the previous clause we obtain
\begin{flalign*}
&&h^{\istar\istar} =\;& [\eta^\nu,h^{\istar\istar}]^\kklstar\guard\bigl(\TF(\id+\nabla) \guard h\bigr)&\by{Clause~2}\\
&&=\;& [\eta^\nu,h^{\istar\istar}]^\kklstar\guard\bigl(\TF(\id+\nabla) \guard (\TF\pi)\guard(\TF\pi)g\bigr)\\
&&=\;& [\eta^\nu,h^{\istar\istar}]^\kklstar\guard\bigl(\TF(\id+\nabla) (\TF\pi)\guard(\TF\pi)g\bigr)&\by{\ref{eq:guard_nab}}\\
&&=\;& [\eta^\nu,h^{\istar\istar}]^\kklstar\guard\bigl(\TF(\id+\nabla) \guard(\TF\pi)g\bigr)\\
&&=\;& [\eta^\nu,h^{\istar\istar}]^\kklstar\guard\bigl(\TF(\id+\nabla) (\TF\pi)g\bigr)&\by{\ref{eq:guard_nab}}\\
&&=\;& [\eta^\nu,h^{\istar\istar}]^\kklstar\guard\bigl(\TF(\id+\nabla) g\bigr)\\
&&=\;& [\eta^\nu,h^{\istar\istar}]^\kklstar\guard\bigl(\TF(\id+\nabla) \guard g\bigr).&\by{\ref{eq:guard_nab}}.
\end{flalign*}
\newcommand{\lshort}{\mathsf{l}} \newcommand{\rshort}{\mathsf{r}}
By~\eqref{eq:fix_dd}, we are done once we show that
\begin{equation}
h^{\istar\istar}=g^{\istar\istar}.\label{eq:inner-guard2}
\end{equation}
This will heavily involve nesting of coproduct injections $\inl$,
$\inr$, which we therefore abbreviate to $\lshort$, $\rshort$,
respectively. Moreover, we generally write $\guard' f$ for
$\TF\pi\guard \TF\pi f$. To begin, we note some easily proved
fixpoint laws for $\istar$, $\guard$, and $\guard'$:
\begin{align}
  \label{eq:unfold-istar}
  \out f^\istar & =
  (T[[\lshort\lshort,\rshort],\lshort\rshort \SigF([\eta^\nu,f^\istar]^\klstar)]
  \out f)^\istar\\
  \label{eq:unfold-guard}
  \out\guard f & =
  (T[[\lshort\lshort\lshort,\rshort],\lshort\rshort]\out f)^\istar\\
  \label{eq:unfold-guard2}
  \out\guard' f & =
  (T[[[\lshort\lshort\lshort\lshort,\rshort],\lshort\lshort\rshort],
  \lshort\rshort]\out f)^\istar.
\end{align}
It is easy to see that $h^\istar=(\guard' g)^\istar$ is guarded. It
therefore suffices to show that $g^{\istar\istar}$ satisfies the
fixpoint law for $(\guard' g)^{\istar\istar}$,
\begin{equation}
  \label{eq:unfold-inner-guard2}
  [\eta^\nu,g^{\istar\istar}](\guard' g)^\istar =
  g^{\istar\istar}.
\end{equation}
Moreover, note that as a special case of naturality, we generally have
\begin{equation*}
  Tf(u)^\istar = (T(f+\id)u)^\istar.
\end{equation*}
We proceed to prove \eqref{eq:unfold-inner-guard2} by applying $\out$
to both sides and unfolding. We start with the left hand side, which
by \eqref{eq:unfold-istar} equals
\begin{equation}
  \label{eq:istarr-lhs-unfold}
  (T[[\lshort\lshort,\rshort],\lshort\rshort \SigF([\eta^\nu,g^{\istar\istar}]^*)]
  \out(\guard'g)^\istar)^\istar.
\end{equation}
For readability, we now focus on the subterm $\out(\guard'g)^\istar$:
\begin{flalign*}
  &&\out(\guard'g)^\istar & =
  (T[[\lshort\lshort,\rshort],\lshort\rshort \SigF([\eta^\nu,g^\istar]^\klstar)]
  \out\guard'g)^\istar
  & \by{by \eqref{eq:unfold-istar}}\\
  &&& =(T[[\lshort\lshort,\rshort],\lshort\rshort \SigF([\eta^\nu,g^\istar]^\klstar)]
  (T[[[\lshort\lshort\lshort\lshort,\rshort],\lshort\lshort\rshort],
  \lshort\rshort]\out g)^\istar)^\istar
  & \by{by \eqref{eq:unfold-guard2}}\\
  &&& =(T([[\lshort\lshort,\rshort],\lshort\rshort \SigF([\eta^\nu,g^\istar]^\klstar)]+\id_X)\\
  &&& \qquad\qquad\qquad\qquad
  T[[[\lshort\lshort\lshort\lshort,\rshort],\lshort\lshort\rshort],
  \lshort\rshort]\out g)^{\istar\istar}
  & \by{naturality}\\
  &&& =
  (T[[[\lshort\lshort\lshort\lshort,\rshort],\lshort\rshort],\lshort\lshort
  \rshort \SigF ([\eta^\nu,g^\istar]^\klstar)]\out g)^{\istar\istar}\\
  &&& =
  (T[[[\lshort\lshort\lshort,\rshort],\rshort],\lshort\rshort
  \SigF([\eta^\nu,g^\istar]^\klstar)]\out g)^\istar.
  & \by{codiagonal}
 \end{flalign*}
 Applying naturality, we obtain that the entire left hand side equals
 \begin{align*}
 & (T([[\lshort\lshort,\rshort],\lshort\rshort
 \SigF([\eta^\nu,g^{\istar\istar}]^\klstar)]+\id)
 T[[[\lshort\lshort\lshort,\rshort],\rshort],\lshort\rshort
 \SigF([\eta^\nu,g^\istar]^\klstar)]\out g)^{\istar\istar}\\
 & =
 (T[[[\lshort\lshort\lshort,\rshort],\rshort],\lshort\lshort\rshort
 \SigF([\eta^\nu,g^{\istar\istar}]^\klstar)
 \SigF([\eta^\nu,g^\istar]^\klstar)]\out g)^{\istar\istar}\\
 & =
 (T[[[\lshort\lshort\lshort,\rshort],\rshort],\lshort\lshort\rshort
 \SigF([\eta^\nu,g^{\istar\istar}]^\klstar)
 \SigF([\eta^\nu,g^\istar]^\klstar)]\out g)^\istar.
 && \by{codiagonal}
 \end{align*}
 For comparison, we unfold the right hand side using
 \eqref{eq:unfold-istar} (twice) and naturality, obtaining
 \begin{align*}
 \out g^{\istar\istar} & =
 (T([[\lshort\lshort,\rshort],\lshort\rshort
 \SigF([\eta^\nu,g^{\istar\istar}]^\klstar)]+\id)
 T[[\lshort\lshort,\rshort],\lshort\rshort
 \SigF([\eta^\nu,g^\istar]^\klstar)]\out g)^{\istar\istar}\\
 & =
 (T[[[\lshort\lshort\lshort,\lshort\rshort],\rshort],\lshort\lshort\rshort
 \SigF([\eta^\nu,g^{\istar\istar}]^\klstar)
 \SigF([\eta^\nu,g^\istar]^\klstar)]\out g)^{\istar\istar}
 \end{align*}
 which by another application of the codiagonal law equals the
 expansion of the left hand side.

 It remains to prove \eqref{eq:guard_nab}, equivalently unfolded,
 i.e.\ with $\out$ applied to both sides. To be precise about
 bracketing of coproducts, we write $[\id_{Y+X},\rshort]$ in place of
 $\id+\nabla$. We then have as the unfolding of the left hand side
 \begin{flalign*}
 &&& \out\guard(\TF[\id_{Y+X},\rshort] \guard g) \\
 &&& = (T[[\lshort\lshort\lshort,\rshort],\lshort\rshort]\out
 \TF[\id_{Y+X},\rshort] \guard g)^\istar & \by{by
   \eqref{eq:unfold-guard}}\\
 &&& = (T[[\lshort\lshort\lshort,\rshort],\lshort\rshort]
 T([\id_{Y+X},\rshort]+\SigF\TF[\id_{Y+X},\rshort]) \out\guard
 g)^\istar
 & \by{naturality of $\out$}\\
 &&& =
 (T[[[\lshort\lshort\lshort,\rshort],\rshort],\lshort\rshort
 \SigF\TF[\id_{Y+X},\rshort]] \out\guard g)^\istar\\
 &&& =
 (T[[[\lshort\lshort\lshort,\rshort],\rshort],\lshort\rshort
 \SigF\TF[\id_{Y+X},\rshort]]  (T[[\lshort\lshort\lshort,\rshort],\lshort\rshort]\out g)^\istar)^\istar
 & \by{by
   \eqref{eq:unfold-guard}}\\
 &&& =
 ((T[[[\lshort\lshort\lshort,\rshort],\rshort],\lshort\rshort
 \SigF\TF[\id_{Y+X},\rshort]]+\id_X) T[[\lshort\lshort\lshort,\rshort],\lshort\rshort]\out g)^{\istar\istar}
 & \by{naturality}\\
 &&& = (T[[[\lshort\lshort\lshort\lshort,\lshort\rshort],\rshort],\lshort\lshort\rshort
 \SigF\TF[\id_{Y+X},\rshort]]\out g)^{\istar\istar}\\
 &&& = (T[[[\lshort\lshort\lshort,\rshort],\rshort],\lshort\rshort
 \SigF\TF[\id_{Y+X},\rshort]]\out g)^\istar.
 & \by{codiagonal}
 \end{flalign*}
For comparison, we calculate the unfolding of the right hand side:
\begin{flalign*}
  &&& \out\guard(\TF(\id+\nabla) g)\\
  &&& = (T[[\lshort\lshort\lshort,\rshort],\lshort\rshort]\out
 \TF[\id_{Y+X},\rshort] g)^\istar && \by{by
   \eqref{eq:unfold-guard}}\\
 &&& = (T[[\lshort\lshort\lshort,\rshort],\lshort\rshort]
 T([\id_{Y+X},\rshort]+\SigF\TF[\id_{Y+X},\rshort]) \out
 g)^\istar
 && \by{naturality of $\out$}\\
 &&& =
 (T[[[\lshort\lshort\lshort,\rshort],\rshort],\lshort\rshort
 \SigF\TF[\id_{Y+X},\rshort]] \out g)^\istar,
\end{flalign*}
which is thus identical to that of the left hand side.
}
\end{cenumerate}
\noqed\end{proof}
\begin{lem}\label{lem:ext_str}
  The assignment of $f^\istar:X\to\TF Y$ to $f:X\to \TF(Y+X)$ defined
  by~\eqref{eq:guard_def} is compatible with strength, i.e.
\begin{align*}
\tau^\nu (\id \times f^{\istar}) = ((\TF\dist) \tau^\nu (\id \times f))^{\istar}.
\end{align*}
\end{lem}
\begin{proof}
Let $f$ be guarded with $\out f = T(\inl + \id) u$. Then, $f' = (\TF \dist) \tau^{\nu} (\id \times f)$ is also
guarded with $\out f' = T(\inl + \id)  T(\id + \SigF((\TF \dist) \tau^\nu)) (T\delta) \tau (\id \times u)$ where
$\delta$ is as in Theorem~\ref{lem:kl_dec} (besides guardedness of $f$, the proof of this equation uses naturality of
$\out$ and the definitions of $\tau$ and $\dist$).
The following calculation shows that $\tau^\nu (\id \times f^{\istar})$ satisfies the fixpoint law for $((\TF\dist) \tau^\nu
  (\id \times f))^\istar$:
\begin{flalign*}
	&&&~\tau^\nu (\id \times f^{\istar}) \\
	&&=&~\tau^\nu (\id \times [\eta^\nu, f^{\istar}]^{\kklstar} f) \\
	&&=&~\tau^\nu (\id \times [\eta^\nu, f^{\istar}]^{\kklstar}) (\id \times f) \\
	&&=&~(\tau^\nu (\id \times [\eta^\nu, f^{\istar}]))^{\kklstar} \tau^\nu (\id \times f) &\by{\sc str$_4$} \\
	&&=&~([\eta^\nu, \tau^\nu (\id \times f^{\istar})] \dist)^{\kklstar} \tau^\nu (\id \times f) &\by{\ref{eq:str_dist}} \\
	&&=&~[\eta^\nu, \tau^\nu (\id \times f^{\istar})]^{\kklstar} (\TF \dist) \tau^\nu (\id \times f) \\
	&&=&~[\eta^\nu, \tau^\nu (\id \times f^{\istar})]^{\kklstar} \guard{(\TF \dist) \tau^\nu (\id \times f)},
\end{flalign*}
and hence $\tau^\nu (\id \times f^{\istar})$ and $((\TF \dist) \tau^\nu(\id \times f))^\istar$ are equal.

The general case reduces to the
guarded case by means of the equation
\begin{equation}\label{eq:tab_str_iter}
	(\TF \dist) \tau^{\nu} (\id \times \guard f) = \guard ((\TF \dist)
	\tau^{\nu} (\id \times f)),
\end{equation}
as follows:
\begin{flalign*}
&&\tau^\nu (\id \times f^{\istar})
=&\,\tau^\nu (\id \times (\guard f)^{\istar})&\by{definition of~$\argument^\istar$}\\
&&=&\, ((T\dist) \tau^\nu(\id \times\guard f))^\istar\\
&&=&\, (\guard (\TF \dist \tau^{\nu} (\id \times f)))^\istar&\by{\ref{eq:tab_str_iter}}\\
&&=&\, (\TF \dist \tau^{\nu} (\id \times f))^\istar.&\by{definition of~$\argument^\istar$}
\end{flalign*}
We show~\eqref{eq:tab_str_iter} by establishing commutativity of the following diagram where $Q = C \times B + C \times A$ 
(the identity in question is read from the border):

\bigskip
\begin{equation*}
\xymatrix@R25pt@C65pt@M5pt{
&
C \times A
	\ar@/_1pc/[dl]_*!/d6pt/-[@!15]{\scriptstyle\qquad\id \times ((T\pi) \out f)^{\istar}\quad}
	\ar@/^1pc/[dr]^-[@!-18]{\quad ((T\pi) \out (\TF \dist) \tau^{\nu} (\id \times f))^{\istar}\qquad}
&
\\
C \times T(B+\SigF\TF(B+A))
	\ar[d]_-{\id \times T(\inl + \id)}
	\ar[rr]^-{T(\id + \SigF\TF \dist) T(\id + \SigF\tau^{\nu}) (T\delta) \tau}
&
&
T(C\times B + \SigF\TF Q)
	\ar[d]^-{T(\inl + \id)}
\\
C \times T( (B+A) + \SigF\TF(B+A))
	\ar@/_1pc/[dr]_-[@!-17]{\qquad(\TF \dist) \tau^{\nu}(\id \times \tuo)}
	\ar[rr]^-{T(\dist + \SigF\TF \dist) T(\id + \SigF\tau^{\nu}) (T\delta) \tau}
&
&
T( Q + \SigF\TF Q)
	\ar@/^1pc/[dl]^*!/u4pt/-[@!10]{\scriptstyle\tuo}
&
\\
&
\TF Q
&
}
\end{equation*}

\bigskip\bigskip
\noindent
The bottom triangle commutes as follows:
\begin{align*}
	&~(\TF\dist) \tau^\nu (\id \times \tuo) \\
	=&~(\TF\dist) \tuo T(\id + \SigF\tau^\nu) (T\delta) \tau (\id \times \out) (\id
	\times \tuo) \\
	=&~\tuo \out (\TF\dist) \tuo T(\id + \SigF\tau^\nu) (T\delta) \tau \\
	=&~\tuo T(\dist + \SigF\TF\dist) T(\id + \SigF\tau^\nu) (T\delta) \tau.
\end{align*}
The middle square commutes by properties of~$\tau$, $\dist$ and~$\delta$:
\begin{flalign*}
	&&&~T(\dist + \SigF(\TF\dist \tau^\nu)) (T\delta) \tau (\id \times T(\inl + \id)) \\
	&&=&~T(\dist + \SigF(\TF\dist \tau^\nu))  (T\delta) T(\id \times (\inl + \id)) \tau &\by{naturality~of~$\tau$}\\
	&&=&~T(\dist + \SigF(\TF\dist \tau^\nu)) T( (\id \times \inl) + \id) (T\delta) \tau &\by{naturality~of~$\delta$}\\
	&&=&~T( \dist (\id \times \inl) + \id) T(\id + \SigF(\TF\dist \tau^\nu)) (T\delta) \tau \\
	&&=&~T(\inl + \id) T(\id + \SigF(\TF\dist \tau^\nu)) (T\delta) \tau.
\end{flalign*}
This leaves us with the top triangle. Let $\alpha = (\id + \SigF(\TF \dist \tau^{\nu})) \delta$. We apply the
assumption that $\tau$ is compatible with iteration to $(T\alpha) \tau (\id \times ((T\pi) \out f)^{\istar})$
and further calculate as follows:
\begin{flalign*}
	&&&~(T\alpha) \tau (\id \times ((T\pi) \out f)^{\istar})&\\
	&&=&~(T\alpha) (T\dist \tau (\id \times (T\pi) \out f))^{\istar} \\
	&&=&~(T(\alpha + \id) (T\dist) \tau (\id \times (T\pi) \out f))^{\istar} \tag{naturality} \\
	&&=&~(T(\alpha + \id) (T\dist) T(\id \times \pi)
	\tau (\id \times \out) (\id \times f))^{\istar} &\by{naturality of~$\tau$}\\
  &&=&~(T((\id + \SigF(\TF \dist \tau^\nu) + \id) \\
  &&&\quad T(\delta + \id) (T\dist) T(\id \times \pi)
  \tau (\id \times \out) (\id \times f))^{\istar} &\by{definition of~$\alpha$}\\
  &&=&~(T((\id + \SigF(\TF \dist \tau^\nu) + \id) \\
  &&&\quad T((\id+\rho) + \id) T(\dist + \id) (T\dist) T(\id \times \pi)\\
  &&&\quad\tau (\id \times \out) (\id \times f))^{\istar}. &\by{definition of~$\delta$}
\end{flalign*}
At this position we apply the obvious identity
\[ T(\dist + \id) (T\dist) T(\id \times \pi) = (T\pi)
  T(\dist + \id) (T\dist) \]
and then proceed as follows:
\begin{flalign*}
  &&&~(T((\id + \SigF(\TF \dist \tau^\nu) + \id) \\
  &&&\quad T((\id+\rho) + \id) (T\pi)
  T(\dist + \id) (T\dist)\\
  &&&\quad\tau (\id \times \out) (\id \times f))^{\istar}\\
  &&=&~(T((\id + \SigF(\TF \dist \tau^\nu) + \id)\\ 
  &&&\quad(T\pi) T(\dist+\rho) (T\dist)
  \tau (\id \times \out) (\id \times f))^{\istar} \\
  &&=&~(T((\id + \SigF(\TF \dist \tau^\nu)) + \id) \\
  &&&\quad(T\pi) T(\dist +\id) (T\delta) \tau (\id
  \times \out) (\id \times f))^{\istar} &\by{definition of~$\delta$}\\
	&&=&~((T\pi) T(\dist + \SigF\TF \dist)\\&&&\quad T(\id + \SigF\tau^{\nu}) (T\delta)  \tau (\id
	\times \out) (\id \times f))^{\istar}\\
	&&=&~((T\pi) T(\dist +\SigF\TF \dist) \out \tau^{\nu} (\id \times f))^{\istar}  &\by{definition of~$\tau^\nu$}\\
	&&=&~((T\pi) \out (\TF \dist) \tau^{\nu} (\id \times f))^{\istar}.
\end{flalign*}
This yields the proof of the top triangle of the diagram and therefore completes
the proof of the lemma.
\end{proof}
\noindent Finally, we can return to the proof of
Theorem~\ref{thm:cr_elg}.
\begin{proof}[Proof of Theorem~\ref{thm:cr_elg}]
As we indicated above, the second clause is already proved by Uustalu~\cite{Uustalu03}.
To show the existence part of the first clause we call on the above Lemmas~\ref{lem:ext_ax},~\ref{lem:ext_cod} 
and~\ref{lem:ext_str} and additionally prove that iteration on $\BBTF$ extends iteration on $\BBT$, i.e.~that~\eqref{eq:cung} implies~\eqref{eq:cungit}.
Let us call morphisms~$f$ for which there is $g$ satisfying~\eqref{eq:cung} \emph{completely unguarded}. Suppose that~\eqref{eq:cung} holds. Then the proof of~\eqref{eq:cungit}
runs as follows: 
\begin{flalign*}
  &&\out f^{\istar} =&~ (\guard f)^{\istar} \\
 && =&~ \out (\tuo{} T(\inl+\id) ((T\pi)\out f)^{\istar})^{\istar} \\
  &&=&~ \out (\tuo{} T(\inl+\id) ((T\pi)(T\inl)g)^{\istar})^{\istar} &\eqref{eq:cung} \\
  &&=&~ \out (\tuo{} T(\inl+\id) ((T\inl+\id)g)^{\istar})^{\istar} \\
 && =&~ \out (\tuo{} T(\inl+\id) (T\inl) g^{\istar})^{\istar} &\by{naturality} \\
 && =&~ \out (\tuo{} T(\inl \inl) g^{\istar})^{\istar} \\
  &&=&~ \out [\eta^{\nu}, f^{\istar}]^{\kklstar} \tuo{} T(\inl \inl) g^{\istar} &\by{fixpoint} \\
  &&=&~ \out [\eta^{\nu}, f^{\istar}]^{\kklstar} \tuo{} T(\inl \inl) g^{\istar} \\
 && =&~ [\out [\eta^{\nu},f^{\istar}],\eta \inr
  \Sigma[\eta^{\nu},f^{\istar}]^{\kklstar}]^{\klstar} T(\inl \inl) g^{\istar} &\eqref{eq:kl_def} \\
  &&=&~ [\out \eta^{\nu}]^{\klstar} g^{\istar} \\
  &&=&~ [\eta \inl]^{\klstar} g^{\istar} \\
  &&=&~ (T\inl)\,g^{\istar}.
\end{flalign*}
It remains to show the uniqueness part of the first clause. 
To that end we first show that any morphism $f: X \to \TF(Y+X)$ can be
decomposed by means of morphisms $g : X \to \TF(Z+X)$ and $h :
Z \to \TF(Y+X)$, where $Z = Y + \SigF\TF(Y+X)$, as
\begin{align}
	f = [h, \eta^\nu \inr]^{\kklstar} g \label{eq:it_uniq_1a}
\end{align}
with completely unguarded $g$. Next we show that
\begin{align}
	f^{\istar} = (h^{\kklstar} g^{\istar})^{\istar} \label{eq:it_uniq_2a}
\end{align}
and that
\begin{align}
	h^{\kklstar} g^{\istar} = \guard{f}. \label{eq:it_uniq_3}
\end{align}
In summary, we obtain that $f^{\istar} = (h^{\kklstar} g^{\istar})^{\istar}=(\guard f)^\istar$. The following proofs of~\eqref{eq:it_uniq_2a} and~\eqref{eq:it_uniq_3} do not depend on the concrete definition of $\argument^\istar$ on $\BBTF$ but only use its abstract properties as an iteration operator of a complete Elgot monad and compatibility with the underlying iteration operator for $\BBT$. Hence, the identity $f^{\istar}=(\guard f)^\istar$ would be valid for any other such operator, but since $(\guard f)^\istar$ is uniquely defined all of them must agree.

Let $g = \tuo{}T({\inl}\,\pi) \out{} f$ (recall that
$\pi = \left[\inl+\id, \inl \inr \right]$), which is, by definition,
completely unguarded, and let $h = \tuo{} \eta (\inl + \id)$.

Then the proof of~\eqref{eq:it_uniq_1a} runs as follows:
\begin{flalign*}
	&&\pheq [h, &\eta^{\nu} \inr]^{\kklstar} g\\
	&&&=[\tuo{} \eta (\inl + \id), \eta^{\nu} \inr]^{\kklstar} g\\
	&&&=[\tuo{} \eta (\inl + \id), \tuo{} \eta \inl \inr]^{\kklstar} g
	&&&\by{Theorem~\ref{lem:kl_dec}}\\
	&&&=(\tuo{} \eta [\inl + \id, \inl \inr])^{\kklstar} g\\
	&&&=(\tuo{} \eta\,\pi)^{\kklstar} g\\
	&&&=\tuo{} [\eta\,\pi, \eta \inr \SigF(\tuo{}{} \eta\,\pi)^{\kklstar}]^{\klstar} \out{} g
	&&&\by{Theorem~\ref{lem:kl_dec}}\\
	&&&=\tuo{} (\eta\, [\pi, \inr \SigF(\tuo{} \eta\,\pi)^{\klstar}])^{\klstar}T({\inl}\,\pi) \out{} f\\
	&&&=\tuo{} (\eta\,\pi)^{\klstar} (T\pi) \out{} f\\
	&&&=\tuo{} (T\pi)\, (T\pi) \out{} f\\
	&&&=f.
  \end{flalign*}
  Next, we show~\eqref{eq:it_uniq_2a}:
\begin{flalign*}
	&&\pheq (h^{\kklstar} g^{\istar})^{\istar}&=( (\TF{\inl}\ h, \eta^{\nu} \inr]^{\kklstar}  g)^{\istar} )^{\istar} &\by{naturality}\\
	&&&=\left( \TF[\id,\inr] [(\TF{\inl}) h, \eta^{\nu} \inr]^{\kklstar} g
	\right)^{\istar} &\by{codiagonal} \\
	&&&=\left( [h, \TF[\id, \inr] \eta^{\nu} \inr]^{\kklstar} g
	\right)^{\istar} \\
	&&&=\left( [h, \eta^{\nu} \inr]^{\kklstar} g \right)^{\istar} \\
	&&&=f^{\istar}.
\intertext{Finally, we prove (\ref{eq:it_uniq_3}):}
	&&\pheq h^{\kklstar} g^{\istar}&=(\tuo{} \eta (\inl + \id))^{\kklstar} g^{\istar} &\by{definition of $\argument^\kklstar$}\\
	&&&=\tuo{} [\eta (\inl + \id), \eta \inr \SigF h^{\kklstar}]^{\klstar} \out{} g^{\istar}\\
	&&&=\tuo{} [\eta (\inl + \id), \eta \inr \SigF h^{\kklstar}]^{\klstar} (T{\inl}) ((T\pi)
	\out{} f)^{\istar}&\by{$g$ compl.\ ung.}\\
	&&&=\tuo{} T(\inl + \id) ((T\pi) \out{} f)^{\istar}\\
	&&&=\guard f.
\end{flalign*}
This finishes the proof.
\end{proof}
\section{A Coproduct Characterization of Coinductive Resumptions}
\label{sec:sums}
\noindent Our second main result is a universal characterization of
the coinductive resumption monad transformer. Essentially, we show
that $\BBTF$ arises as the coproduct of $\BBT$ with the free complete
Elgot monad over $\Sigma$ (modulo existence of the latter) in the
category of complete Elgot monads on~$\BC$ (see
Section~\ref{sec:related} for discussion of a similar result on
\emph{completely iterative} monads). In other words, $\BBTF$ really
does freely extend $\BBT$ by $\Sigma$ in a fully formal sense. We
begin by recording the relevant notion of morphism of complete Elgot
monads:
 
\begin{defi}%
\label{defn:mor} 
  A \emph{complete Elgot monad morphism} $\xi:\BBR\to \BBS$ between complete
  Elgot monads $\BBR$, $\BBS$ is a morphism $\xi$ between the underlying
  strong monads (i.e.\ $\xi \comp \eta = \eta$,
  $\xi \comp f^\klstar = (\xi \comp f)^\klstar \comp \xi$ for $f:X\to RY$, and
  $\xi \comp \tau = \tau \comp (\id\times \xi)$, see~\cite{Moggi89}) additionally satisfying
  \begin{equation*}
  (\xi \comp g)^\istar = \xi \comp g^\istar
  \end{equation*}
  for $g:X\to R(Y+X)$. Complete Elgot monads over $\BC$ and their
  morphisms form an (overlarge) category $\mathbf{CElg}(\BC)$. We have
  a forgetful functor from $\mathbf{CElg}(\BC)$ to the category of
  strong functors and strong natural transformations; mention of
  \emph{free} complete Elgot monads refers to this forgetful functor.
\end{defi}
\noindent Note next that the coinductive resumption monad $\TF$
implements, by construction, all the operations of $\Sigma$, that is,
we have a canonical strong natural transformation
$\iota^\BBT:\Sigma\to\TF$, given by
\begin{equation*}
  \iota^{\BBT}_X=\tuo\eta\inr\Sigma\eta^\nu
\end{equation*}
where the typing of the composite is shown in
\begin{equation*}
  \Sigma X\xto{~~\Sigma\eta^\nu~~}
  S \TF X\xto{~~\inr~~} 
  X+ S \TF X \xto{~~\eta~~}
  T(X+ S \TF X) \xto{~~\tuo~~} 
  \TF X.
\end{equation*}
Moreover, recall that $T$ maps into $\TF$ via a natural transformation
$\ext:T\to\TF$ defined in Equation~\eqref{eq:ext}. We have
\begin{lem}\label{lem:ext-mm}
  The natural transformation $\ext:\BBT \to \BBTF$ is a complete Elgot
  monad morphism.
\end{lem}
\begin{proof}
Let us verify the identities
\begin{align}\label{eq:inr_mor}
\xi\eta = \eta && \xi f^\klstar = (\xi f)^\klstar\xi && \xi\tau = \tau(\id\times \xi) && (\xi g)^\istar = \xi g^\istar
\end{align}
with $f:X\to TY$ and $g:X\to T(Y+X)$ from left to right.
\begin{citemize}
	\item Compatibility of $\ext$ with unit is a straightforward consequence of Theorem~\ref{lem:kl_dec}: $\ext{} \eta = \tuo{} (T{\inl})\eta=\tuo{} \eta\inl =\eta^\nu$.
	\item In order to show compatibility of $\ext$ with Kleisli star we call the definition of the latter from Theorem~\ref{lem:kl_dec}:
		\begin{align*}
			(\ext{} g)^{\kklstar} \ext =&~ (\tuo{} (T  {\inl}) g)^{\kklstar}
			\tuo{} (T  {\inl}) \\
			=&~\tuo{} [\out{} \tuo{} (T  {\inl})  g, \eta \inr \SigF (\ext{}
			g)^{\kklstar}]^{\klstar} (T  {\inl}) \\
			=&~\tuo{} ((T  {\inl}) g)^{\klstar} \\
			=&~\tuo{} (T  {\inl})  g^{\klstar} \\
			=&~\ext{} g^{\klstar}.
		\end{align*}
	\item Recall the distributivity transformation $\delta: A \times (B+ \SigF C) \to A \times B + \SigF(A \times C)$ from Theorem~\ref{lem:kl_dec}. Then by the corresponding definition of $\tau^\nu$,
		\begin{align*}
			\tau^{\nu} (\id \times \ext)
			=&~\tuo{} T (\id + \SigF\tau^{\nu}) (T {\delta}) \tau
			(\id \times \out\ext) \\
			=&~\tuo{} T (\id + \SigF\tau^{\nu}) (T {\delta}) \tau
			(\id \times T {\inl}) \\
			=&~\tuo{} T (\id + \SigF\tau^{\nu}) T(\id+\rho) (T\dist) T(\id+\inl)\tau \\
			=&~\tuo{} T (\id + \SigF\tau^{\nu}) T(\id+\rho) (T\inl) \tau \\
			=&~\tuo{} T (\id + \SigF\tau^{\nu}) (T {\inl}) \tau \\
			=&~\tuo{} (T {\inl}) \tau \\
			=&~\ext{} \tau.
		\end{align*}
	\item Since $\out\,\ext g = (\TF{\inl}) g$, then by Theorem~\ref{thm:cr_elg}, $\out (\ext g)^{\istar} = (T\inl) g^\istar$,
		from which the last identity in~\eqref{eq:inr_mor} follows by composition with $\tuo$ on the left.
\qed
\end{citemize}
\noqed\end{proof}
\noindent Summing up, we have, slightly abusing notation, a cospan of
strong natural transformations
\begin{equation*}
  \BBT \xto{~~\ext~~} \BBTF \xleftarrow{~{}~~\iota^\BBT} \Sigma
\end{equation*}
with the left arrow being a complete Elgot monad morphism. It turns
out that this gives a universal characterization of $\BBTF$ in terms
of being composed of $\BBT$ and $\Sigma$:
\begin{thm}\label{thm:resump-universal}
  The cospan
  $\BBT \xto{~~\ext~~} \BBTF \xleftarrow{~{}~~\iota^\BBT} \Sigma$
  is universal. Explicitly: Given a complete Elgot monad $\BBS$, a
  strong natural transformation $\upsilon: \SigF\to S$, and a complete
  Elgot monad morphism $\sigma:\BBT\to\BBS$, there exists a unique
  complete Elgot monad morphism $\xi:\BBTF\to\BBS$ such that
  $\xi\ext=\sigma$ and $\xi\iota^\BBT=\upsilon$:
  \begin{equation}\label{diag:resump-universal}
  \xymatrix@C=3em{\BBT \ar[r]^\ext\ar[dr]_{\sigma} & \BBTF\ar@{.>}[d]_{\xi} & 
  \Sigma\ar[l]_{\iota^\BBT}\ar[dl]^{\upsilon}\\
  &\BBS.}
  \end{equation}
  Specifically, $\xi$ is given as $\xi=\zeta^\istar$ with
  $\zeta$ defined componentwise by
  \begin{equation*}
  \TF X\xto{~~\out~~} T(X+\SigF\TF X)\xto{~~\sigma~~}
  S(X+\SigF\TF X)
  \xto{~~[\eta\inl,(S\inr)\upsilon]^\klstar~~} S(X+\TF X).
  \end{equation*}
\end{thm}
\noindent In other words, $\BBTF$ is free as a complete Elgot monad over
$\Sigma$ that extends $\BBT$.
\renewcommand{\thmcontinues}[1]{Continued}
\begin{exa}\label{expl:ops}
  Let us spell out what a strong natural transformation
  $\upsilon:\Sigma\to S$ amounts to in the case
  $\Sigma X = \sum_{i} a_i\times X^{b_i}$ (Example~\ref{expl:sig}).  A
  natural transformation $\upsilon:\Sigma\to S$ is equivalent to a
  family of natural transformations $a_i\times(-)^{b_i}\to S$,
  equivalently $(-)^{b_i}\to S^{a_i}$, each of which is, by the
  enriched Yoneda lemma, equivalent to an element of $(Sb_i)^{a_i}$,
  i.e.~a morphism $u_i:a_i\to Sb_i$.
  Concretely, $\upsilon$ is assembled
  from the $u_i$ as follows:
\begin{align*}
\upsilon_X = [\lambda\brks{x,f}.\,(S f)(u_1(x)),\ldots,\lambda\brks{x,f}.\,(S f)(u_n(x))]. 
\end{align*} 
Note that the above generic argument makes use of the assumption that
$\BC$ is Cartesian closed. In fact it suffices to assume that only the
exponentials $(-)^{b_i}$ exist (in particular we do not actually need
the exponentials $(-)^{a_i}$ mentioned in between). The expressions
$\lambda\brks{x,f}.\,(S f)(u_i(x))$ above then have to be read as
\begin{displaymath}
a_i\times X^{b_i}
\xto{~\mathsf{swap}~} X^{b_i}\times a_i 
\xto{~\id\times u_i~} X^{b_i}\times Sb_i
\xto{~\tau~} S(X^{b_i}\times b_i)
\xto{~S\ev~} SX.
\end{displaymath}
where $\mathsf{swap}$ and $\ev$ are the obvious swapping and evaluation 
transformations respectively. 
\end{exa}

\noindent If $\mathbf{CElg}(\BC)$ has an initial object, then the
statement of Theorem~\ref{thm:resump-universal} can be phrased
slightly more concisely. We later give a sufficient criterion on $\BC$
that ensures this (Theorem~\ref{thm:lifting-initial}).
\begin{cor}\label{cor:sum}
  Suppose that $\mathbf{CElg}(\BC)$ has an initial object $\BBL$. Then
  \begin{enumerate}
  \item\label{item:free} $\BBL_\SigF$ is the free complete Elgot monad
  over the strong functor $\SigF :\BC\to\BC$, with universal arrow
  \begin{equation*}
  \iota^\BBL:\Sigma\to L_\Sigma.
  \end{equation*}
  \item\label{item:sum} For any complete Elgot monad $\BBT$, the
  coinductive generalized resumption monad $\BBTF$ is the coproduct
  of\/ $\BBT$ and $\BBL_\SigF$ in $\mathbf{CElg}(\BC)$, with left
  injection $\ext:\BBT\to\BBTF$ and with the right injection being
  the free extension of $\iota^\BBT:\Sigma\to\TF$ to $\BBL_\Sigma$.
\end{enumerate}
\end{cor}%
\begin{proof}
  Claim~(\ref{item:free}) is proved by taking $\BBT=\BBL$ in
  Theorem~\ref{thm:resump-universal}. Claim~(\ref{item:sum}) is then
  immediate.
\end{proof}
\noindent We assemble some auxiliary results before embarking on the
proof of Theorem~\ref{thm:resump-universal}.
\begin{lem}\label{lem:klstar-iter}
	The Kleisli composition of a complete Elgot monad\/ $\BBT$ can be characterized in terms of iteration as follows:
	\begin{equation}
		g^{\klstar}  f = [T (\inr  \inr) f, (T\inl)\  g]^{\istar}  \inl
	\end{equation}
\end{lem}
\begin{proof} By straightforward calculation:
	\begin{flalign*}
		&&&~[T (\inr  \inr)  f, (T\inl)  g]^{\istar}  \inl\\
		&&=&~[\eta, [T (\inr  \inr)  f, (T\inl)  g]^{\istar}]^{\klstar}  T (\inr  \inr)
		 f&\by{fixpoint}\\
		&&=&~([\eta, [T (\inr  \inr)  f, (T\inl)  g]^{\istar}]  \inr  \inr)^{\klstar}
		f\\
		&&=&~([T (\inr  \inr)  f, (T\inl)  g]^{\istar}  \inr)^{\klstar}  f \\
		&&=&~([\eta, [T (\inr  \inr)  f, (T\inl)\,  g]^{\istar}]^{\klstar}  (T\inl)\
		g)^{\klstar}  f &\by{fixpoint}\\
		&&=&~g^{\klstar}  f.\tag*{\qedhere}\\[-5pt]&&&&
	\end{flalign*}
\noqed\end{proof}
\begin{lem}\label{lem:iter-T-sum}
Let $f:X\to T(Y+X)$. Then $[\eta,f^\istar]^\klstar=(T(\id+f))^\istar$.
\end{lem}
\begin{proof}
Consider the following trivially commuting diagram
\begin{align*}
	\xymatrix@R20pt@C60pt@M6pt{
	X \ar[d]^-{f} \ar[r]^-{f} & T(Y+X) \ar[d]^-{T(\id+f)} \\
	T(Y+X)  \ar[r]^-{T(\id+f)} & T(Y+T(Y+X))
	}
\end{align*}
By uniformity, this implies $f^\istar =(T(\id+f))^\istar f$. Therefore $[\eta,f^\istar]^\klstar=[\eta,(T(\id+f))^\istar f]^\klstar=[\eta,(T(\id+f))^\istar]^\klstar T(\id+f)=(T(\id+f))^\istar$ and we are done.
\end{proof}
\noindent
We now proceed with the proof of the universal property:
\begin{proof}[Proof of Theorem~\ref{thm:resump-universal}]
  We first show that $\xi$ has the requisite properties, and then
  prove uniqueness. 

  \subsubsection*{Commutation of Diagram~\ref{diag:resump-universal}.}

  We need to show that $\xi\iota^\BBT=\upsilon$ and
  $\xi\ext=\sigma$. Put $w=[\eta\inl,
  (S\inr)\upsilon]^\klstar\sigma$. Then we have
\begin{flalign*}
&&\xi\iota^\BBT=&\;\xi\tuo\eta\inr \SigF\eta^\nu\\
&&=&\;(w\out)^\istar\tuo\eta\inr \SigF\eta^\nu\\
&&=&\;[\eta,(w\out)^\istar]^\klstar w\out\tuo\eta\inr \SigF\eta^\nu\\
&&=&\;[\eta,(w\out)^\istar]^\klstar [\eta\inl, (S\inr)\upsilon]^\klstar\sigma\eta\inr \SigF\eta^\nu\\
&&=&\;[\eta,(w\out)^\istar]^\klstar [\eta\inl, (S\inr)\upsilon]\inr \SigF\eta^\nu\\
&&=&\;[\eta,(w\out)^\istar]^\klstar S(\inr)\upsilon \SigF\eta^\nu\\
&&=&\;[\eta,\xi]^\klstar S(\inr\eta^\nu)\upsilon\\
&&=&\;(\xi\eta^\nu)^\klstar \upsilon\\
&&=&\;\eta^\klstar \upsilon\\
&&=&\; \upsilon,
\intertext{and}
&&\xi \ext=\,&\bigl([\eta\inl,(S\inr)\upsilon]^\klstar\sigma \out\bigr)^\istar \ext\\
&&=\,&[\eta, ([\eta\inl, \sigma \out]^\klstar (S\inr)\upsilon)^\istar]^\klstar\sigma\out\ext&\by{dinaturality, Lemma~\ref{lem:dinat}}\\
&&=\,&[\eta, ([\eta\inl, \sigma \out]^\klstar (S\inr)\upsilon)^\istar]^\klstar\sigma\out\tuo T\inl\\
&&=\,&[\eta, ([\eta\inl, \sigma \out]^\klstar (S\inr)\upsilon)^\istar]^\klstar\sigma\, T\inl\\
&&=\,&[\eta, ([\eta\inl, \sigma \out]^\klstar (S\inr)\upsilon)^\istar]^\klstar (S\inl) \sigma\\
&&=\,&\sigma
\end{flalign*}
\subsubsection*{$\xi$ is a complete Elgot monad morphism:}  
We have to show that $\xi_X=\zeta_X^\istar:\TF X\to SX$ is natural in
$X$ and satisfies the identities~\eqref{eq:inr_mor}. We successively
reduce verification of these properties to the last identity
in~\eqref{eq:inr_mor}, whose proof is the major challenge in
establishing the claim.

Note that $\zeta$ is a natural transformation (being a composite of
natural transformations), and hence
$\zeta\, \TF f=S(f+\TF f) \zeta=S(\id+\TF f)S(f+\id) \zeta$ for any
$f$. Therefore, by the uniformity and naturality laws we obtain
\begin{align*}
\xi\,\TF f=\zeta^\istar\TF f = (S(f+\id)\zeta)^\istar=(Sf)\zeta^\istar=(Sf)\xi,
\end{align*} 
i.e.~$\xi$ is natural.
The equation $\xi  \eta = \eta$ is shown as as follows:

{\allowdisplaybreaks[0]
\begin{align*}
	~\xi  \eta 
	=&~[\eta, \xi]^{\klstar}  [\eta  \inl, (S\inr)\upsilon]^{\klstar}  \sigma
	 \out{}  \eta \\
	=&~[\eta, \xi]^{\klstar}  [\eta  \inl, S(\inr f)\upsilon]^{\klstar}  \sigma
	 \out{}  \tuo{}  \eta  \inl \\
	=&~[\eta, \xi]^{\klstar}  \eta  \inl \\
	=&~\eta.
\end{align*}
}
Compatibility of $\xi$ with Kleisli lifting follows from
Lemma~\ref{lem:klstar-iter} and compatibility of $\xi$ with iteration, which we
argue later:

\begin{align*}
\xi\, g^{\kklstar} f =&\, \xi[\TF (\inr  \inr) f, (\TF\inl)  g]^{\istar}  \inl\\
=&\, [S (\inr  \inr) \xi f, (S\inl)  \xi g]^{\istar}  \inl\\
=&\, (\xi g)^{\klstar} (\xi f).
\end{align*}
We now show that $\xi = \zeta^{\istar}$ is compatible with strength,
i.e.~$\xi \tau^{\nu}=\tau (\id \times \xi)$. With a view to applying
uniformity, we calculate $\zeta\tau^\nu$:
\begin{flalign*}
&&  \zeta \tau^{\nu} =&~ [\eta \inl, (S\inr) \upsilon]^{\klstar} \sigma \out \tau^{\nu} \\
&&=&~ [\eta \inl, (S\inr) \upsilon]^{\klstar} \sigma T(\id + \Sigma \tau^{\nu}) (T\delta) \tau (\id \times \out) & \eqref{eq:str_def} \\
&&=&~ [\eta \inl, (S\inr) \upsilon]^{\klstar} S(\id + \Sigma \tau^{\nu}) (S\delta) \tau (\id \times \sigma \out) & \by{$\sigma$ monad morph.} \\
&&=&~ [\eta \inl, (S\inr) \upsilon (\Sigma \tau^{\nu})]^{\klstar} (S\delta) \tau (\id \times \sigma \out) \\
&&=&~ [\eta \inl, (S\inr) (S\tau^{\nu}) \upsilon]^{\klstar} (S\delta) \tau (\id \times \sigma \out) & \by{naturality of $\upsilon$} \\
&&=&~ [\eta \inl, S(\id + \tau^{\nu})(S\inr) \upsilon]^{\klstar} (S\delta) \tau (\id \times \sigma \out) \\
&&=&~ S(\id +\,\tau^{\nu}) [\eta \inl, (S\inr) \upsilon]^{\klstar} (S\delta) \tau (\id \times \sigma \out) \\
&&=&~ S(\id +\,\tau^{\nu}) [\eta \inl, (S\inr) \upsilon]^{\klstar} S(\id + \rho) (S\dist) \tau (\id \times \sigma \out) & \by{def. of $\delta$}\\
&&=&~ S(\id +\,\tau^{\nu}) ([\eta \inl, (S\inr) \upsilon \rho] \dist)^{\klstar} \tau (\id \times \sigma \out) \\
&&=&~ S(\id +\,\tau^{\nu}) ([\eta \inl, (S\inr) \tau (\id \times \upsilon)] \dist)^{\klstar} \tau (\id \times \sigma \out). & \by{strong nat. of $\upsilon$} \\
\intertext{Furthermore we simplify the tail of the latter expression:}
&&&~([\eta \inl, (S\inr) \tau (\id \times \upsilon)] \dist)^{\klstar} \tau (\id \times \sigma \out) \\
&&=&~ (S\dist) (S\dist^{\mone}) ([\eta \inl, (S\inr) \tau (\id \times \upsilon)] \dist)^{\klstar} \tau (\id \times \sigma \out) \\
&&=&~ (S\dist) ([\eta (\id \times \inl), S(\id \times \inr) \tau (\id \times \upsilon)] \dist)^{\klstar} \tau (\id \times \sigma \out) & \by{def. of $\dist^{\mone}$} \\
&&=&~ (S\dist) ([\tau (\id \times \eta \inl), \tau (\id \times (S\inr) \upsilon)] \dist)^{\klstar} \tau (\id \times \sigma \out) & \by{\sc str$_3$} \\
&&=&~ (S\dist) (\tau (\id \times [\eta \inl, (S\inr) \upsilon]))^{\klstar} \tau (\id \times \sigma \out) \\
&&=&~ (S\dist) \tau (\id \times [\eta \inl, (S\inr) \upsilon]^{\klstar}) (\id \times \sigma \out) \\
&&=&~ (S\dist) \tau (\id \times \zeta).
\end{flalign*}
We have obtained in summary that $\zeta \tau^{\nu}=S(\id +\,\tau^{\nu})(S\dist) \tau (\id \times \zeta)$. Therefore, by uniformity and compatibility of strength and iteration we obtain the desired identity:
\[
  \xi \tau^{\nu} =\zeta^\istar \tau^{\nu}= ((S\dist) \tau (\id \times \zeta))^{\istar} = \tau (\id \times \zeta^\istar)= \tau (\id \times \xi). 
\]
Finally, we are left to show that
\begin{align}\label{eq:xi-morph}
\xi f^\istar = (\xi f)^\istar
\end{align}
for any $f:X\to\TF(Y+X)$ %
where $\xi=([\eta\inl,(S\inr)\upsilon]^\klstar\sigma\out)^\istar$. We proceed by
successive reduction of the unrestricted identity~\eqref{eq:xi-morph} to the partial
cases when $f$ is guarded, and when $f$ is \emph{strongly guarded}. The latter
auxiliary notion is defined as follows. Recall that guardedness of $f$ means that 
$\out f$ factors through some $g:X\to T(Y+\SigF\TF(Y+X))$. We us call $f$ 
\emph{strongly guarded} if moreover there is $g':X\to T(Y+\SigF X)$ such that 
$\out f = T(\inl+\SigF(\eta^\nu\inr)) g'$.

\begin{citemize}
\item\textit{Reduction from unrestricted $f$ to guarded $f$.} Assuming that~\eqref{eq:xi-morph}
holds for guarded $f$, we obtain that for any $f$, $\xi f^\istar=\xi (\guard f)^\istar = (\xi \guard f)^\istar$. 
We are left to show that $(\xi \guard f)^\istar=(\xi f)^\istar$. To that end, consider the morphism $w$ given by the composition
\begin{align*}
X\xto{~\out f~} T( (Y+X)+\SigF\TF(Y+ X)) \xto{~[ \eta(\inl+\id),(S\inl)\xi^\klstar\upsilon]^\klstar\sigma~} S((Y+X)+X).
\end{align*}
Now, on the one hand
\begin{flalign*}
&&(S[\id,\inr]w)^\istar =&\, ([ \eta [\inl,\inr],\xi^\klstar\upsilon]^\klstar\sigma\out f)^\istar\\
&&=&\, ([\eta,\xi^\klstar\upsilon]^\klstar\sigma\out f)^\istar\\
&&=&\, ([\eta,\xi]^\klstar [\eta\inl,(S\inr)\upsilon]^\klstar\sigma\out f)^\istar \\
&&=&\, (\xi f)^\istar\\
\intertext{and on the other hand, by naturality of $\argument^\istar$,}
&&w^{\istar\istar} =&\,\bigl(([ \eta(\inl+\id),(S\inl)\xi^\klstar\upsilon]^\klstar\sigma\out f)^\istar\bigr)^\istar\\
&&=&\,\bigl(([(S\inl)[\eta\inl,\xi^\klstar\upsilon],\eta\inr]^\klstar S[\inl+\id,\inl\inr]\sigma\out f)^\istar\bigr)^\istar\\
&&=&\,\bigl(([(S\inl)[\eta\inl,\xi^\klstar\upsilon],\eta\inr]^\klstar \sigma (T\pi)\out f)^\istar\bigr)^\istar\\
&&=&\,\bigl([\eta\inl,\xi^\klstar\upsilon]^\klstar(\sigma (T\pi) \out f)^\istar\bigr)^\istar&\by{naturality}\\
&&=&\,\bigl([\eta,\xi^\klstar\upsilon]^\klstar S(\inl+\id) (\sigma (T\pi) \out f)^\istar\bigr)^\istar\\
&&=&\,\bigl([\eta,\xi^\klstar\upsilon]^\klstar\sigma T(\inl+\id) ((T\pi) \out f)^\istar\bigr)^\istar\\
&&=&\,([\eta,\xi^\klstar\upsilon]^\klstar\sigma \out\guard f)^\istar&\by{definition of~$\guard$}\\
&&=&\, ([\eta,\xi]^\klstar [\eta\inl,(S\inr)\upsilon]^\klstar\sigma\out \guard f)^\istar\\
&&=&\,(\xi\guard f)^\istar.&\by{fixpoint for~$\xi$}
\end{flalign*}
We thus obtain by the codiagonal law that
\begin{align*}
(\xi f)^\istar = (S[\id,\inr]w)^\istar = w^{\istar\istar} = (\xi\guard f)^\istar.
\end{align*}
\item\textit{Reduction from guarded $f$ to strongly guarded $f$.}  We
  proceed under the assumption that $f$ is guarded, i.e.\
  $\out f = T(\inl + \id) g$ for some $g : X \to
  T(Y+\SigF\TF(Y+X))$. Let $w$ be the following morphism:
\begin{align*}
\TF(Y+X)\xto{~\out~} &~T( (Y+X)+\SigF\TF(Y+X))\\
 \xto{~[[\eta\inl\inl, [\eta\inl\inl,(S\inl\inr)\upsilon]^\klstar\sigma g],(S\inr)\upsilon]^\klstar\sigma~} &~S((Y+\TF(Y+X))+\TF(Y+X)).
\end{align*}
Then, on the one hand, using dinaturality (Lemma~\ref{lem:dinat}),
\begin{flalign*}
&&w^{\istar\istar} =&\,\bigl(([[\eta\inl\inl, [\eta\inl\inl,(S\inl\inr)\upsilon]^\klstar\sigma g],(S\inr)\upsilon]^\klstar\sigma\out)^\istar\bigr)^\istar\\
&&=&\,\bigl([\eta\inl,  [\eta\inl,(S\inr)\upsilon]^\klstar\sigma g]^\klstar([\eta\inl,(S\inr)\upsilon]^\klstar\sigma\out)^\istar\bigr)^\istar&\by{naturality}\\
&&=&\,([\eta\inl,  [\eta\inl,(S\inr)\upsilon]^\klstar\sigma g]^\klstar\xi)^\istar&\by{definition~of~$\xi$}\\
&&=&\,[\eta, \left([\eta\inl,\xi]^{\klstar}  [\eta\inl,(S\inr)\upsilon]^\klstar\sigma g\right)^{\istar}]^{\klstar} \xi&\by{dinaturality}\\
&&=&\,[\eta, ([\eta,\xi]^\klstar [\eta\inl\inl,(S\inr)\upsilon]^\klstar\sigma g)^{\istar}]^{\klstar} \xi\\
&&=&\,[\eta, ([\eta,\xi]^\klstar [\eta\inl,(S\inr)\upsilon]^\klstar\sigma\out f)^{\istar}]^{\klstar} \xi&\by{definition~of~$g$}\\
&&=&\,[\eta, ([\eta,\xi]^\klstar \xi f)^{\istar}]^{\klstar} \xi&\by{definition~of~$\xi$}\\
&&=&\,[\eta, (\xi f)^{\istar}]^{\klstar} \xi.&\by{fixpoint}
\end{flalign*}
and hence $w^{\istar\istar} \eta^\nu\inr =[\eta, (\xi f)^{\istar}]^{\klstar} \xi \eta^\nu\inr =[\eta, (\xi f)^{\istar}]^{\klstar} \eta^\nu\inr= (\xi f)^{\istar}$. Next we introduce the following morphism $t$:
{\samepage
\begin{align*}
\TF(Y+X)\xto{~\out~} &~T((Y+X)+\SigF\TF(Y+X))\\
 \xto{~[[\eta \inl,g],\eta \inr]^{\klstar}~} &~T(Y+\SigF\TF(Y+X))\\
 \xto{~T(\inl+\SigF(\eta^{\nu}\inr))~} &~T((Y+\TF(Y+X))+\SigF\TF(Y+\TF(Y+X)))\\
 \xto{~\tuo~} &~\TF(Y+\TF(Y+X)).
\end{align*}}%
By definition, $t$ is strongly guarded, hence $\xi t^\istar=(\xi t)^\istar$.
\begin{align*}
\xi t =&\, [\eta,\xi]^\klstar\xi t\\
=&\, [\eta,\xi]^\klstar [\eta\inl,(S\inr)\upsilon]^\klstar\sigma\out t\\
=&\, [\eta,\xi^\klstar\upsilon]^\klstar\sigma\out t\\
=&\, [\eta,\xi^\klstar\upsilon]^\klstar S(\inl+\SigF(\eta^{\nu}\inr))[[\eta \inl,\sigma g],\eta \inr]^{\klstar}\sigma\out\\
=&\, [\eta\inl,\xi^\klstar S(\eta^\nu\inr)\upsilon]^\klstar [[\eta \inl,\sigma g],\eta \inr]^{\klstar}\sigma\out\\
=&\, [\eta\inl,(S\inr)\upsilon]^\klstar [[\eta \inl,\sigma g],\eta \inr]^{\klstar}\sigma\out\\
=&\, [[\eta\inl, [\eta\inl,(S\inr)\upsilon]^\klstar\sigma g],(S\inr)\upsilon]^\klstar\sigma\out\\
=&\, S[\id,\inr] w
\end{align*}
Using the identities derived above and the codiagonal law, we obtain
that
\begin{align*}
(\xi f)^\istar=w^{\istar\istar}\eta^\nu\inr=(S[\id,\inr] w)^\istar\eta^\nu\inr=(\xi t)^\istar\eta^\nu\inr=\xi t^\istar \eta^\nu\inr.
\end{align*}
We are left to show that $\xi t^\istar \eta^\nu\inr=\xi f^\istar$. We strengthen the latter to $t^\istar=[\eta^\nu,f^\istar]^\kklstar$, which would imply it as follows: $\xi t^\istar\eta^\nu\inr=\xi[\eta^\nu,f^\istar]^\kklstar\eta^\nu\inr=\xi f^\istar$. 

Since $t$ is guarded, we will be done once we show
that $[\eta^\nu,f^\istar]^\kklstar$ satisfies the fixpoint law for $t^\istar$. It is easy to verify that
$\out f^\istar = T(\id+\SigF[\eta^\nu,f^\istar]^\kklstar) g$. Then we have
\begin{align*}
\out [\eta^\nu, f^\istar]^\kklstar=&\,[\out[\eta^\nu,f^\istar],\eta\inr \SigF[\eta^\nu,f^\istar]^\kklstar]^\klstar\out\\
=&\,[[\eta\inl,\out f^\istar],\eta\inr \SigF[\eta^\nu,f^\istar]^\kklstar]^\klstar\out\\
=&\,[[\eta\inl,T(\id+\SigF[\eta^\nu,f^\istar]^\kklstar) g],\eta\inr \SigF[\eta^\nu,f^\istar]^\kklstar]^\klstar \out
\intertext{while, on the other hand,}
\out t^\istar=&\,\out [\eta^\nu,t^\istar]^\kklstar\tuo T(\inl+\SigF(\eta^{\nu}\inr))[[\eta \inl,g],\eta \inr]^{\klstar}\out\\
=&\,[\out[\eta^\nu,t^\istar],\eta\inr \SigF[\eta^\nu,t^\istar]^\kklstar]^\klstar T(\inl+\SigF(\eta^{\nu}\inr))[[\eta \inl,g],\eta \inr]^{\klstar}\out\\
=&\,[\out\eta^\nu,\eta\inr \SigF t^\istar]^\klstar [[\eta \inl,g],\eta \inr]^{\klstar}\out\\
=&\, [[\eta\inl,T(\id+\SigF t^\istar) g],\eta\inr \SigF t^\istar]^\klstar\out.
\end{align*}
Hence, indeed, $[\eta^\nu,f^\istar]^\kklstar=t^\istar$.

 \item\textit{Strongly guarded $f$.}
Finally, let us show~\eqref{eq:xi-morph} with strongly guarded $f$. Suppose that $h$ is such that $\out f = T(\inl+\SigF(\eta^\nu\inr)) h$. Recall that $\xi=([\eta\inl,(S\inr)\upsilon]^\klstar\sigma\out)^\istar$. By uniformity, it suffices to show that
\begin{align*}
[\eta\inl,(S\inr)\upsilon]^\klstar\sigma\out f^\istar = S(\id+f^\istar)\xi f.
\end{align*}
On the one hand,
\begin{align*}
[\eta\inl,(S\inr&)\upsilon]^\klstar\sigma\out f^\istar\\
=&\, [\eta\inl,(S\inr)\upsilon]^\klstar\sigma\out [\eta^\nu,f^\istar]^\kklstar f\\
=&\, [\eta\inl,(S\inr)\upsilon]^\klstar[\out[\eta^\nu,f^\istar],\eta\inr \SigF[\eta^\nu,f^\istar]^\kklstar]^\klstar \sigma\out f\\
=&\, [\eta\inl,(S\inr)\upsilon]^\klstar[\out[\eta^\nu,f^\istar],\eta\inr \SigF[\eta^\nu,f^\istar]^\kklstar]^\klstar S(\inl+\SigF(\eta^\nu\inr))\sigma h\\
=&\, [\eta\inl,(S\inr)\upsilon]^\klstar[\eta\inl,\eta\inr \SigF f^\istar]^\klstar \sigma h\\
=&\, [\eta\inl,S(\inr f^\istar)\upsilon]^\klstar\sigma h.
\intertext{and on the other hand,}
S(\id+f^\istar)\xi f
=&\, S(\id+f^\istar) [\eta,\xi]^\klstar [\eta\inl,(S\inr)\upsilon]^\klstar\sigma\out f\\
=&\, S(\id+f^\istar) [\eta,\xi]^\klstar [\eta\inl,(S\inr)\upsilon]^\klstar\sigma T(\inl+\SigF(\eta^\nu\inr)) h\\
=&\, S(\id+f^\istar) [\eta,\xi]^\klstar [\eta\inl,(S\inr)\upsilon]^\klstar S(\inl+\SigF(\eta^\nu\inr)) \sigma h\\
=&\, S(\id+f^\istar) [\eta,\xi]^\klstar [\eta\inl\inl,S(\inr\eta^\nu\inr)\upsilon]^\klstar  \sigma h\\
=&\, S(\id+f^\istar)  [\eta\inl,\xi^\klstar S(\eta^\nu\inr)\upsilon]^\klstar  \sigma h\\
=&\, S(\id+f^\istar)  [\eta\inl,(S\inr)\upsilon]^\klstar  \sigma h\\
=&\, [\eta\inl,S(\inr f^\istar)\upsilon]^\klstar  \sigma h.
\end{align*}
\end{citemize}
This finishes the proof that $\xi$ is a complete Elgot monad morphism.

\subsubsection*{Uniqueness}

Let $\rho:\BBTF\to\BBS$ be a complete Elgot monad morphism such that
$\sigma=\rho\ext$ and $\upsilon=\rho\tuo\,\eta\inr \SigF\eta^\nu$. We
have to show that
$\rho=\xi=([\eta\inl,(S\inr)\upsilon]^\klstar\sigma \out)^\istar$. We
rewrite the last term as follows:
\begin{align*}
&\bigl([\eta\inl,(S\inr)\upsilon]^\klstar\sigma \out\bigr)^\istar\\
=\,&([\eta\inl,(S\inr)\rho\tuo\eta\inr \SigF\eta^\nu]^\klstar\sigma \out)^\istar\\
=\,&([\eta\inl,\rho\, (\TF\inr)\tuo\eta\inr \SigF\eta^\nu]^\klstar\sigma \out)^\istar\\
=\,&([\eta\inl,\rho \tuo T(\inr+(\SigF\TF\inr))\eta\inr \SigF\eta^\nu]^\klstar\sigma \out)^\istar\\
=\,&([\eta\inl,\rho \tuo\eta\inr \SigF(\TF\inr)\SigF\eta^\nu]^\klstar\sigma \out)^\istar\\
=\,&([\eta\inl,\rho \tuo\eta\inr \SigF(\eta^\nu\inr)]^\klstar\sigma \out)^\istar\\
=\,&((\rho[\eta\inl,\tuo\eta\inr \SigF(\eta^\nu\inr)])^\klstar\rho\ext \out)^\istar\\
=\,&(\rho[\eta\inl,\tuo\eta\inr \SigF(\eta^\nu\inr)]^\kklstar\ext \out)^\istar\\
=\,&\rho([\eta\inl,\tuo\eta\inr \SigF(\eta^\nu\inr)]^\kklstar\ext \out)^\istar.
\end{align*}
To finish the calculation we have to verify that the term after $\rho$
vanishes. Note that the term under the iteration operator is
guarded. Hence, it suffices to show that $\id$ satisfies the
corresponding characteristic equation for iteration, i.e.\ that
\begin{align*}
[\eta,\id]^\kklstar [\eta\inl,\tuo\eta\inr \SigF(\eta^\nu\inr)]^\kklstar\ext \out = \id.
\end{align*}
We reduce the left hand side to $\id$ as follows:
\begin{flalign*}
\qquad&&&[\eta,\id]^\kklstar [\eta\inl, \tuo\eta\inr \SigF(\eta^\nu\inr)]^\kklstar\ext \out\\
&&=\,&[\eta,[\eta,\id]^\kklstar\tuo\eta\inr \SigF(\eta^\nu\inr)]^\kklstar\ext \out\\
&&=\,&[\eta,\tuo[\out[\eta,\id],\eta\inr \SigF{[\eta,\id]^\kklstar}]^\klstar\eta\inr \SigF(\eta^\nu\inr)]^\kklstar\ext \out\\
&&=\,&[\eta,\tuo\eta\inr]^\kklstar\ext \out\\
&&=\,&[\eta,\tuo\eta\inr]^\kklstar\tuo (T\inl)\out\\
&&=\,&\tuo[\out  [\eta,\tuo\eta\inr],\eta\inr \SigF[\eta,\tuo\eta\inr]^\kklstar]^\klstar
 (T\inl)\out
&&&&  \by{Theorem~\ref{lem:kl_dec}}\\
&&=\,&\tuo(\out[\eta,\tuo\eta\inr])^\klstar\out\\
&&=\,&\tuo[\out\eta,\eta\inr]^\klstar\out\\
&&=\,&\tuo[\eta\inl,\eta\inr]^\klstar\out\\
&&=\,&\tuo\out\\
&&=\,&\id.
\end{flalign*}
This finishes the proof.
\end{proof}

\noindent The existence and the exact shape of the initial complete
Elgot monad $\BBL$ mentioned in Corollary~\ref{cor:sum} depend on the
properties of the base category. We recall the key definition of a
hyperextensive category~\cite{AdamekBorgerEtAl08}:
\begin{defi}
  A category $\BC$ is \emph{hyperextensive} if
  \begin{enumerate}
  \item $\BC$ has countable coproducts that are \emph{disjoint}, i.e.\
  the pullback of any two distinct injections is an initial object,
  and \emph{universal}, i.e.\ stable under pullbacks; and
  \item in $\BC$, subobjects that are coproduct injections are closed
  under countable disjoint unions; that is, given countably many
  pairwise disjoint subobjects $A_i\to B$ that are coproduct
  injections, their union $\sum_i A_i\to B$ is again a coproduct
  injection.
  \end{enumerate}
\end{defi}
\noindent Examples of hyperextensive categories include $\Set$,
$\Cpo$, and bounded complete metric spaces as well as all presheaf
categories~\cite{AdamekBorgerEtAl08}. We refer to subobjects whose
inclusion morphisms are (binary) coproduct injections as
\emph{summands}, and given a summand, we refer to the partner
injection of the corresponding binary coproduct as its \emph{coproduct
  complement} (we will not need uniqueness of complements). In this
terminology, summands are closed under pullbacks (i.e.~under
preimages) and under countable disjoint unions in hyperextensive
categories. From countable disjoint unions we obtain unions of chains:
\begin{lem}\label{lem:hyperext-chains}
  Let\/ $\BC$ be hyperextensive. Then\/ $\BC$ has unions of\/
  $\omega$-chains of summands; such unions are again summands, and are
  \emph{universal}, i.e.~stable under pullbacks (and, hence, under
  products).
\end{lem}
\begin{proof}
  Any ascending chain of summands can be transformed into a disjoint
  union of summands: if $A_1$ and $A_2$ are summands of $X$ and $A_1$
  is contained in $A_2$, then by universality of coproducts, $A_1$ is
  also a summand of $A_2$ so we can replace $A_2$ with the coproduct
  complement of $A_1$ in $A_2$, preserving the union. Universality of
  unions of ascending chains of summands is then inherited from
  countable disjoint unions.
\end{proof}

\begin{thm}\label{thm:lifting-initial}
  Let $\BC$ be hyperextensive and have binary coproducts. Then the
  monad $\BBL$ given by $LX=X+1$ is $\omega$-continuous. Equipped with
  the arising complete Elgot monad structure according to
  Theorem~\ref{thm:co_elg}, $\BBL$ is the initial complete Elgot monad
  over $\BC$.
\end{thm}
\noindent (The conditions of the theorem imply our running assumption
that $\BC$ is distributive~\cite{CarboniLackEtAl93}.)
\begin{rem}\label{rem:iter-def}
  Let us spell out the definition of the iteration operator figuring
  in the statement of Theorem~\ref{thm:lifting-initial}
  explicitly. Suppose that $e:X\to L(Y+X)$.  Let $X_1$ be the preimage
  of\/ $Y$ under $e$ and $e_1:X_1\to Y$ the arising restriction
  of~$e$; for $i\geq 1$ let $X_{i+1}$ be the preimage of $X_i$ under
  $e$, and let $e_{i+1}:X_{i+1}\to X_i$ be the arising restriction
  of~$e$. By universality of finite coproducts, the $X_i$ are pairwise
  disjoint summands. By stability of summands under countable disjoint
  unions, $\sum_i X_i$ is a summand of~$X$, whose complement we
  denote~$X_{\infty}$.  We obtain the presentation
  $X=\sum_i X_i + X_{\infty}$. Now $e^\istar:X\to LY$ is the universal
  map induced by the $\eta\comp e_1\ldots e_i:X_i\to LY$ and
  $\bot: X_{\infty}\to LY$.
\end{rem}

\noindent Now $L$ clearly admits only very simple recursive
definitions: an equation morphism $e:X\to L(Y+X)=(Y+X)+1$ essentially
defines each variable in $X$ either as a result from $Y$ or as another
variable from $X$ or as divergence. In preparation of the proof of
Theorem~\ref{thm:lifting-initial}, the following lemma shows that the
solution of all possible such definitions of this shape is, in any
complete Elgot monad, uniquely determined by the complete Elgot monad laws.
\begin{lem}\label{lem:solutions-lifting}
  Let\/ $\BBT$ be a complete Elgot monad over a hyperextensive category $\BC$, let $e:X\to T(Y+X)$, and let
  $m:Z\to X$. Then the following holds.
  \begin{enumerate}
  \item If $e\comp m=\eta\comp\inl u$ for some $u:Z\to Y$ then
  $e^\istar\comp m=\eta\comp u$.
  \item If $e\comp m=\bot_{Z,Y+X}$ then $e^\istar\comp m = \bot_{Z,X}$.
  \item If $e\comp m=\eta\inr u$ for some $u:Z\to X$ then
  $e^\istar\comp m= e^\istar\comp u$.
  \item If $e\comp m=\eta\inr m\comp u$ for some $u:Z\to Z$ then
  $e^\istar\comp m=\bot_{Z,Y}$.
  \end{enumerate}
\end{lem}
\noindent That is: If a variable is defined as a result value, then
the solution of the recursive definition for that variable is that
result value; if a variable is defined as $\bot$, then the solution is
$\bot$; if a variable is defined as another variable, then its
solution is that of the other variable; and if a set of variables is
defined by mutual recursion without any base case and without use of
the algebraic operations of the monad, then the solution for all these
variables is $\bot$.
\begin{proof}
  The first three claims are immediate from the fixpoint law and
  coconstancy of~$\bot$ (Lemma~\ref{lem:divergence}). We show the last
  claim. We have
  \begin{displaymath}
  e\comp m=\eta\inr m\comp u =\eta\comp (\id+m)\comp \inr u =T(\id+m)\comp \eta\inr u ,
  \end{displaymath} 
  which by uniformity implies $e^\istar m=(\eta\inr u)^\istar$.
  The claim then follows by Lemma~\ref{lem:one-divergence}.
\end{proof} 
\begin{proof}[Proof of Theorem~\ref{thm:lifting-initial}]
  The base category $\BC$ is, a fortiori, extensive. In any extensive
  category, $\BBL$ is the partial map classifier for partial morphisms
  whose domains are summands; we will call such partial morphisms
  \emph{summand-partial}. Explicitly, a summand-partial morphism $f$
  from $X$ to $Y$ is thus a span
  $X\xleftarrow{~~m~~}D\xto{~~f~~}Y$ where $m$ is a summand;
  the \emph{domain} of $f$ is $m$ or, by abuse of notation, $D$. By
  \emph{preimages under $f$} we mean pullbacks along the map
  $f:D\to Y$ in this span.

  Thus, the Kleisli category of $\BBL$ inherits orderings on its
  hom-sets from the extension ordering on partial functions. The fact
  that $\BC$ has unions of $\omega$-chains of summands which are again
  summands (Lemma~\ref{lem:hyperext-chains}) then guarantees that
  these orderings are $\omega$-complete, and since $0$ is a summand,
  they have bottoms $\bullet\leftarrow 0 \to \bullet$. 
  We have to
  verify that Kleisli composition for $\BBL$ is continuous on both
  sides and that the remaining conditions of
  Definition~\ref{def:omega-cont} are satisfied. We will phrase all
  arguments in terms of summand-partial morphisms.

  \emph{Continuity of left Kleisli composition:} Let $g$ be a
  summand-partial morphism from $Y$ to $Z$, and let $(f_i)_{i\in\Nat}$
  be an ascending chain of summand-partial morphisms from $X$ to $Y$,
  with domains $D_i$. Denoting unions and joins of ascending chains by
  $\bigsqcup$ and composition of partial morphisms simply by
  juxtaposition, we have to show that
  $(\bigsqcup_i f_i)g=\bigsqcup_i f_ig$. The only problem here is to
  show that the domains of the two sides agree. The domain of $f_ig$
  is the preimage $E_i$ of $D_i$ under $g$; the domain of
  $\bigsqcup_i f_ig$ is the union $\bigsqcup_i E_i$ of the ascending
  chain $(E_i)_i$; the domain of $\bigsqcup_i f_i$ is the union
  $D=\bigsqcup_i D_i$; and the domain of $(\bigsqcup_i f_i)g$ is the
  preimage $E$ of $D$ under $g$. By universality of unions of
  ascending chains, $E=\bigsqcup_i E_i$.

  \emph{Continuity of right Kleisli composition:} Let $g$ be a
  summand-partial morphism from $X$ to $Y$ with domain $C$, and let
  $(f_i)_{i\in\Nat}$ be an ascending chain of summand-partial
  morphisms from $Y$ to $Z$, with domains $D_i$ and supremum $f$. We
  have to show $g\comp (\bigsqcup_i f_i)=\bigsqcup_i g\comp f_i$;
  again, we focus only on the domains. The domain of $gf_i$ is the
  preimage $E_i$ of $C$ under~$f_i$; the domain of $\bigsqcup_i f_ig$
  is the union $\bigsqcup_i E_i$; the domain of $f$ is the union
  $D=\bigsqcup_i D_i$; and the domain of $gf$ is the preimage $E$ of
  $C$ under $f$. By construction,~$E_i$ is contained in~$D_i$, and~$E$
  is contained in~$D$. Moreover, since~$f_i$ maps~$E_i$ into~$C$, so
  does~$f$, and hence~$E_i$ is also contained in $E$ (by the universal
  property of $E$ as a pullback). Denoting the restriction of
  $f:D\to Y$ to $E\to C$ by $f'$, we thus have the diagram
  \begin{equation*}
  \xymatrix@R25pt@C58pt@M6pt{E_i\ar@{>->}[d]\ar@{>->}[r] & E \ar[r]^{f'}\ar @{>->}[d] & 
  C \ar@{>->}[d] \\
  D_i \ar@{>->}[r] & \bigsqcup D_i \ar[r]_f & Y
  }
  \end{equation*}
  where the outer rectangle and the right hand square are pullbacks by
  construction. By the pullback lemma, it follows that the left hand
  square is also a pullback. By universality of unions of ascending
  chains of summands, it now follows that $E=\bigsqcup_i E_i$, as
  required.

  \emph{Continuity of the strength:} If the Kleisli morphism
  $f:Y\to Z+1$ corresponds to a summand-partial map with domain $D$,
  the Kleisli morphism $\tau(\id\times f):X\times Y\to Z+1$
  corresponds to a summand-partial map with domain $X\times D$.
  Continuity of $\tau(\id\times(-))$ is then immediate from stability
  of unions of ascending chains of summands under products.

  \emph{Continuity of copairing:} Immediate from the fact that
  generally, $\bigsqcup_i (D_i + E_i)=\bigsqcup_i E_i+\bigsqcup_i D_i$
  because unions of ascending chains of summands are defined via
  coproducts.

  \emph{Preservation of $\bot$ by left Kleisli composition:} The
  bottom element of the Kleisli hom-set from $X$ to $Y$ is the unique
  (summand-)partial morphism with domain $0$. Left Kleisli composites
  of this morphism have domains that are pullbacks of $0$, which in
  extensive categories are again $0$. 

  \emph{Preservation of $\bot$ by the strength:} The domain of the
  partial morphism corresponding to $\tau(\id\times\bot):X\to Y+1$ is
  $X\times 0$, which by extensivity (in fact already by
  distributivity) is $0$.

  This establishes that $\BBL$ is $\omega$-continuous, and hence a
  complete Elgot monad; by the standard construction of least
  fixpoints in $\omega$-cpos, the iteration operator of $\BBL$ then
  has the form described in Remark~\ref{rem:iter-def}. To see
  initiality of $\BBL$, let $\BBS$ be a complete Elgot monad
  on~$\BC$. For clarity, we denote the unit of~$\BBL$ by~$\eta^\BBL$
  and that of~$\BBS$ by $\eta^\BBS$. We need to show existence of a
  unique complete Elgot monad morphism $\xi:\BBL\to\BBS$. Since~$\xi$
  must preserve the unit and unproductive divergence $\bot$ (the
  latter by preservation of iteration), the only candidate is
  $\xi=[\eta^\BBS,\bot]$. It remains to show that $\xi$ is a complete
  Elgot monad morphism. Thanks to the simplicity of the monad
  structure of $\BBL$, it is clear that $\xi$ is a strong monad
  morphism. The main task is to prove preservation of iteration. So
  let $e:X\to L(Y+X)=(Y+X)+1$.
  We inductively construct infinite sequences $X_1,X_2,\ldots$ and
  $D_1,D_2,\ldots$ of summands of $X$ as follows: Like in Remark~\ref{rem:iter-def}, we take $X_1$ to be
  the preimage of $Y$ under $e$, and for $i>1$ we take $X_i$ to be the
  preimage of $X_{i-1}$ under $e$; similarly, we take $D_1$ to be the
  preimage of $1$ under $e$, and for $i>1$ we take $D_i$ to be the
  preimage of $D_{i-1}$ under $e$. By universality of coproducts, the
  $X_i$ and $D_i$ are pairwise disjoint summands (that is, the $X_i$
  are pairwise disjoint, the $D_i$ are pairwise disjoint, and every
  $X_i$ is disjoint with every $D_j$).  

  Let $X'=\sum_i X_i$, $D'=\sum_i D_i$ and let $Z$ be the complement
  of $X'+D'$ in $X$. For the remainder we regard~$X$ as being
  decomposed into the coproduct $X=X'+X_{\infty}$ where
  $X_\infty=D'+Z$.

By definition, there are $e_1:X_1\to Y$ and $e_i:X_i\to X_{i-1}$ ($i>1$) such that
\begin{align*}
e\comp\inl\inj_1 = \eta^\BBL\inl\comp e_1&&
e\comp\inl\inj_{i} = \eta^\BBL\inr\comp\comp\inl\inj_{i-1} e_{i}
\end{align*}
where $\inj_i$ denotes the $i$-th coproduct injection into a countable coproduct. 

By applying the fixpoint law $i$ times we obtain
$(\xi e)^\istar\comp\inl\inj_i=\eta^\BBS\comp e_1\ldots e_i$ and
analogously,
$\xi\comp e^\istar\comp\inl\inj_i =\xi\comp\eta^\BBL\comp e_1\ldots
e_i=\eta^\BBS\comp e_1\ldots e_i$.
We are left to show that
$(\xi e)^\istar\comp\inr=\xi\comp e^\istar\comp\inr$. Noting that by
definition, $e^\istar\comp\inr=\bot$ and $\xi$ preserves $\bot$, this
amounts to showing that $(\xi e)^\istar\comp\inr=\bot$.
By construction of the~$D_i$, for every $i>1$ there is $d_i:D_i\to D_{i-1}$ such that
\begin{align*}
e\comp\inr\inl\inj_1 = \bot&&
e\comp\inr\inl\inj_{i} = \eta^\BBL\inr\comp\inr\comp\inl\inj_{i-1} d_{i},
\end{align*}
hence, by applying the fixpoint law $i$ times we obtain that $(\xi e)^\istar\comp\inr\inl\inj_i=\bot$, which
implies $(\xi e)^\istar\comp\inr\inl=\bot$ and hence we are left to show $(\xi e)^\istar\comp\inr\inr=\bot$.

Notice that the preimages of $Y$ and $1$ under $e\comp \inr\inr$ 
must be $0$ and therefore there is $m_1:Z\to X$
such that $e\comp\inr\inr = \eta^\BBL\inr\comp m_1$. 
Analogously for every $i>1$ we construct $m_{i}:Z\to X$ such that $e\comp m_{i-1} = \eta^\BBL\inr\comp m_{i}$.
Let us denote by $\widehat Z$ the sum of $\omega$ copies of $Z$ and by $\hat m:\widehat Z\to X$
the cotuple formed by the morphisms $m_i$ with $i>0$. Now, 
\begin{align*}
\xi\comp e\comp[\inr\inr,\hat m] 
& = \xi\comp\eta^\BBL\inr\comp\hat m\comp w\\
& = \eta^\BBS\inr[\inr\inr,\hat m]\comp\inr\comp w\\
& = \eta^\BBS(\id+[\inr\inr,\hat m])\comp\inr\comp\inr\comp w\\
& = S(\id+[\inr\inr,\hat m])\comp \eta^\BBS\inr\comp \inr\comp w
\end{align*}
where $w:Z+\widehat Z\to\widehat Z$ is the obvious canonical isomorphism.
By uniformity, and by Lemma~\ref{lem:one-divergence}, this implies $(\xi\comp e)^\istar[\inr\inr,\hat m] = (\eta^\BBS\inr\inr\comp w)^\istar=\bot$ and therefore
$(\xi\comp e)^\istar\inr\inr=(\xi\comp e)^\istar[\inr\inr,\hat m]\inl=\bot$ as required.
\end{proof}
\begin{rem}
  The above proof of Theorem~\ref{thm:lifting-initial} uses the full
  power of the definition of hyperextensive categories, including
  universality of countable coproducts. It has been shown
  previously~\cite{BhaduriSubramanian97} that assuming only
  universality of finite coproducts and stability of summands under
  countable disjoint unions, one can still define the iteration
  operator and prove the fixpoint law. However, we do not see how
  to show the uniformity law in this weaker setting. At the same
  time, we have the impression that the uniformity law is the only
  place where universality of countable coproducts is needed.
\end{rem}

\section{Related Work}\label{sec:related}

\noindent The above results benefit from
extensive previous work on monad-based axiomatic iteration.
In particular we draw on the concept of \emph{complete Elgot monad}
studied by Ad\'amek et al.~\cite{AdamekMiliusEtAl10}; the construction
of the free complete Elgot monad over a
functor~\cite{AdamekMiliusEtAl11} is strongly related to
Corollary~\ref{cor:sum}.(\ref{item:free}), and we do not claim
Part~(\ref{item:free}) of Corollary~\ref{cor:sum} as a contribution of
this paper. %
There is extensive literature on solutions of (co)recursive program
schemes~\cite{Bartels03,AczelAdamekEtAl03,MiliusMossEtAl13,GoncharovSchroder13,PirogGibbons13,PirogGibbons14},
from which our present work differs primarily in that we do not
restrict to guarded systems of equations. In particular, as mentioned
in the introduction, Piróg and Gibbons~\cite{PirogGibbons13} actually
work with the same monad transformer, the coinductive generalized
resumption transformer.  The same
authors~\cite[Corollary~4.6]{PirogGibbons14} prove a coproduct
characterization of the coinductive generalized resumption transformer
that is similar to our Theorem~\ref{thm:resump-universal}; but again,
this takes place in a different category, that is, in completely
iterative monads (admitting guarded recursive definitions) rather than
complete Elgot monads (admitting unrestricted recursive
definitions). Technically, results on $\TF$ being a completely
iterative monad are incomparable to our result on $\TF$ being a
complete Elgot monad -- we prove a stronger recursion scheme for~$\TF$
but need to assume that $T$ is a complete Elgot monad, while $\TF$ is
completely iterative without any assumptions on $T$.
 
Moss~\cite{Moss03} proves that given a $\Set$-endofunctor $F$ and a
distinguished point $\bot:1\to\nu F$ of the final $F$-coalgebra, the
monad $M$ given on objects by
$M_FX=X+\nu\gamma.\,F(X+\gamma)\cong \nu\gamma.\,X +F\gamma$ is completely
Elgot, with unproductive divergence induced by $\bot$ (Moss in fact
establishes a completeness result over such monads). This result does
not appear to be an immediate application of our
Theorem~\ref{thm:cr_elg}, as there is no implicit complete Elgot base
monad in~$M_F$. %

We construct solutions of unguarded recursive equations from solutions
of guarded recursive equations, for the latter relying crucially on
results by Uustalu on guarded recursion over parametrized
monads~\cite{Uustalu03}, which in particular has allowed us to make do
without idealized monads.

The axiomatic treatment of iteration via complete Elgot monads is
essentially dual to the axiomatic treatment of recursion by Simpson
and Plotkin~\cite{SimpsonPlotkin00}, who work in a category~$\BD$ with
a \emph{parametrized uniform recursion operator}
$\Hom_{\BD}(Y\times X,X)\to \Hom_{\BD}(Y,X)$ and a subcategory~$\BS$
of \emph{strict} functions in~$\BD$. Given a distributive
category~$\BC$ equipped with a complete Elgot monad, we can take
$\BS=\BC^{op}$ and $\BD=(\BC_{\BBT})^{op}$. Then the iteration
operator over $\BC_{\BBT}$ sending $f:X\to T(Y+X)$ to
$f^\istar:X\to TY$ induces precisely a parametrized uniform recursion
operator for the pair $(\BD,\BS)$ in the sense of Simpson and Plotkin.

The proof of Theorem~\ref{thm:cr_elg} can be embedded into a generic
framework connecting guarded and unguarded iteration that we have
developed in further work~\cite{GoncharovSchroderEtAl17}.

\section{Conclusions and Future Work}\label{sec:concl}

\noindent We have developed semantic foundations for non-wellfounded
side-effecting recursive definitions, specifically for recursive
definitions over the so-called \emph{coinductive generalized
  resumption transformer} that extends a base monad $\BBT$ with
operations represented by a functor $\Sigma$ to obtain a monad $\BBTF$
defined by taking final coalgebras, i.e.~consisting of non-wellfounded
trees. While previous work on the same monad transformer was focussed
on guarded recursive definitions, in the framework of completely
iterative monads, we work in the setting of (complete) Elgot monads,
which admit unrestricted recursive definitions. Our main results state
that
\begin{itemize}
\item $\BBTF$ is a complete Elgot monad if $\BBT$ is a complete Elgot
  monad (Theorem~\ref{thm:cr_elg});
\item the structure of $\BBTF$ as a complete Elgot monad is uniquely
  determined as extending that of $\BBT$ (Theorem~\ref{thm:cr_elg});
\item if the underlying category $\BC$ admits an initial complete
  Elgot monad $\BBL$ (often $L=(\argument)+1$), then
  $\BBTF\cong \BBT+\BBL_\SigF$ in the category of complete Elgot
  monads on $\BC$
  (Theorem~\ref{thm:resump-universal}/Corollary~\ref{cor:sum}).
\end{itemize}
In particular this requires proving the equational laws of complete
Elgot monads for the solution operator that we construct on $\TF$. We
have implemented a formal verification of our results, which are
technically quite involved, in the Coq proof assistant, see
\url{https://git8.cs.fau.de/redmine/projects/corque}.

Besides the fact that applying the coinductive resumption monad
transformer to a complete Elgot monad $\BBT$ again yields a complete
Elgot monad $\BBTF$, the resulting object obviously has a richer
structure provided by the adjoined free operations. One topic for
further investigation is to identify (and possibly axiomatize) this
structure. We aim to use this structure to program definitions of free
operations as morphisms $\TF X\to TX$ in a similar spirit as in the
paradigm of \emph{handling algebraic effects}~\cite{PlotkinPretnar13}.
In conjunction with iteration this actually produces a recursion
operator that is more expressive than iteration. This however requires
going beyond the first-order setting of this paper (which was
sufficient for iteration), as call-by-value recursion is known to be
an inherently higher-order concept. There is an concept of complete \emph{Elgot
  algebra}~\cite{AdamekMiliusEtAl06} complementing complete Elgot monads. It has been
shown that the algebras of complete Elgot monads are complete Elgot algebras
satisfying additional conditions~\cite{GoncharovMiliusEtAl16}; the precise
relationship between complete Elgot monads and complete Elgot algebras remains to be
determined, possibly using our results on iteration-congruent
retracts of monads with iteration~\cite{GoncharovSchroderEtAl17}.

\bigskip\noindent \textbf{Acknowledgements}
The authors wish to thank Stefan Milius and Paul Blain Levy for useful
discussions.

\bibliographystyle{myabbrv}
\bibliography{monads}

\clearpage

\end{document}

